\documentclass[reqno,a4paper,11pt]{article}
\pdfoutput=1

\usepackage{xcolor}
\usepackage{graphicx}
\usepackage[textwidth = 430 pt, textheight = 630 pt]{geometry}

\usepackage{epsfig,rotating}
\usepackage{amsmath,amssymb}
\usepackage{amsfonts}
\usepackage{mathrsfs}
\usepackage{bbm}
\usepackage[normalem]{ulem}

\usepackage{latexsym}
\usepackage{amsthm}
\usepackage{thmtools}
\usepackage[all,knot]{xy}
\xyoption{arc}

\usepackage[utf8x]{inputenc}
\usepackage{orcidlink}

\definecolor{MyDarkBlue}{rgb}{0.15,0.25,0.45}

\usepackage{hyperref}
\hypersetup{
    hypertexnames=false,
    colorlinks=true,
    citecolor=MyDarkBlue,
    linkcolor=MyDarkBlue,
    urlcolor=MyDarkBlue,
    pdfauthor={Hyungrok Kim,Christian S\"amann},
    pdftitle={},
    pdfsubject={hep-th math-ph},
    breaklinks=true
}
\usepackage[nameinlink,noabbrev]{cleveref}

\usepackage{tikz}
\usetikzlibrary{matrix,cd,arrows}
\usepackage{mathtools}
\usepackage[all,knot]{xy}
\xyoption{arc}

\usepackage[toc,page]{appendix}

\usepackage[vcentermath]{youngtab}

\let\fn\footnote
\renewcommand{\footnote}[1]{\linespread{1.1}\fn{#1}\linespread{1.29}}

\declaretheorem[numberwithin=section,name=Theorem,refname={theorem,theorems},Refname={Theorem,Theorems}]{theorem}
\declaretheorem[sibling=theorem,name=Lemma,refname={lemma,lemmas},Refname={Lemma,Lemmas}]{lemma}
\declaretheorem[sibling=theorem,name=Definition,refname={definition,definitions},Refname={Definition,Definitions}]{definition}
\declaretheorem[sibling=theorem,name=Proposition,refname={proposition,propositions},Refname={Proposition,Propositions}]{proposition}
\declaretheorem[sibling=theorem,name=Corollary,refname={corollary,corollaries},Refname={Corollary,Corollaries}]{corollary}
\declaretheorem[sibling=theorem,name=Remark,refname={remark,remarks},Refname={Remark,Remarks}]{remark}

\declaretheorem[numberwithin=subsection,name=Theorem,refname={theorem,theorems},Refname={Theorem,Theorems}]{atheorem}
\declaretheorem[sibling=atheorem,name=Definition,refname={definition,definitions},Refname={Definition,Definitions}]{adefinition}

\makeatletter
\renewcommand{\section}{\@startsection
    {section}{1}{\z@}{-3.5ex plus -1ex minus
        -.2ex}{2.3ex plus .2ex}{\mathversion{bold}\bf }}
\renewcommand{\subsection}{\@startsection{subsection}{2}{\z@}{-3.25ex
        plus -1ex minus
        -.2ex}{1.5ex plus .2ex}{\mathversion{bold}\bf }}
\renewcommand{\subsubsection}{\@startsection{subsubsection}{3}{-2.45ex}{-3.25ex
        plus -1ex minus -.2ex}{1.5ex plus .2ex}{\it }}
\renewcommand\paragraph{\@startsection{paragraph}{4}{\z@}%
    {3.25ex \@plus1ex \@minus.2ex}%
    {-1em}%
    {\normalfont\normalsize\bfseries\mathversion{bold}}}

\renewcommand{\thesection}{\arabic{section}}
\renewcommand{\thesubsection}{\arabic{section}.\arabic{subsection}}
\renewcommand{\@seccntformat}[1]{\@nameuse{the#1}.~~}

\renewcommand{\theequation}{\thesection.\arabic{equation}}
\makeatletter \@addtoreset{equation}{section}
\def\Ddots{\mathinner{\mkern1mu\raise\p@
        \vbox{\kern7\p@\hbox{.}}\mkern2mu
        \raise4\p@\hbox{.}\mkern2mu\raise7\p@\hbox{.}\mkern1mu}}
\renewcommand{\appendices}{
    \section*{Appendix}\label{appendices}\setcounter{subsection}{0}
    \addcontentsline{toc}{section}{Appendix}
    \setcounter{equation}{0}
    \renewcommand{\theequation}{\Alph{subsection}.\arabic{equation}}
    \renewcommand{\thesubsection}{\Alph{subsection}}
    \@addtoreset{equation}{subsection}
    \crefalias{subsection}{appendix}
}

\makeatother

\newcommand{\makecommand}[3]{%
    \foreach \i in #3 {%
        \expandafter\xdef\csname #1\i\endcsname{\noexpand#2{\unexpanded\expandafter{\i}}}%
    }%
}

\makecommand{fr}{\mathfrak}{{c,der,g,h,A,X,gl,sl,o,so,osp,u,su,sp,spin,string,Poisson,td,tb}}

\makecommand{sf}{\mathsf}{{id,String,Lie,Hom,SU,SO,Cat,Set,Spin,Pin,inv,CE,Diff,HSymp,TD,TB,pr,Aut,GL,SL,IO,Mat,Sym,Inn,Out,AUT,hom}}

\makecommand{rm}{\mathrm}{{Ham,diag,im,deg,id,Ad}}

\newcommand{\latinalphabet}{A,a,B,b,C,c,d,D,E,e,F,f,G,g,H,h,I,i,J,j,K,k,L,l,M,m,N,n,O,o,P,p,Q,q,R,r,S,s,T,t,U,u,V,v,W,w,X,x,Y,y,Z,z}
\newcommand{\caplatinalphabet}{A,B,C,D,E,F,G,H,I,J,K,L,M,N,O,P,Q,R,S,T,U,V,W,X,Y,Z}
\makecommand{I}{\mathbbm}{\latinalphabet}
\makecommand{bf}{\mathbf}{\latinalphabet}
\makecommand{bm}{\bm}{\latinalphabet}
\makecommand{ca}{\mathcal}{\caplatinalphabet}
\makecommand{fr}{\mathfrak}{\latinalphabet}
\makecommand{rm}{\mathrm}{\latinalphabet}
\makecommand{sf}{\mathsf}{\latinalphabet}
\makecommand{sc}{\mathscr}{\latinalphabet}

\newcommand{\ihom}{\underline{\sfhom}}

\newcommand{\acton}{\triangleright}

\def\slasha#1{\setbox0=\hbox{$#1$}#1\hskip-\wd0\hbox to\wd0{\hss\sl/\/\hss}}

\def\periodb#1{\setbox0=\hbox{$#1$}#1\hskip-\wd0\hbox to\wd0{-}}

\newcommand{\unit}{\mathbbm{1}}   			

\newcommand{\comment}[1]{}     				
\def\tyng(#1){\raisebox{0.05cm}{\hbox{\tiny$\yng(#1)$}}}			
\def\tyoung(#1){\hbox{\tiny$\young(#1)$}}			

\newcommand{\eand}{~~~\mbox{and}~~~}

\newcommand{\beq}{\begin{eqnarray}}
    \newcommand{\eeq}{\end{eqnarray}}

\newcommand{\eps}{\varepsilon}

\let\oldref\ref
\AtBeginDocument{\renewcommand{\ref}[1]{\cref{#1}}}
\AtBeginDocument{\renewcommand{\eqref}[1]{{\rm (\oldref{#1})}}}
\tikzset{Rightarrow/.style={double equal sign distance,>={Implies},->},
    triple/.style={-,preaction={draw,Rightarrow}},
    quadruple/.style={preaction={draw,Rightarrow,shorten >=0pt},shorten >=1pt,-,double,double
        distance=0.2pt}}

\usepackage{stmaryrd}

\begin{document}
    \begin{titlepage}
        \setcounter{page}{0}
        \renewcommand{\thefootnote}{\fnsymbol{footnote}}
        \vspace*{2.0cm}
        \begin{center}
            {\LARGE \bf
                Principal 3-Bundles with Adjusted Connections
            }
            \vskip1.5cm
            {\large Gianni Gagliardo\,\orcidlink{0009-0005-5724-7965}\,$^{a}$, Christian Saemann\,\orcidlink{0000-0002-5273-3359}\,$^{a}$, and Roberto Tellez-Dominguez\,\orcidlink{0000-0002-8348-253X}\,$^{b}$}~\footnote{{\it E-mail addresses:\/}
                \href{mailto:ggg2001@hw.ac.uk}{\ttfamily ggg2001@hw.ac.uk},
                \href{mailto:c.saemann@hw.ac.uk}{\ttfamily c.saemann@hw.ac.uk}, 
                \href{mailto:roberto.tellezdominguez@ceu.es}{\ttfamily roberto.tellezdominguez@ceu.es}
            }
            
            \setcounter{footnote}{0}
            \renewcommand{\thefootnote}{\arabic{thefootnote}}
            \vskip1cm
            $^a$ 
            Maxwell Institute for Mathematical Sciences\\
            Department of Mathematics,
            Heriot--Watt University\\
            Edinburgh EH14 4AS, United Kingdom\\[.3cm]
            
            $^b$
            Instituto de Ciencias Matem\'aticas\\
            C. Nicol\'as Cabrera, 13-15, 28049 Madrid, Spain
        \end{center}
        \vskip1.0cm
        \begin{center}
            {\bf Abstract}
        \end{center}
        \begin{quote}
            We explore the notion of an adjusted connection for principal 3-bundles. We first derive the explicit form of an adjustment datum for 3-term $L_\infty$-algebras, which allows us to give a local description of such adjusted connections and their infinitesimal symmetries. We then integrate the corresponding action Lie 3-algebroid to an action Lie 3-groupoid, encoding local connection forms with finite (higher) symmetries. This also yields the notion of an adjusted 2-crossed module of Lie groups. Stackifying the action Lie 3-groupoid then gives us the explicit description of principal 3-bundles with adjusted connections in terms of differential cohomology. These connections appear in a number of contexts within high-energy physics, and we list local examples arising in gauged supergravity as well as a global example arising in various contexts in string/M-theory. Our primary motivation, however, stems from U-duality, and we also define a notion of categorified torus that forms an adjusted 2-crossed module, which we hope to be useful in lifting T-duality to M-theory.
        \end{quote}

    \end{titlepage}
    
    \tableofcontents
    
    \section{Introduction and results}
    
    \paragraph{Motivation.} The interplay between differential geometry and physics, particularly high-energy physics and within this area, especially string/M-theory, has been an extraordinarily productive one, both on the mathematical and the physical side. It is therefore natural to study further predictions and requirements of the framework known as string theory and to explore their differential geometric consequences. In this process, one is soon led to higher generalizations of parallel transport underlying so-called higher gauge theories~\cite{Baez:2010ya,Borsten:2024gox}. These are physical field theories whose kinematic data is given by connections on higher or categorified principal bundles, which are locally described by higher degree differential forms. In the literature, these higher bundles are also known as ($n-$)gerbes~\cite{Giraud:1971,0817647309}.
    
    In the past, higher principal bundles have been defined in many different ways, see e.g.~\cite{Nikolaus:2011ag} for a review of some of them. Their topological definition is usually rather straightforward, following evident principles of categorification. Endowing these higher bundles with connections, however, is more subtle. In particular, the natural higher generalizations lead to connections that are either slightly too general~\cite{Breen:math0106083,Aschieri:2003mw} or too restrictive, see~\cite{Borsten:2024gox} for a detailed summary of the situation. 
    
    Physics fortunately provides guidance on how to fix this issue, e.g.~in the so-called tensor hierarchies of gauged supergravity. Essentially, the unwanted freedom in the formulation of connections from~\cite{Breen:math0106083,Aschieri:2003mw} can be fixed by providing an additional datum on the higher gauge group. For particular higher gauge Lie algebras and in the local case, this datum was called a Chern--Simons term~\cite{Sati:2008eg}. A more general notion called an adjustment was subsequently introduced in~\cite{Saemann:2019dsl}, but the exploration of this mathematical structure is still ongoing.
    
    In two classes of examples, there is an interpretation of the adjustment data. First, there is a class of higher gauge Lie algebras modeled by homotopy Lie algebras, which arise from differential graded Lie algebras in a particular shift-truncation procedure. For these, the (local and infinitesimal) adjustment datum is given by an alternator in a homotopy Lie algebra, which describes the lift of anti-symmetry of the Lie bracket up to homotopy~\cite{Borsten:2021ljb}. Second, adjustments are also required for Lie groupoid bundles, and here they were identified with Cartan connections on the structure Lie groupoid~\cite{Fischer:2024vak}.
    
    An explicit global description of adjusted connections on general principal 2-bundles in the form of differential cocycles for adjusted connections was provided in~\cite{Rist:2022hci}. In the case that a Lie crossed module also admits a model as a multiplicative gerbe, adjustments on the crossed module were proven to be equivalent to connections on the multiplicative gerbe in~\cite{Tellez-Dominguez:2023wwr}, leading to a unification of the notions of adjusted connections and trivializations of the Chern--Simons 2-gerbe. 
    
    Adjusted connections on higher principal bundles have already found a number of applications in theoretical physics. The above-mentioned class of examples where the adjustment is given by an alternator arises in the tensor hierarchies of gauged supergravities. The other class of examples on Lie groupoid bundles allows for a vast generalization of gauge-matter theories ubiquitous in elementary particle physics. Within string theory, or rather its low-energy limit of supergravity, adjusted connections on principal 2-bundles describe the gauge potentials in the Neveu--Schwarz sector. One concrete application of this perspective is a differentially refined description of the topological T-duality of~\cite{Bouwknegt:2003vb} in terms of principal 2-bundles with adjusted connections\footnote{See~\cite{Nikolaus:2018qop} for the purely topological case.}~\cite{Kim:2022opr,Kim:2023hqx}, reproducing and extending a number of results known from the literature.
    
    String theory itself is believed to be a limit of M-theory, a theory on an eleven-dimensional space-time, with 11d supergravity the low-energy limit. This theory naturally comes with a gauge potential 3-form, which is part of a connection on a principal 3-bundle. Moreover, the T-duality of string theory lifts to a much richer and currently intensely studied symmetry called U-duality. It is therefore natural to develop the explicit formulation of adjusted connections on principal 3-bundle, and this is the goal of this paper.
    
    \paragraph{Key results.} We start from a local description of adjusted connections given by 3-term $L_\infty$-algebra-valued differential forms. Such connections are morphisms of differential graded commutative algebras from the Weil algebra of an $L_\infty$-algebra to the differential forms on a local patch of a manifold. Considering a particular truncation of the corresponding inner homomorphisms, one obtains the so-called BRST complex of the connection, which describes the local action of infinitesimal bundle isomorphisms, or gauge transformations in physics parlance, as well as their higher analogs. 
    
    Consistency of this BRST complex requires a particular choice of generators of the Weil algebra, and the Weil algebra automorphism bringing the natural choice of generators as in~\cite{Sati:2008eg} to a consistent choice is called a (local, infinitesimal) adjustment. We identify the relevant moduli of this automorphism, providing explicit descriptions for adjustments on 3-term $L_\infty$-algebras in \ref{sec:local_connections}. 
    
    The resulting local description of connection forms, together with their (higher) infinitesimal gauge transformations form a strict action Lie algebroid. This BRST Lie 3-algebroid serves as a starting point for the discussion of the finite symmetries: it identifies the relevant curvature forms, and provides a consistency check for the gauge transformations.
    
    To keep computations manageable, we then restrict ourselves to 2-crossed modules of Lie algebras, integrating to 2-crossed modules of Lie groups. These describe a semistrict 3-group or Gray group. Equivalently, they describe simplicial Lie groups with a Moore complex with non-trivial components in degrees $-2$, $-1$ and $0$. The corresponding 2-crossed modules of Lie algebras have an overlap with 3-term $L_\infty$-algebras, but neither is a subset of the other; see \ref{app:Lie-3-algebras}.
    
    In a first step, we integrate the (strict) BRST Lie 3-algebroid for a 2-crossed module of Lie algebras to a strict BRST Lie 3-groupoid in \ref{sec:BRST-groupoid}. This BRST Lie 3-groupoid captures local connection forms as well as the action of (finite) gauge and higher gauge transformations and their composition. The underlying calculations are made non-trivial by the involved axioms of the structure maps in 2-crossed modules, and they are very lengthy and error-prone. 
    
    Computing the integration amounts to postulating gauge and higher gauge transformations. The unadjusted gauge transformations, i.e., those that are consistent for fake flat connections\footnote{That is, connections for which all curvature forms but the one of highest form degree vanish.}, are easily guessed. For general connections, however, we have to introduce deformations that ensure that all (higher) morphisms in the BRST Lie 3-groupoid act consistently. In particular, they must compose associatively and link lower morphisms with the same target. We show that these deformations are captured by an adjustment, additional structure maps on the 2-crossed module of Lie groups that must satisfy certain adjustment relations. This yields the definition of an adjusted 2-crossed module of Lie groups in~\ref{def:adjusted-2XM}.
    
    Once the BRST Lie 3-groupoid is established, we can stackify it to the differential cohomology describing principal 3-bundles with adjusted connections subordinate to a chosen cover of the base manifold. This is done in \ref{sec:diff_cohomology}. Concretely, we link local connection forms over double overlaps by gauge transformations, which are then glued together over triple overlaps using higher gauge transformations, etc. We note that so far, only partial information on the form of even the unadjusted differential cocycles was available in the literature~\cite{Saemann:2013pca}, see also~\cite{Jurco:2009px} and~\cite{Martins:2009aa} for Čech cocycles and a discussion of holonomies, as well as~\cite{Wang:2013dwa,Song:2021jnw} for related work. Particularly the higher gauge transformations were mostly missing. Our results fill this gap for fake-flat connections and extend it to the general non-fake-flat case.
    
    \paragraph{Applications.} Our key results are obtained after a computational tour de force, and to justify this effort, we give a number of applications in \ref{sec:examples}. We start with the infinitesimal case and the tensor hierarchy of $d=4$ gauged supergravity, showing that the previous results of~\cite{Borsten:2021ljb} fit our perspective here and that they provide an example of adjustments for local 3-term $L_\infty$-valued connections.
    
    We then turn to twisted differential string structures as introduced in a somewhat different form in~\cite{Sati:2009ic} as a first finite example. In this example, only few of the adjustment moduli are non-trivial, but we include it due to its importance in string theory. 
    
    Finally, we come to the example that is inspired by our original motivation for this research: the lift of T-duality to M-theory. To this end, we define a 3-group version of the categorified tori introduced in~\cite{Ganter:2014zoa}. Recall that the latter were instrumental for interpreting topological T-duality in terms of principal 2-bundles in~\cite{Nikolaus:2018qop,Kim:2022opr,Kim:2023hqx}, where they played the role of the relevant structure 2-group. Moreover, their algebraic data contained a natural adjustment, which allowed for the differential refinement of the principal 2-bundles describing the topological T-duality. 
    
    We show that our 3-group categorified tori form 2-crossed modules, and the underlying algebraic data again contains a natural adjustment. The application to T-duality-like constructions is the main topic of future research~\cite{Gagliardo:2025b}.
    
    \section{Local connections on principal 3-bundles}\label{sec:local_connections}
    
    It is useful to start with local connections and their infinitesimal gauge symmetries: the computations here are substantially simpler than in the global case with finite gauge symmetries, and the local results serve as a reference and as a source of cross-checks for the global computations.
    
    The construction of global $L_\infty$-algebra-valued connection forms was first discussed in~\cite{Sati:2008eg}, with related previous work~\cite{Cartan:1949aaa,Cartan:1949aab,Bojowald:0406445} and in particular~\cite{Kotov:2007nr}. In this paper, we will use the local construction of adjusted connection forms on a contractible patch $U$ of a manifold $M$, as explained in~\cite{Saemann:2019dsl}, see also~\cite{Fischer:2024vak} for the case of Lie algebroid-valued connections and~\cite{Borsten:2024gox} for a review. 
    
    \subsection{Local \texorpdfstring{$L_\infty$}{L infinity}-algebra-valued connection forms}
    
    In the following, we model Lie 3-algebras by 3-term $L_\infty$-algebras, see \ref{app:Lie-3-algebras,app:L-infty} for details. We use $L_\infty$-algebra-valued local connection forms, as introduced in~\cite{Sati:2008eg}, and we briefly review the construction. Recall that any $L_\infty$-algebra $\frL$ dually defines a semi-free\footnote{i.e., there are relations involving the differential} differential graded commutative algebra (dgca), its Chevalley--Eilenberg algebra $\sfCE(\frL)$, given by the free symmetric tensor algebra\footnote{Here and in the following, the notation $V[i]$ denotes the graded vector space $V=\oplus_{k\in \IZ} V_k$, grade-shifted such that the homogeneously graded degree-$j$ elements $(V[i])_j$ are given by $V_{i+j}$.}\textsuperscript{,}\footnote{We also assume that $\frL[1]$ allows for a degree-wise dualization to a vector space $\frL[1]^*$, e.g.~because the homogeneously graded components are finite-dimensional.}
    \begin{equation}
        \sfCE(\frL)=(\odot^\bullet (\frL[1]^*),\sfd_{\sfCE})~.
    \end{equation} 
    The differential $\sfd_{\sfCE}$ is fully defined by its action on $t\in\frL[1]^*$, 
    \begin{equation}
        \sfd_{\sfCE} t=-f_1(t)-\tfrac12 f_2(t,t)-\tfrac{1}{3!}f_3(t,t,t)+\ldots~,
    \end{equation} 
    and the expression $f_i(t,\ldots,t)$ are, up to signs, given by the higher brackets $\mu_i$. For example, in the special case of a differential graded Lie algebra $\frL$ with basis $(e_\alpha)$ and structure constants defined by\footnote{We use Einstein summation convention, i.e., every index appearing twice is summed over.}
    \begin{equation}
        \mu_1(e_\alpha)\eqqcolon f_{1\alpha}^\beta e_\beta
        \eand
        \mu_1(e_\alpha,e_\beta)\eqqcolon f_{2\alpha\beta}^\gamma e_\gamma~,
    \end{equation}
    the Chevalley--Eilenberg algebra is freely generated by generators $t^\alpha$ shifted-dual to the basis elements $e_\alpha$ and of degree~$1-|e_\alpha|$, and the differential acts as
    \begin{equation}
        \sfd_{\sfCE} t^\alpha=-f_{1\beta}^\alpha t^\beta-\tfrac12 f_{2\beta\gamma}^\alpha t^\beta t^\gamma~.
    \end{equation} 
    
    Given an $L_\infty$-algebra $\frL$ concentrated in non-positive degrees, a flat local $\frL$-valued connection on a contractible manifold $U$ is given by a dgca-morphism
    \begin{equation}\label{eq:conn_morph}
        \caA:\sfCE(\frL)\rightarrow (\Omega^\bullet(U),\rmd)~.
    \end{equation} 
    The map $\caA$ is completely determined by the image of the generators, and in the case of the differential graded Lie algebra, we have local connection forms
    \begin{equation}
        A^\alpha\coloneqq \caA(t^\alpha)\in \Omega^{|t^\alpha|}(U)~.
    \end{equation} 
    Compatibility with the differential implies that the curvature
    \begin{equation}
        F^\alpha\coloneqq \rmd A^\alpha + \tfrac12 f^\alpha_{2\beta\gamma} A^\beta A^\gamma+f^\alpha_{1\beta} A^\beta
    \end{equation} 
    vanishes.    
    
    To remove the flatness constraint, we can switch to the Weil algebra of $\frL$ as in~\cite{Sati:2008eg}. This differential graded commutative algebra is a doubling of the Chevalley--Eilenberg algebra,
    \begin{equation}
        \sfW(\frL)=\sfCE(T[1]\frL)=(\odot^\bullet(\frL[1]^*\oplus \frL[2]^*),\sfd_{\sfW})~.
    \end{equation} 
    We denote by $\sigma$ the canonical shift isomorphism $\frL[1]^*\rightarrow\frL[2]^*$, which extends to a nil\-qua\-dratic derivation $\sfW(\frL)$ of degree~$1$. The differential $\sfd_\sfW$ is again defined by its action on $t+\sigma(t)\in \frL[1]^*\oplus \frL[2]^*$, and we have
    \begin{equation}
        \sfd_\sfW t=\sfd_\sfCE t+\sigma(t)~,~~~\sfd_\sfW \sigma(t)=-\sigma(\sfd_\sfW t)~.
    \end{equation} 
    
    A general $\frL$-valued connection is then a dgca-morphism
    \begin{equation}
        \caA:\sfW(\frL)\rightarrow (\Omega^\bullet(U),\rmd)~,
    \end{equation} 
    where the de~Rham complex $(\Omega^\bullet(U),\rmd)$ can be regarded as the Weil algebra of $U$, which, in turn, is trivially regarded as a Lie algebroid. 
    
    As an example, let us consider again the special case of $\frL$ being a differential graded Lie algebra. The morphism $\caA$ is defined by the image of the generators, and we define
    \begin{equation}
        \begin{aligned}
            A^\alpha&\coloneqq \caA(t^\alpha)\in \Omega^{|t^\alpha|}(U)~,
            \\
            F^\alpha&\coloneqq \caA(\sigma(t^\alpha))\in \Omega^{|t^\alpha|+1}(U)~.
        \end{aligned}
    \end{equation}
    Compatibility with the differential gives the definition of the curvatures $F^\alpha$ in terms of the potential forms $A^\alpha$ as well as their Bianchi identities:
    \begin{equation}
        \begin{aligned}
            F^\alpha&=\rmd A^\alpha + \tfrac12 f^\alpha_{2\beta\gamma} A^\beta A^\gamma+f^\alpha_{1\beta} A^\beta~,
            \\
            0&=\rmd F^\alpha+f^\alpha_{2\beta\gamma} A^\beta F^\gamma+f^\alpha_{1\beta} F^\beta~.
        \end{aligned}
    \end{equation}
    For further concrete examples, see e.g.~the discussion in~\cite{Borsten:2024gox} or our discussion of Lie 3-algebras below.
    
    We note that in the description of connections as a morphism of the form~\eqref{eq:conn_morph}, we have a splitting of the generators of $\sfW(\frL)$ into ``horizontal ones'' and ``vertical ones'' with respect to the natural projection
    \begin{equation}
        \sfW(\frL)\rightarrow \sfCE(\frL)~.
    \end{equation} 
    The images of the former define the components of the local connection forms, while the images of the latter define the local curvature forms. Because the form of the curvature fixes the local and infinitesimal bundle isomorphisms, the choice of horizontal and vertical generators of the Weil algebra is relevant. In particular, dgca-automorphisms of the Weil algebra which rotate these generators lead to different notions of curvatures, which induce different forms of gauge transformations. This will become important in the following.
    
    \subsection{Adjusting \texorpdfstring{$L_\infty$}{L infinity}-algebras}
    
    Recall that the BRST complex is a differential graded commutative algebra that describes how the (higher) Lie algebra of infinitesimal bundle isomorphisms, also known as (higher) gauge transformations, acts on the local (higher) connection forms. Formally, it is the Chevalley--Eilenberg algebra of the (higher) action algebroid of the (higher) Lie algebra of gauge transformations acting on the connection, see e.g.~\cite{Jurco:2018sby} for more details.
    
    Consider again an $L_\infty$-algebra $\frL$ and a contractible manifold $U$. For flat $\frL$-valued connections, the corresponding BRST complex is given by the inner homomorphisms\footnote{Recall that morphisms of graded vector spaces are described, dually, by a map of graded rings.} $\ihom(T[1]U,\frL[1])$ in the category of N$Q$-manifolds (i.e., non-positively graded manifolds). Note that $\ihom(T[1]U,\frL[1])$ consists of elements of degrees $g\leq 0$, with $g$ what is called the negative ghost degree in physics. That is, elements of degree~$0$ describe the connection, elements of degree~$-1$ describe the ghosts or parameters of infinitesimal gauge transformations or bundle isomorphisms, elements of degree~$-2$ describe the ghosts for ghosts or parameters of higher infinitesimal gauge transformations or bundle isomorphisms, etc.
    
    Concretely, the inner hom is given by its image on generators $t^\alpha$ of $\sfCE(\frL[1])$, 
    \begin{equation}\label{eq:inner_hom_expansion}
        \caA(t^\alpha)=A^\alpha+\sum_{0<j\leq|t^\alpha|}\xi_{|t^\alpha|j}c^\alpha_{|t^\alpha|j}~.
    \end{equation} 
    Here, $A^\alpha$ is a local connection $|t^\alpha|$-form, $c^\alpha_{ij}\in \Omega^{i-j}(U)$ is a (higher) ghost of ghost degree~$j$, and $\xi_{ij}$ is a formal variable of ghost degree~$-j$, giving the required freedom for the inner homomorphisms.
    
    Note that $\ihom(T[1]U,\frL[1])$ carries a natural differential $Q$ which acts on a function $\caA \in (\ihom(T[1]U,T[1]\frL[1]))^*$ as\footnote{Note that inner homs here do not automatically commute with the differential.} 
    \begin{equation}
        Q\caA\coloneqq\rmd_U\circ \caA-\caA\circ \sfd_{\sfCE}~,
    \end{equation} 
    where $\rmd_U$ is the de~Rham differential on $U$ and $\rmd_{\sfW}$ the differential in the Weil algebra $\sfW(\frL)$. The BRST differential $Q_{\rm BRST} $ is obtained as the part of $Q$ with ghost degree~$1$.
    
    For general connections, we again replace $\frL[1]$ by $T[1]\frL[1]$ or, dually, switch from the Chevalley--Eilenberg algebra to the Weil algebra. This means that the maps~\eqref{eq:inner_hom_expansion} are extended to
    \begin{equation}\label{eq:inner_hom_expansion2}
        \begin{aligned}
            \caA(t^\alpha)&=A^\alpha+\sum_{0<j\leq|t^\alpha|}\xi_{|t^\alpha|j}c^\alpha_{|t^\alpha|j}~,
            \\
            \caA(\sigma(t^\alpha))&=F^\alpha+\sum_{0<j\leq|t^\alpha|+1}\zeta_{|t^\alpha|j}d^\alpha_{|t^\alpha|j}~,
        \end{aligned}
    \end{equation}
    where $A^\alpha$, $\xi_{ij}$ and $c^\alpha_{ij}$ are as in~\eqref{eq:inner_hom_expansion}, $F^\alpha$ are the local curvature forms, $d^\alpha_{ij}\in\Omega^{i+1-j}(U)$ are additional ghosts of ghost degree~$j$, and $\zeta_{ij}$ are additional formal parameters of ghost degree~$-j$. 
    
    Evidently, we get too many ghosts or, equivalently, too much gauge freedom, and we need to set the $d^\alpha_{ij}$ to zero. This truncation, however, is only consistent if the BRST differential, i.e., the part of ghost degree~1 of
    \begin{equation}
        \hat Q\caA\coloneqq\rmd_U\circ \caA-\caA\circ \sfd_{\sfW}
    \end{equation}
    vanishes on these; otherwise, we expect the BRST complex to be {\em open}. Since the flat BRST complex is unaffected, the failure of closure must be proportional to the curvature forms $F^\alpha$, for $\alpha$ in a certain index set $I$. 
    
    The $F^\alpha$ with $\alpha\in I$ are usually called {\em fake curvatures} and generically comprise all but the curvature components of highest form degree. The fact that the BRST complex is open means that connections and (higher) gauge transformations do not compose into a higher action Lie algebroid, as one would expect. At the finite level, this can be reflected in either gauge transformations not composing to gauge transformations or higher gauge transformations linking (higher) gauge transformations with different images. Both are problematic.
    
    For many $L_\infty$-algebras\footnote{We are not aware of any counterexample.}, one can achieve compatibility of the reduction with the differential by performing an automorphism $\phi$ of $\sfW(\frL)$. This changes the generators mapped to curvatures and therefore the action of the differential $\sfd_\sfW$ on the $d^\alpha_{ij}$. If the automorphism also preserves the notion of flat connections, i.e., the following diagram commutes:
    \begin{equation}\label{eq:respect_projection}
        \begin{tikzcd}
            \sfW(\frL)\arrow[r,"\phi"] \arrow[d] & \sfW(\frL) \arrow[d]
            \\
            \sfCE(\frL)\arrow[r,"\rmid"] & \sfCE(\frL)
        \end{tikzcd}
    \end{equation} 
    then it is called an \uline{adjustment}~\cite{Saemann:2019dsl,Kim:2019owc,Rist:2022hci,Fischer:2024vak}.
    
    Concretely, the automorphism $\phi:\sfW(\frL)\rightarrow \sfW(\frL)$, which is defined in terms of the images of $t+\sigma(t)\in \frL[1]^*\oplus \frL[2]^*$, is then of the form
    \begin{equation}
        \begin{aligned}
            \phi(t)&=t~,
            \\
            \phi(\sigma(t))&=\sigma(t)+\kappa_2(t,\sigma(t))+\kappa_{3,1}(t,t,\sigma(t))+\kappa_{3,2}(t,\sigma(t),\sigma(t))+\ldots~.
        \end{aligned}
    \end{equation}
    
    We note that a Weil differential is generically of the form
    \begin{equation}
        \begin{aligned}
            \sfd_\sfW(\phi(\sigma(t)))&=\sfd_\sfW(\sigma(t))+\sfd(\kappa_2(t,\sigma(t))+\ldots)
            \\
            &=-f_1(\sigma(t))\pm f_2(t,\sigma(t))\pm\ldots\pm \kappa_2(\sfd_\sfW t,\sigma(t))\pm\kappa_2(t,\sfd_\sfW \sigma(t))+\ldots~,
        \end{aligned}
    \end{equation}
    which induces a BRST differential
    \begin{equation}
        Q_{\rm BRST} d_{ij}^\alpha = \rmd_U d^\alpha_{i(j+1)}-R_{j+1}~,
    \end{equation} 
    where $R_{j+1}$ is the part of 
    \begin{equation}
        (\scA\circ \sfd_\sfW )(\phi(\sigma(t^\alpha)))
    \end{equation} 
    of ghost degree~$j+1$. Setting $d_{ij}^\alpha$ to zero leaves us with the condition
    \begin{equation}
        R_{j+1}=0~~~\mbox{for}~~~j\geq 1~.
    \end{equation} 
    This condition is satisfied if and only if $\sfd_\sfW(\phi(\sigma(t)))$ is of a particular form:
    \begin{proposition}\label{prop:adjustment_L_infty}
        Consider an $L_\infty$-algebra $\frL$ together with an automorphism $\phi$ on its Weil algebra covering the identity map on its Chevalley--Eilenberg algebra as in~\eqref{eq:respect_projection}. The map $\phi$ defines an adjustment if and only if we have
        \begin{equation}\label{eq:diff_form}
            \begin{aligned}
                \sfd_\sfW(\phi(\sigma(t)))=\sum_i \frac{1}{i!}\Big(&m_{\alpha_1\ldots \alpha_i}\phi(\sigma(t^{\alpha_1}))\ldots\phi(\sigma(t^{\alpha_i}))
                \\
                &+n_{\alpha_0\alpha_1\ldots \alpha_i}t^{\alpha_0}\phi(\sigma(t^{\alpha_1}))\ldots\phi(\sigma(t^{\alpha_i}))\Big)~,
            \end{aligned}
        \end{equation}
        for some structure constants $m_{\alpha_1\ldots \alpha_i}$ and $n_{\alpha_0\alpha_1\ldots \alpha_i}$, where $n_{\alpha_0\alpha_1\ldots \alpha_i}$ vanishes unless $|t^{\alpha_0}|=1$.
    \end{proposition}
    \begin{proof}
        The contribution of $m_{\alpha_1\ldots \alpha_i}\phi(\sigma(t^{\alpha_1}))\ldots\phi(\sigma(t^{\alpha_i}))$ to $R_{j+1}$ consists of a product of curvatures $F^{\alpha}$ and curvature ghosts $d_{ij}^{\alpha}$. Because $R_{j+1}$ has positive ghost degree, at least one factor is a ghost $d_{ij}^\alpha$ which is set to zero, and hence the contribution vanishes. The contribution of the second type to $R_{j+1}$ consists of a product of curvatures $F^{\alpha}$ and curvature ghosts $d_{ij}^{\alpha}$, as well as a single factor of either a gauge potential $A^\alpha$ or a ghost $c^\alpha_{11}$, where the restriction on the ghost is due to $|t^{\alpha_0}|=1$. Since $\deg(R_{j+1})\geq 2$, this implies that at least one factor of curvature ghosts $d_{ij}^{\alpha}$ has to appear, leading to a vanishing of the contribution after truncation. This argument also makes it clear that there cannot be any further terms in~\eqref{eq:diff_form}.
    \end{proof}
    
    We can now define the adjusted BRST complex, describing adjusted $L_\infty$-algebra-valued connections with their (higher) gauge transformations.
    \begin{definition}[\cite{Fischer:2024vak}](Adjusted BRST complex)
        Consider an $L_\infty$-algebra $\frL$ with adjustment, concentrated in non-positive degrees, a contractible manifold $U$, and the projection
        \begin{subequations}
            \begin{equation}
                \sfp_{\rm curv}: T[1]\frL[1]\cong\frL[1]\oplus \frL[2] \rightarrow \frL[2]
            \end{equation} 
            onto the vertical generators $\phi(\sigma(t))$ induced by the adjustment. The BRST complex for $\frL$-valued local connection forms on $U$ is the presheaf
            \begin{equation}
                \mathfrak{brst}^{\rm adj}(U,\frL)=\ihom^{\rm red}(T[1]U,T[1]\frL[1])~,
            \end{equation}
            where the reduced inner homomorphisms are given by 
            \begin{equation}\label{eq:ihom_red}
                \ihom^{\rm red}(\caN,T[1]\frL[1])\coloneqq \{\caA^*\in\ihom(\caN,T[1]\frL[1])~|~\rmdeg(\sfp_{\rm curv}\circ \caA^*)=0\}
            \end{equation}
            with $\deg(-)$ the ghost degree of the image.            The BRST differential is the ghost degree~1 part of 
            \begin{equation}
                \hat Q\caA\coloneqq\rmd_U\circ \caA-\caA\circ \sfd_{\sfW}
            \end{equation}
        \end{subequations}
        for $\caA\in (\ihom^{\rm red}(T[1]U,T[1]\frL[1]))^*$.
    \end{definition}
    
    The constrained form of the adjusted differential also has the following consequence.
    \begin{corollary}\label{cor:Bianchi}
        The Bianchi identities for adjusted $L_\infty$-algebra-valued connections are of the form
        \begin{equation}
            \rmd F=\sum_i \frac{1}{i!}\Big(n_{\alpha_0\alpha_1\ldots \alpha_i}A^{\alpha_0}F^{\alpha_1}\ldots F^{\alpha_i}+m_{\alpha_1\ldots \alpha_i}F^{\alpha_1}\ldots F^{\alpha_i}\Big)~.
        \end{equation} 
    \end{corollary}
    
    \begin{remark}\label{rem:action_modification}
        \Cref{cor:Bianchi} makes it clear that our construction is more general than the Chern--Simons terms of~\cite{Sati:2008eg}, in which only the contributions $m_{\alpha_1\ldots \alpha_i}F^{\alpha_1}\ldots F^{\alpha_i}$ were permitted. This is sufficient for the string-like $L_\infty$-algebras $\frL$, which contain a base Lie algebra $\frL_0$, which does not act on the other components $\frL_i$ with $i<0$. Hence, the de~Rham differential does not get ``covariantized'' by the terms $n_{\alpha_0\alpha_1\ldots \alpha_i}A^{\alpha_0}F^{\alpha_1}\ldots F^{\alpha_i}$. In the general case, however, the adjustments modify the action of the elements in $\frL_0$ on $\frL_i$ with $i<0$, which is invisible from the Chern--Simons term perspective.
    \end{remark}
    
    \subsection{Semi-strict Lie 3-algebras}\label{eq:example_3_term_L_infinity}
    
    We now calculate the form of the adjustment of a general 3-term $L_\infty$-algebra with underlying differential complex
    \begin{equation}
        \frL=(\frL_{-2} \xrightarrow{~\mu_1~} \frL_{-1} \xrightarrow{~\mu_1~} \frL_{0})\coloneqq (\frl \xrightarrow{~\mu_1~} \frh \xrightarrow{~\mu_1~} \frg)~.
    \end{equation}
    We denote by $t,r,s$ elements of $\frg[1]^*,\frh[1]^*,\frl[1]^*$, respectively. These generate the Cheval\-ley--Eilenberg algebra $\sfCE(\frL)$, and the differential $\sfd_\sfCE$ acts on these as follows:
    \begin{equation}
        \begin{aligned}
            \sfd_{\sfCE} t &=-\tfrac12 f_{2}(t,t)-f_1(r)~,
            \\
            \sfd_{\sfCE} r &=-\tfrac1{3!}f_3(t,t,t)-f_{2}(t,r)+f_1(s)~,
            \\
            \sfd_{\sfCE} s &=\tfrac1{4!}f_4(t,t,t,t)+\tfrac12f_3(t,t,r)-f_2(t,s)+\tfrac12f_2(r,r)~,
        \end{aligned}
    \end{equation}
    where the $f_i$ are multilinear maps $\odot^i \frL[1]^* \rightarrow \odot^i \frL[1]^*$ and the signs are a convenient convention. The corresponding Weil algebra is generated by elements $t,s,r$ and $\hat t=\sigma(t),\hat s=\sigma(s),\hat r=\sigma(r)$, which are elements in $\frg[2]^*,\frh[2]^*,\frl[2]^*$, respectively. The differential $\sfd_\sfW$ acts on these as 
    \begin{equation}
        \begin{aligned}
            \sfd_{\sfW} t &=-\tfrac12 f_{2}(t,t)-f_1(r)+\hat t~,
            \\
            \sfd_{\sfW} r &=-\tfrac1{3!}f_3(t,t,t)-f_{2}(t,r)+f_1(s)+\hat r~,
            \\
            \sfd_{\sfW} s &=\tfrac1{4!}f_4(t,t,t,t)+\tfrac12f_3(t,t,r)-f_2(t,s)+\tfrac12f_2(r,r)+\hat s~,
            \\
            \sfd_{\sfW} \hat t &=-f_{2}(t,\hat t)+f_1(\hat r)~,
            \\
            \sfd_{\sfW} \hat r &=\tfrac1{2!}f_3(t,t,\hat t)-f_{2}(t,\hat r)+f_{2}(\hat t,r)-f_1(\hat s)~,
            \\
            \sfd_{\sfW} \hat s &=\tfrac1{3!}f_4(t,t,t,\hat t)-\tfrac12f_3(t,t,\hat r)+f_3(t,\hat t,r)-f_2(r,\hat r)-f_2(t,\hat s)+f_2(\hat t,s)~,
        \end{aligned}
    \end{equation}
    where we extended the $f_i$ in the obvious manner to the tensor algebra of $\frL[1]^*\oplus \frL[2]^*$.
    
    The most general automorphism $\phi$ on $\sfW$ of the form~\eqref {eq:respect_projection} is evidently given by 
    \begin{equation}
        \begin{aligned}
            \tilde t&=\phi(t)=t~,~&\tilde r&=\phi(r)=r~,~&\tilde s&=\phi(s)=s~,
            \\
            \tilde{\hat{t}}&=\phi(\hat t)=\hat t~,~&\tilde{\hat{r}}&=\phi(\hat r)=\hat r-\kappa_1(t,\hat t)~,~&\tilde{\hat{s}}&=\phi(\hat s)=\hat s-\kappa_2(t,\tilde{\hat{r}})-\kappa_3(r,\hat t)-\tfrac12 \kappa_4(t,t,\hat t)
        \end{aligned}
    \end{equation}
    for some linear maps
    \begin{equation}
        \kappa_1:\frg\times \frg\rightarrow \frh
        ~,~~~
        \kappa_2:\frg\times \frh\rightarrow \frl
        ~,~~~
        \kappa_3:\frh\times \frg\rightarrow \frl
        ~,\eand
        \kappa_4:\frg\times\frg\times \frg\rightarrow \frl~,
    \end{equation} 
    with $\kappa_4$ anti-symmetric in its first two arguments. The Weil differential on these generators is readily computed.
    
    A map $\caA\in (\ihom(T[1]U,T[1]\frL[1]))^*$ now has components
    \begin{equation}\label{eq:component_maps}
        \begin{aligned}
            \caA(\tilde t)&=\xi_{11} c_{11}+A~,
            \\
            \caA(\tilde r)&=\xi_{22} c_{22}+\xi_{21} c_{21}+B~,
            \\
            \caA(\tilde s)&=\xi_{33} c_{33}+\xi_{32} c_{32}+\xi_{31} c_{31}+C~,
            \\
            \caA(\tilde{\hat{t}})&=\zeta_{12} d_{12}+\zeta_{11} d_{11}+F~,
            \\
            \caA(\tilde{\hat{r}})&=\zeta_{23} d_{23}+\zeta_{22} d_{22}+\zeta_{21} d_{21}+H~,
            \\
            \caA(\tilde{\hat{s}})&=\zeta_{34} d_{34}+\zeta_{33} d_{33}+\zeta_{32} d_{32}+\zeta_{31} d_{31}+G~,
        \end{aligned}
    \end{equation}
    where the components $A,F,c_{1i},d_{1i}$ take values in $\frg$, the components $B,H,c_{2i},d_{2i}$ take values in $\frh$, and the components $C,G,c_{3i},d_{3i}$ take values in $\frl$. Moreover, components in the $k$th column of~\eqref{eq:component_maps} are $(k-1)$-forms on $U$, and $\xi_{ij}$ and $\zeta_{ij}$ denote again formal variables of (ghost) degree~$-j$, as in~\eqref{eq:inner_hom_expansion2}.
    
    Using \ref{prop:adjustment_L_infty} and requiring that the Weil differential becomes of the form~\eqref{eq:diff_form} yields the following relations that the adjustment maps $\kappa_{1,2,3,4}$ must satisfy: 
    \begin{subequations}\label{eq:kappa-axioms-infinitesimal}
        \begin{equation}
            \begin{aligned}
                \kappa_1(\mu_1(Y),X)-\mu_1(\kappa_3(Y,X))+\mu_2(X,Y)&=0~,
                \\
                \kappa_2(\mu_1(Y_1),Y_2)+\mu_2(Y_1,Y_2)-\kappa_3(Y_1,\mu_1(Y_2))&=0~,
                \\
                \kappa_3(\mu_1(Z),X)-\mu_2(X,Z)&=0~,
            \end{aligned}
        \end{equation}
        \begin{equation}\label{eq:kappa-axioms-infinitesimal_k_1_2}
            \begin{aligned}
                \kappa_1(\mu_2(X_1,X_2),X_3)&=\kappa_1\Big(X_1,\mu_2(X_2,X_3)-\mu_1(\kappa_1(X_2,X_3))\Big)
                \\
                &\hspace{1cm}-\kappa_1\Big(X_2,\mu_2(X_1,X_3)-\mu_1(\kappa_1(X_1,X_3))\Big)
                \\
                &\hspace{1cm}+\mu_2(X_1,\kappa_1(X_2,X_3))-\mu_2(X_2,\kappa_1(X_1,X_3))
                \\
                &\hspace{1cm}-\mu_1(\kappa_4(X_1,X_2,X_3))+\mu_3(X_1,X_2,X_3)~,
            \end{aligned}
        \end{equation}
        \begin{equation}
            \begin{aligned}
                \kappa_2(\mu_2(X_1,X_2),Y)&=\kappa_2\Big(X_1,\mu_2(X_2,Y)-\kappa_1(X_2,\mu_1(Y))+\mu_1(\kappa_2(X_2,Y))\Big)
                \\
                &\hspace{1cm}-\kappa_2\Big(X_2,\mu_2(X_1,Y)-\kappa_1(X_1,\mu_1(Y))+\mu_1(\kappa_2(X_1,Y))\Big)
                \\
                &\hspace{1cm}+\mu_2(X_1,\kappa_2(X_2,Y))-\mu_2(X_2,\kappa_2(X_1,Y))
                \\
                &\hspace{1cm}-\kappa_4(X_1,X_2,\mu_1(Y))-\mu_3(X_1,X_2,Y)~,
            \end{aligned}
        \end{equation}
        \begin{equation}
            \begin{aligned}
                \kappa_3(\mu_2(X_1,Y),X_2)&=-\kappa_3\Big(Y,\mu_2(X_1,X_2)-\mu_1(\kappa_1(X_1,X_2))\Big)+\mu_2(X_1,\kappa_3(Y,X_2))
                \\
                &\hspace{1cm}-\mu_2(Y,\kappa_1(X_1,X_2))-\kappa_4(X_1,\mu_1(Y),X_2)+\mu_3(X_1,X_2,Y)~,
            \end{aligned}
        \end{equation}
        \begin{equation}
            \begin{aligned}
                \kappa_3(\mu_3(X_1,X_2,X_3),X_4)&=-\kappa_4\Big(X_1,X_2,\mu_2(X_3,X_4)-\mu_1(\kappa_1(X_3,X_4))\Big)
                \\
                &\hspace{1cm}+\kappa_4\Big(X_1,X_3,\mu_2(X_2,X_4)-\mu_1(\kappa_1(X_2,X_4))\Big)
                \\
                &\hspace{1cm}-\kappa_4\Big(X_2,X_3,\mu_2(X_1,X_4)-\mu_1(\kappa_1(X_1,X_4))\Big)
                \\
                &\hspace{1cm}+\kappa_4(X_1,\mu_2(X_2,X_3),X_4)-\kappa_4(X_2,\mu_2(X_1,X_3),X_4)
                \\
                &\hspace{1cm}+\kappa_4(X_3,\mu_2(X_1,X_2),X_4)+\mu_2(X_1,\kappa_4(X_2,X_3,X_4))
                \\
                &\hspace{1cm}-\mu_2(X_2,\kappa_4(X_1,X_3,X_4))+\mu_2(X_3,\kappa_4(X_1,X_2,X_4))
                \\
                &\hspace{1cm}+\mu_3(X_1,X_2,\kappa_1(X_3,X_4))-\mu_3(X_1,X_3,\kappa_1(X_2,X_4))
                \\
                &\hspace{1cm}+\mu_3(X_2,X_3,\kappa_1(X_1,X_4))-\mu_4(X_1,X_2,X_3,X_4)
            \end{aligned}
        \end{equation}
    \end{subequations}
    for all $X,X_i\in \frg$, $Y,Y_i\in\frh$, and $Z\in\frl$.
    
    Furthermore, the field strengths are given by 
    \begin{equation}
        \begin{aligned}
            F&=\rmd A+\tfrac12 \mu_2(A,A)+\mu_1(B)~,
            \\
            H&=\rmd B+\mu_2(A,B)-\tfrac1{3!}\mu_3(A,A,A)+\mu_1(C)-\kappa_1(A,F)~,
            \\
            G&=\rmd C+\mu_2(A,C)-\tfrac12\mu_2(B,B)-\tfrac12\mu_3(A,A,B)-\tfrac1{4!}\mu_4(A,A,A,A)
            \\
            &\hspace{1cm}+\kappa_2(A,H)+\kappa_3(B,F)-\tfrac12 \kappa_4(A,A,F)~,
        \end{aligned}
    \end{equation}
    and the corresponding Bianchi identities are
    \begin{subequations}\label{eq:Bianchi_identities}
        \begin{equation}
            \begin{aligned}
                \rmd F+\mu_2(A,F)-\mu_1(\kappa_1(A,F))-\mu_1(H)&=0~,
                \\
                \rmd H+\mu_2(A,H)-\kappa_1(A,\mu_1(H))+\mu_1(\kappa_2(A,H))+\kappa_1(F,F)-\mu_1(G)&=0~,
            \end{aligned}
        \end{equation}
        and
        \begin{equation}
            \begin{aligned}
                &\rmd G+\mu_2(A,G)+\kappa_2(A,\mu_1(G))-\kappa_2(A,\kappa_1(F,F))-\kappa_3(\kappa_1(A,F),F)-\kappa_4(A,F,F)
                \\
                &\hspace{1cm}=\kappa_2(F,H)+\kappa_3(H,F)~.
            \end{aligned}
        \end{equation}
    \end{subequations}
    
    The (1-)gauge transformations of the gauge potential components read as 
    \begin{equation}
        \begin{aligned}
            \delta A&=\rmd \alpha+\mu_2(A,\alpha)-\mu_1(\Lambda)~,
            \\
            \delta B&=\rmd \Lambda+\mu_2(A,\Lambda)-\mu_2(\alpha,B)+\tfrac{1}{2} \mu_3(A,A,\alpha)-\mu_1(\Sigma)+\kappa_1(\alpha,F)~,
            \\
            \delta C&=\rmd \Sigma+\mu_2(A,\Sigma)+\mu_2(B,\Lambda)-\mu_2(\alpha,C)+\tfrac{1}{2} \mu_3(A,A,\Lambda)-\mu_3(A,\alpha,B)
            \\
            &\hspace{1cm}-\tfrac{1}{3!} \mu_4(A,A,A,\alpha)-\kappa_2(\alpha,H)-\kappa_3(\Lambda,F)-\kappa_4(A,\alpha,F)~,
        \end{aligned}
    \end{equation}
    where we identified
    \begin{equation}
        \alpha=c_{11}~,~~~\Lambda=c_{21}~,\eand \Sigma=c_{31}~.
    \end{equation}
    The curvature components transform according to 
    \begin{equation}
        \begin{aligned}
            \delta F&=-\mu_2(\alpha,F)+\mu_1(\kappa_1(\alpha,F))~,
            \\
            \delta H&=-\mu_2(\alpha,H)+\kappa_1(\alpha,\mu_1(H))-\mu_1(\kappa_2(\alpha,H))~,
            \\
            \delta G&=-\mu_2(\alpha,G)-\kappa_2(\alpha,\mu_1(G))+\kappa_2(\alpha,\kappa_1(F,F))+\kappa_3(\kappa_1(\alpha,F),F)+\kappa_4(\alpha,F,F)~.
        \end{aligned}
    \end{equation}
    
    There are also higher or 2-gauge transformations, which are parameterized by
    \begin{equation}
        \beta=c_{22}~,~~~\Xi=c_{32}~,
    \end{equation} 
    and their infinitesimal action on the parameters of 1-gauge transformations read as
    \begin{equation}
        \begin{aligned}
            \delta \alpha&=\mu_1(\beta)~,
            \\
            \delta \Lambda&=\rmd \beta+\mu_2(A,\beta)+\mu_1(\Xi)~,
            \\
            \delta \Sigma&=\rmd \Xi+\mu_2(A,\Xi)-\mu_2(B,\beta)+\kappa_3(\beta,F)-\tfrac12 \mu_3(A,A,\beta)~,
        \end{aligned}
    \end{equation}
    as well as 3-gauge transformations, parameterized by $\gamma=c_{33}$, which act as
    \begin{equation}
        \begin{aligned}
            \delta \beta&=-\mu_1(\gamma)~,
            \\
            \delta \Xi=&\rmd \gamma+\mu_2(A,\gamma)~.
        \end{aligned}
    \end{equation}
    
    Altogether, we note that there are four adjustment maps $\kappa_{1,2,3,4}$, which satisfy the relations~\eqref{eq:kappa-axioms-infinitesimal}. Their role is to deform the linear action of $\frg$ onto the total curvature form $F+G+H$ to the following non-linear one:
    \begin{equation}
        \begin{aligned}
            \mu_2(X,F+H+G)&\rightarrow \mu_2(X,F)-\mu_1(\kappa_1(X,F))
            \\
            &\hspace{1cm}+\mu_2(X,H)-\kappa_1(X,\mu_1(H))+\mu_1(\kappa_2(X,H))
            \\
            &\hspace{1cm}+\mu_2(X,G)+\kappa_2(X,\mu_1(G))
            \\
            &\hspace{1cm}-\kappa_2(X,\kappa_1(F,F))-\kappa_3(\kappa_1(X,F),F)-\kappa_4(X,F,F)~,
        \end{aligned}
    \end{equation}
    as indicated in \ref{rem:action_modification}. Contrary to the case of connections taking values in a Lie 2-algebra, here both the gauge transformations and 2-gauge transformations are deformed by the maps $\kappa_i$.
    
    \section{The BRST Lie 3-groupoid}\label{sec:BRST-groupoid}
    
    In order to define connections on principal 3-bundles, we need to develop a suitable BRST\footnote{The name derives from the Becchi--Rouet--Stora--Tyutin (BRST) complex in physics, which is essentially the Chevalley--Eilenberg algebra of the gauge action algebroid of the (higher) gauge algebra under consideration.} Lie 3-groupoid integrating the infinitesimal gauge transformations derived earlier.
    
    The integration of general semi-strict Lie 3-algebras to a Lie 3-group is complicated. Since we are interested in formulations that allow for concrete computations, we restrict ourselves to semistrict 3-groups. These can be described by 2-crossed modules of Lie groups. The corresponding 2-crossed modules of Lie algebras have an overlap with Lie 3-algebras, see \ref{app:Lie-3-algebras} and \ref{app:2crossedmodules} for more details on 2-crossed modules.
    
    Concretely we will always work with a structure Lie 3-group given by a 2-crossed module of Lie groups $\caG=(\sfL \rightarrow \sfH \rightarrow \sfG)$ with corresponding 2-crossed module of Lie algebras $\sfLie(\caG)=(\frl \rightarrow \frh \rightarrow \frg)$.  We will also restrict ourselves to some manifold $U$, which is to be seen as a local patch of the base manifold of a higher principal bundle.
    
    We first make explicit the relevant strict BRST Lie 3-algebroid. This is a slight and evident modification of the BRST Lie 3-algebroid given in \ref{eq:example_3_term_L_infinity} which is obtained by using the relations between 3-term $L_\infty$-algebras and 2-crossed modules of Lie algebras explained in~\ref{app:L-infty}. We then construct, step by step, the integrating BRST Lie 3-groupoid, and explain in some detail the underlying computations.
    
    \subsection{The BRST Lie 3-algebroid}
    
    For the gauge Lie 3-algebra given by $\sfLie(\caG)$ and a particular manifold $U$, we have a strict action Lie 3-algebroid, describing local connection forms, their (1-)gauge, 2-gauge, and 3-gauge transformations. The base manifold of this algebroid\footnote{We shall ignore analytical subtleties relating to the infinite-dimensionality of the function spaces, as it is irrelevant to the discussion.} consists of the local connection forms 
    \begin{equation}
        A\in \Omega^1(U,\frg)~,~~~B\in \Omega^2(U,\frh)~,~~~C\in \Omega^3(U,\frl)~.
    \end{equation}
    Infinitesimal gauge transformations are parameterized by
    \begin{equation}
        \alpha\in \Omega^0(U,\frg)~,~~~\Lambda\in \Omega^1(U,\frh)~,~~~\Sigma\in\Omega^2(U,\frl)~,
    \end{equation}
    and act according to
    \begin{equation}
        \begin{aligned}
            \delta A&=\rmd \alpha+[A,\alpha]-\sft(\Lambda)~,
            \\
            \delta B&=\rmd \Lambda+A\acton \Lambda-\alpha\acton B-\sft(\Sigma)+\kappa_1(\alpha,F)~,
            \\
            \delta C&=\rmd \Sigma+A\acton \Sigma-\{B,\Lambda\}-\{\Lambda,B\}-\alpha\acton C-\kappa_2(\alpha,H)-\kappa_3(\Lambda,F)~,
        \end{aligned}
    \end{equation}
    where we introduced the maps 
    \begin{equation}
        \kappa_1:\frg\times \frg\rightarrow \frh~,~~~\kappa_2:\frg\times \frh\rightarrow \frl~,\eand \kappa_3:\frh\times \frg\rightarrow \frl~,
    \end{equation} 
    extended in the obvious way to $\sfLie(\caG)$-valued differential forms, that satisfy
    \begin{subequations}\label{eq:kappa-axioms-infinitesimal-2xm}
        \begin{equation}
            \begin{aligned}
                \kappa_1(\sft(Y),X)-\sft(\kappa_3(Y,X))+X\acton Y&=0~,
                \\
                \kappa_2(\sft(Y_1),Y_2)-\{Y_1,Y_2\}-\{Y_2,Y_1\}-\kappa_3(Y_1,\sft(Y_2))&=0~,
                \\
                \kappa_3(\sft(Z),X)-X\acton Z&=0~,
            \end{aligned}
        \end{equation}
        \begin{equation}
            \begin{aligned}
                \kappa_1([X_1,X_2],X_3)&=\kappa_1\Big(X_1,[X_2,X_3]-\sft(\kappa_1(X_2,X_3))\Big)
                -\kappa_1\Big(X_2,[X_1,X_3]-\sft(\kappa_1(X_1,X_3))\Big)
                \\
                &\hspace{1cm}+X_1\acton\kappa_1(X_2,X_3)-X_2\acton \kappa_1(X_1,X_3)~,
            \end{aligned}
        \end{equation}
        \begin{equation}
            \begin{aligned}
                \kappa_2([X_1,X_2],Y)&=\kappa_2\Big(X_1,X_2\acton Y-\kappa_1(X_2,\sft(Y))+\sft(\kappa_2(X_2,Y))\Big)
                \\
                &\hspace{1cm}-\kappa_2\Big(X_2,X_1\acton Y-\kappa_1(X_1,\sft(Y))+\sft(\kappa_2(X_1,Y))\Big)
                \\
                &\hspace{1cm}+X_1\acton \kappa_2(X_2,Y)-X_2\acton \kappa_2(X_1,Y)~,
            \end{aligned}
        \end{equation}
        \begin{equation}
            \begin{aligned}
                \kappa_3(X_1\acton Y,X_2)&=-\kappa_3\Big(Y,[X_1,X_2]-\sft(\kappa_1(X_1,X_2))\Big)+X_1\acton \kappa_3(Y,X_2)
                \\
                &\hspace{1cm}+\{Y,\kappa_1(X_1,X_2)\}+\{\kappa_1(X_1,X_2),Y\}~,
            \end{aligned}
        \end{equation}
    \end{subequations}
    Infinitesimal 2-gauge transformations are parameterized by
    \begin{equation}
        \beta\in\Omega^0(U,\frh)
        \eand
        \Xi\in \Omega^1(U,\frl)
    \end{equation}
    and act according to
    \begin{equation}
        \begin{aligned}
            \delta \alpha&=\sft(\beta)~,
            \\
            \delta \Lambda&=\rmd \beta+A\acton \beta+\sft(\Xi)~,
            \\
            \delta \Sigma&=\rmd \Xi+A\acton \Xi+\{B,\beta\}+\{\beta,B\}+\kappa_3(\beta,F)~.
        \end{aligned}
    \end{equation}
    Finally, we have infinitesimal 3-gauge transformations, parameterized by 
    \begin{equation}
        \gamma\in \Omega^0(U,\frl)
    \end{equation}
    and acting as 
    \begin{equation}
        \begin{aligned}
            \delta \beta&=\sft(\gamma)~,
            \\
            \delta \Xi=&\rmd \gamma+A\acton \gamma~.
        \end{aligned}
    \end{equation}
    The adjusted curvature forms read as
    \begin{equation}
        \begin{aligned}
            F&=\rmd A+\tfrac12[A,A]+\sft(B)~,
            \\
            H&=\rmd B+A\acton B+\sft(C)-\kappa_1(A,F)~,
            \\
            G&=\nabla C+\{B,B\}+\kappa_2(A,H)+\kappa_3(B,F)~.
        \end{aligned}
    \end{equation}
    
    \begin{definition}
        The above-listed spaces and actions uniquely define a strict action Lie 3-algebroid $\frA(U,\sfLie(\caG))$, which we call the \uline{BRST Lie 3-algebroid} for the Lie 3-algebra $\sfLie(\caG)$.
    \end{definition}
    \noindent We refrain from discussing this BRST Lie 3-algebroid in any explicit form (e.g.~as an N$Q$-manifold) since it is not relevant to our further discussion. For our purposes, it is fully sufficient to specify what is being acted on and what these actions are; see also the clarification after \ref{thm:integration}. 
    
    \subsection{Objects and 1-morphisms}
    
    We now come to the integration of the BRST Lie 3-algebroid $\frA(U,\sfLie(\caG))$. We will construct a strict integrating Lie 3-groupoid iteratively. We present a path through the required computations, but we will suppress individual steps that can be completed by iterative application of the axioms in a 2-crossed module of Lie groups and Lie algebras, outlined in~\ref{app:2crossedmodules} and~\ref{app:Formulas}. We stress that many of these computations are of considerable length and best done using a computer algebra programme.
    
    We will postulate parts of 1-, 2-, and 3-morphisms that can be deduced from the corresponding expressions for principal 2-bundles or the Čech cocycles explored e.g.~in~\cite{Saemann:2013pca}. We then have to check that the morphisms compose correctly, that the composition is unital, associative and comes with an inverse, and that higher morphisms do not change the images of the lower morphisms.
    
    The strict Lie 3-groupoid $\scB(\caG,U)$ has as its objects
    \begin{equation}
        \scB(\caG,U)_0\coloneqq \{~(A,B,C)~|~ A\in \Omega^1(U,\frg)~,~B\in \Omega^2(U,\frh)~,~C\in \Omega^3(U,\frl)~\}~.
    \end{equation}
    The objects of its morphism 2-category are
    \begin{equation}
        \scB(\caG,U)_1\coloneqq \scB(\caG,U)_0\times \{~(g,\Lambda,\Sigma)~|~g\in \Omega^0(U,\sfG)~,~\Lambda\in \Omega^1(U,\frh)~,~\Sigma\in\Omega^2(U,\frl)~\}~,
    \end{equation} 
    and we have 1-morphisms
    \begin{equation}
        \begin{tikzcd}[column sep=3.5cm]
            (\tilde A,\tilde B,\tilde C) & \arrow[l,"{(A,B,C,g,\Lambda,\Sigma)}"'] (A,B,C)
        \end{tikzcd}~.
    \end{equation}
    
    The formulas for $\tilde A$, $\tilde B$, $\tilde C$ can be gleaned from the known results for undeformed principal 3-bundles~\cite{Saemann:2013pca} as well as adjusted principal 2-bundles~\cite{Rist:2022hci}. We therefore know that
    \begin{equation}\label{eq:1-morphisms}
        \begin{aligned}
            \tilde A&=g^{-1}A g+g^{-1}\rmd g-\sft(\Lambda)~,
            \\
            \tilde B&=g^{-1}\acton B+\tilde\nabla\Lambda+\tfrac12[\Lambda, \Lambda]-\sft(\Sigma)-\kappa_1(g^{-1},F)~,
            \\
            \tilde C&=g^{-1}\acton C+\tilde \nabla\Sigma+\sft(\Lambda)\acton \Sigma-\{\tilde B,\Lambda\}-\{\Lambda,\tilde B\}+\{\Lambda,\tilde \nabla\Lambda\}+\tfrac12\{\Lambda,[\Lambda,\Lambda]\}
            \\
            &\hspace{1cm}+\kappa_2(g^{-1},H)-R_1(g,\Lambda,F)~.
        \end{aligned}
    \end{equation}
    Here, we have introduced the following functions parameterizing deformations of the gauge transformations of fake-flat curvatures (as found in~\cite{Saemann:2013pca}):
    \begin{equation}
        \begin{aligned}
            \kappa_1&:&\sfG\times \frg&\rightarrow \frh~,~~~&\kappa_2&:&\sfG\times \frh&\rightarrow \frl~,
            \\
            R_1&:&\sfG\times \frh\times \frg&\rightarrow \frl~,
        \end{aligned}
    \end{equation}
    which are linear in their Lie-algebra-valued arguments, and which we extend in the obvious manner to $\sfLie(\caG)$-valued differential forms. In the following, we will also denote the linearizations (i.e., differentials restricted to the units) of the maps $\kappa_1$ and $\kappa_2$, i.e., the maps $\kappa_1:\frg\times \frg\rightarrow \frh$ and $\kappa_2:\frg\times \frg\rightarrow \frl$, by the same symbols.
    
    The fact that the unit gauge transformation $(g,\Lambda,\Sigma)=(\unit,0,0)$ has to preserve $A$, $B$, and $C$ implies the following rules:
    \begin{equation}\label{eq:der_k1_k2_unitality}
        \kappa_1(\unit,X)=0\eand
        \kappa_2(\unit,Y)=0~.
    \end{equation}
    
    The form of the 1-morphisms also implies that the 2- and 3-form curvatures of source and target are related by
    \begin{equation}
        \begin{aligned}
            \tilde F&=g^{-1}F g-\sft(\kappa(g^{-1},F))~,
            \\
            \tilde H&=g^{-1}\acton H+\sft(\kappa_2(g^{-1},H))+
            \tilde F\acton \Lambda-\sft(\kappa_3(g,\Lambda,F))+\kappa_1(\sft(\Lambda),\tilde F)-\rmd(\kappa_1(g^{-1},F))
            \\
            &\hspace{1cm}-\kappa_1(g^{-1}\nabla g,\tilde F)-(g^{-1}\nabla g)\acton\kappa_1(g^{-1},F)+g^{-1}\acton\kappa_1(A,F)~.
            \\
            \tilde G&=g^{-1}\acton G+\tilde F\acton \Sigma+\{\Lambda,\tilde H\}-\{\tilde H,\Lambda\}-\{\Lambda,\tilde F\acton \Lambda\}+\ldots~,
        \end{aligned}
    \end{equation}
    where $\ldots$ captures terms given by deformations $\kappa_{1,2}$ and $R_{1}$. In particular, $\ldots$ vanishes for flat connections $(A,B,C)$ with $(F,H,G)=(0,0,0)$. Inspection immediately leads to the following lemma.
    \begin{lemma}
        The 1-morphisms~\eqref{eq:1-morphisms} map flat connections to flat connections.
    \end{lemma}
    
    Consider now the (strict) composition of 1-morphisms:
    \begin{equation}\label{eq:diag_comp_1_morphisms}
        \begin{tikzcd}[column sep=0.5in]
            & (A_2,B_2,C_2) \arrow[ld,"{(\ldots,g_2,\Lambda_2,\Sigma_2)}"']
            \\
            (A_3,B_3,C_3) & {} & \arrow[lu,"{(\ldots,g_1,\Lambda_1,\Sigma_1)}"'] (A_1,B_1,C_1)
            \arrow[ll,"{(\ldots,g_3,\Lambda_3,\Sigma_3)}"]
        \end{tikzcd}
    \end{equation} 
    Comparing the component $A_3$ and allowing arbitrary $A_1$, we obtain
    \begin{equation}
        g_3=g_1g_2\eand \sft(\Lambda_3)=\sft(\Lambda_2+g^{-1}_2\acton \Lambda_1)~.
    \end{equation} 
    Comparing the component $B_3$ and allowing arbitrary $A_1$, $B_1$, we obtain
    \begin{equation}
        g_3=g_1g_2~,~~~\Lambda_3=\Lambda_2+g^{-1}_2\acton \Lambda_1~,~~~\sft(\Sigma_3)=\sft(\Sigma_2+g^{-1}_2\acton \Sigma_1+\{\Lambda_2,g_2^{-1}\acton \Lambda_1\}-\kappa_4(g_1,g_2,F))
    \end{equation}
    as well as the first adjustment condition:
    \begin{equation}\label{eq:adj_cond_1_1}
        \begin{aligned}
            g_1\acton \kappa_1(g_2,X)&=\kappa_1(g_1 g_2,X)-\kappa_1\left(g_1,\rmAd_{g_2}(X)-\sft(\kappa_1(g_2,X))\right)-\sft\left(\kappa_4(g_2^{-1},g_1^{-1},X)\right)~,
        \end{aligned}
    \end{equation}
    where we have an adjustment function\footnote{This function is included for generality; we will see that it vanishes for the strict setting considered here.}
    \begin{equation}
        \kappa_4:\sfG\times \sfG\times \frg \rightarrow \frl~.
    \end{equation}
    We note that the full linearization of this adjustment function, i.e., the function $\kappa_4:\frg\times \frg\times \frg\rightarrow \frl$ also incorporates a factor of $\tfrac12$. This becomes obvious comparing condition~\eqref{eq:adj_cond_1_1} and~\eqref{eq:kappa-axioms-infinitesimal_k_1_2}.
    
    Condition \eqref{eq:adj_cond_1_1} also implies the relation
    \begin{equation}\label{eq:k1_derivation_rule}
        \begin{aligned}
            \rmd\kappa_1(g,X)&=g\acton \kappa_1(g^{-1}\rmd g,X)+\kappa(g,[g^{-1}\rmd g,X]-\sft(\kappa_1(g^{-1}\rmd g,X)))+\kappa_1(g,\rmd X)
            \\
            &\hspace{1cm}-\sft(\kappa_4(g^{-1}\rmd g,g^{-1},X))~.
        \end{aligned}
    \end{equation}    
    
    Comparing the component $C_3$ for arbitrary $A_1$, $B_1$, $C_1$, we then obtain the complete composition rules
    \begin{equation}
        g_3=g_1g_2~,~~~\Lambda_3=\Lambda_2+g^{-1}_2\acton \Lambda_1~,~~~\Sigma_3=\Sigma_2+g^{-1}_2\acton \Sigma_1+\{\Lambda_2,g_2^{-1}\acton \Lambda_1\}-\kappa_4(g_1,g_2,F)~.
    \end{equation}
    as well as an additional consistency condition
    \begin{equation}
        \caC_1=0~,
    \end{equation}
    where $\caC_1$ is an expression that vanishes for flat connections. It is listed in \ref{app:raw_adjustment_condition}, and we will analyze this condition later.
    
    Setting one of the morphisms $(\ldots,g_{1,2},\Lambda_{1,2},\Sigma_{1,2})$ to $(\ldots,\unit,0,0)$, we recover the expected transformation in the composition if
    \begin{equation}\label{eq:der_k4_unitality}
        \kappa_4(g,\unit,X)=\kappa_4(\unit,g,X)=0
    \end{equation} 
    for all $g\in \sfG$ and $X\in \frg$.
    
    Demanding associativity of composition amounts to the condition
    \begin{equation}\label{eq:der_k4}
        \begin{aligned}
            g_1\acton \kappa_4(g_2,g_3,X)&=\kappa_4(g_2,g_3g_1^{-1},X)-\kappa_4(g_2g_3,g_1^{-1},X)
            \\
            &\hspace{1cm}+\kappa_4(g_3,g_1^{-1},g_2^{-1}Xg_2-\sft(\kappa_1(g_2^{-1},X)))~.
        \end{aligned}
    \end{equation}
    
    Inverses of 1-morphisms exists if and only if
    \begin{equation}
        \kappa_4(g^{-1},g,X-Ad_g(X)+\sft(\kappa_1(g^{-1},X)))=0~,
    \end{equation} 
    and they are given by
    \begin{equation}
        (A,B,C,g,\Lambda,\Sigma)^{-1}=\Big(\tilde A,\tilde B,\tilde C,g^{-1},-g\acton \Lambda,-g\acton \Sigma+g\acton\{\Lambda,\Lambda\}-\kappa_4(g,g^{-1},F)\Big)~.
    \end{equation}
    
    \subsection{2-morphisms}
    
    The 2-morphisms are parameterized by
    \begin{equation}
        \scB(\caG,U)_2\coloneqq \scB(\caG,U)_1\times \{~(h,\Xi)~|~h\in \Omega^0(U,\sfH)~,~\Xi\in \Omega^1(U,\frl)~\}~,
    \end{equation} 
    and take the form
    \begin{equation}
        \begin{tikzcd}[column sep=3.5cm]
            (A,B,C,\tilde g,\tilde \Lambda,\tilde \Sigma) & \arrow[Rightarrow,l,"{(A,\ldots,\Sigma,h,\Xi)}"'] (A,B,C,g,\Lambda,\Sigma)
        \end{tikzcd}~.
    \end{equation}
    
    Consider then the globularity of the 2-morphisms:
    \begin{equation}\label{eq:glob_2_morphisms}
        \begin{tikzcd}[column sep=3.5cm]
            (\tilde A,\tilde B,\tilde C) & (A,B,C)
            \arrow[l,bend right=30,"{(\ldots,g,\Lambda,\Sigma)}"{name=U},swap]
            \arrow[l,bend left=30,"{~(\ldots,\tilde g,\tilde \Lambda,\tilde \Sigma)}"{name=D}]
            \arrow[Rightarrow,"{(\ldots,h,\Xi)}", from=U, to=D,start anchor={[yshift=-1ex]},end anchor={[xshift=0.325ex,yshift=1ex]}]
        \end{tikzcd}
    \end{equation}
    
    Extending the known result for the corresponding BRST Lie 2-groupoid for principal 2-bundles~\cite{Rist:2022hci}, we have
    \begin{equation}
        \begin{aligned}
            \tilde g&=\sft(h)g~,
            \\
            \tilde \Lambda &= \Lambda+g^{-1}\acton (h^{-1}\nabla h)+\sft(g^{-1}\acton\Xi)~.
        \end{aligned}
    \end{equation}
    The 2-commutativity of diagram~\eqref{eq:glob_2_morphisms} in the $\tilde B$ and $\tilde C$ components then implies 
    \begin{equation}
        \begin{aligned}
            \tilde \Sigma&=\Sigma+g^{-1}\acton (\nabla \Xi+\tfrac12[\Xi,\Xi]+\{\sft(\Xi),\Lambda\}-\{h^{-1},B\}-\{B,h^{-1}\})
            \\
            &\hspace{1cm}+\{\Lambda,g^{-1}\acton (h^{-1}\nabla h+\sft(\Xi))\}+R_2(g,h,F)~,
        \end{aligned}
    \end{equation}
    where we introduced a function
    \begin{equation}
        R_2:\sfG\times \sfH\times \frg\rightarrow \frl~.
    \end{equation} 
    Moreover, we have the further conditions
    \begin{equation}\label{eq:der_k1_1}
        \kappa_1(\sft(h),X)=\rmAd^\acton_h(X)-\sft(R_2(\unit,h^{-1},X))~,
    \end{equation} 
    where
    \begin{equation}
        \rmAd^\acton(f,A_a)\coloneqq f (A_a\acton f^{-1})~,
    \end{equation} 
    and
    \begin{equation}
        R_2(g,h,X)=-\kappa_4(\sft(h),g,X)+g^{-1}\acton \kappa_3(h,X)
    \end{equation}
    for some function $\kappa_3:\sfH\times \frg\rightarrow \frl$, as well as 
    \begin{equation}
        \caC_2=0~,
    \end{equation}
    given again in \ref{app:raw_adjustment_condition}. We again postpone analysis of this condition to later.
    
    Also, since $(h,\Xi)=(\unit,0)$ should parameterize the identity gauge transformation, we have
    \begin{equation}\label{eq:der_k3_unitality}
        R_2(g,\unit,X)=0
    \end{equation} 
    for all $g\in \sfG$ and $X\in\frg$.
    
    Next, we consider the vertical composition of 2-morphisms
    \begin{equation}\label{eq:diag_comp_2_morphisms}
        \begin{tikzcd}[column sep=0.7in]
            & (\ldots,g_2,\Lambda_2,\Sigma_2) \arrow[Rightarrow,ld,"{(\ldots,h_2,\Xi_2)}"']
            \\
            (\ldots,g_3,\Lambda_3,\Sigma_3) & & (\ldots,g_1,\Lambda_1,\Sigma_1)\arrow[Rightarrow,lu,"{(\ldots,h_1,\Xi_1)}"']
            \arrow[Rightarrow,ll,"{(\ldots,h_3,\Xi_3)}"]
        \end{tikzcd}
    \end{equation} 
    Comparing the images at $g_3$ and $\Lambda_3$, we find\footnote{The inverse order in the product $h_2h_1$ is a choice of convention, which we make to match the conventions in earlier papers, such as~\cite{Saemann:2013pca}.}
    \begin{equation}
        \begin{aligned}
            h_3&=h_2h_1~,
            \\
            \Xi_3&=\Xi_1+\sft(h_1^{-1})\acton\Xi_2-\{h_1^{-1},h_2^{-1}\nabla h_2\}~,
        \end{aligned}
    \end{equation}
    as well as an additional condition
    \begin{equation}
        \caC_3=0~.
    \end{equation} 
    
    \subsection{3-morphisms}
    
    Finally, the 3-morphisms are parameterized by
    \begin{equation}
        \scB(\caG,U)_3\coloneqq \scB(\caG,U)_2\times \{~\ell~|~\ell\in \Omega^0(U,\sfL)~\}~,
    \end{equation} 
    and take the form
    \begin{equation}
        \begin{tikzcd}[column sep=3.5cm]
            (A,B,C,g,\Lambda,\Sigma,\tilde h,\tilde \Xi) & \arrow[triple,l,"{(A,\ldots,\Xi,\ell)}"'] (A,B,C,g,\Lambda,\Sigma,h,\Xi)
        \end{tikzcd}~.
    \end{equation}
    We start again by considering the globularity of the 3-morphisms:
    \begin{equation}\label{eq:glob_3_morphisms}
        \begin{tikzcd}[column sep=3.5cm]
            (\ldots,\tilde g,\tilde \Lambda,\tilde \Sigma) & (\ldots,g,\Lambda,\Sigma)
            \arrow[Rightarrow,l,bend right=30,"{(\ldots,h,\Xi)}"{name=U},swap]
            \arrow[Rightarrow,l,bend left=30,"{~(\ldots,\tilde h,\tilde \Xi)}"{name=D}]
            \arrow[triple,"{(\ldots,\ell)}", from=U, to=D,start anchor={[yshift=-1ex]},end anchor={[xshift=0.325ex,yshift=1ex]}]
        \end{tikzcd}
    \end{equation}
    From the expected component of the 3-morphism,
    \begin{equation}
        \tilde h=\sft(\ell)h~,
    \end{equation} 
    we then conclude that
    \begin{equation}
        \tilde \Xi=\Xi-h^{-1}\acton (\ell^{-1}\nabla \ell)
    \end{equation}
    by demanding 3-commutativity of the diagram~\eqref{eq:glob_3_morphisms}. This also implies the additional condition
    \begin{equation}
        \caC_4=0~.
    \end{equation}
    
    The trivial three morphisms are of the form $(A,\ldots,\Xi,\unit)$. Transversal composition\footnote{By transversal, we mean the third in a sequence of horizontal, vertical, transversal.} of 3-morphisms
    \begin{equation}\label{eq:diag_comp_3_morphisms}
        \begin{tikzcd}[column sep=0.7in]
            & (\ldots,h_2,\Xi_2) \arrow[triple,ld,"{(\ldots,\ell_2)}"']
            \\
            (\ldots,h_3,\Xi_3) & & (\ldots,h_1,\Xi_1)\arrow[triple,lu,"{(\ldots,\ell_1)}"']
            \arrow[triple,ll,"{(\ldots,\ell_3)}"]
        \end{tikzcd}
    \end{equation} 
    is simply given by
    \begin{equation}
        \ell_3=\ell_2\ell_1~.
    \end{equation} 
    This composition is obviously associative and has an evident inverse.

    \subsection{Remaining 3-groupoid structures}

	Beyond the above data, we also need the following structures for a 3-groupoid, cf.~\ref{app:Gray}.
	
	\begin{lemma}
		The vertical composition of 3-morphisms is given by
		\begin{equation}
			\begin{aligned}
				&(\ldots,h_2,\Xi_2,\ell_2)\bullet (\ldots,h_1,\Xi_1,\ell_1)
				\\
				&\hspace{1cm}\coloneqq(\ldots, h_2h_1,\Xi_1+\sft(h_1^{-1})\acton\Xi_2-\{h_1^{-1},h_2^{-1}\nabla h_2\},\ell_2(h_2\acton \ell_1))~.
			\end{aligned}
		\end{equation}
		This composition is associative and unital, and it satisfies the interchange rule~\eqref{eq:interchange_rule}.
	\end{lemma}
	\begin{proof}
		The form of the horizontal composition is read off the defining diagram for the $h$-component of 2-morphisms. One then verifies that the sources and targets also match for the $\Xi$-component.
		
		Both associativity and unitality are evident by inspection.
		
		The interchange rule~\eqref{eq:interchange_rule} is implied by the fact that this condition is satisfied int the Gray tricategory $\sfB\caG$, as explained in \ref{app:2crossedmodules}, together with the fact that sources and targets match.
	\end{proof}
	
	\begin{lemma}
		Given composable 1-morphisms $(\ldots,g_2,\Lambda_2,\Sigma_2)$ and $(\ldots,g_1,\Lambda_1,\Sigma_1)$ and a 3-morphism $(\ldots,g_1,\Lambda_1,\Sigma_1,h_1,\Xi_1,\ell_1)$, left-whiskering is uniquely defined as 
		\begin{equation}
			\begin{aligned}
				&L_{(\ldots,g_2,\Lambda_2,\Sigma_2)}(\ldots,g_1,\Lambda_1,\Sigma_1,h_1,\Xi_1,\ell_1)
				\\
				&\hspace{1cm}\coloneqq(\ldots,g_1g_2,g_2^{-1}\acton \Lambda_1+\Lambda_2,\{\Lambda_2,g_2^{-1}\acton \Lambda_1\}+g_2^{-1}\acton \Sigma_1+\Sigma_2,h_1,\Xi_1,\ell_1)~.
			\end{aligned}
		\end{equation}
		Given composable 1-morphisms $(\ldots,g_2,\Lambda_2,\Sigma_2)$ and $(\ldots,g_1,\Lambda_1,\Sigma_1)$ and a 3-morphism $(\ldots,g_2,\Lambda_2,\Sigma_2,h_2,\Xi_2,\ell_2)$, right-whiskering is uniquely defined as 
		\begin{equation}
			\begin{aligned}
				&R_{(\ldots,g_1,\Lambda_1,\Sigma_1)}(\ldots,g_2,\Lambda_2,\Sigma_2,h_2,\Xi_2,\ell_2)
				\\
				&\hspace{1cm}\coloneqq(\ldots,g_1g_2,g_2^{-1}\acton \Lambda_1+\Lambda_2,\{\Lambda_2,g_2^{-1}\acton \Lambda_1\}+g_2^{-1}\acton \Sigma_1+\Sigma_2,
				\\
				&\hspace{2cm}g_1\acton h_2,g_1\acton (\Xi_2-\{\Lambda_1,h_2\}-\{h_2,\Lambda_1\}),g_1\acton \ell_2)~.
			\end{aligned}
		\end{equation}
		The consistent definition of right-whiskering requires 
		\begin{equation}\label{eq:kappa_3_condition}
			g\acton \kappa_3(h,F)=\kappa_3(g\acton h,\nu_2(g,F))+\{\kappa_1(g,F),g\acton h^{-1}\}+\{g\acton h^{-1},\kappa_1(g,F)\}~.
		\end{equation}
	\end{lemma}
	\begin{proof}
		The definitions of left- and right-whiskering are fixed by considering the composition of 1-morphisms, as inspection shows. Consistency of right-whiskering in the $\Sigma$-component then implies relation~\eqref{eq:kappa_3_condition} after a lengthy but straightforward computation.
		
		The relevant compatibility relations for left- and right-whiskering, \eqref{eq:whiskering_compatibility}, are evidently satisfied.
		
		Finally, left-whiskering evidently preserves vertical and transversal composition of 3-morphisms. A short computation then shows that this is also true for right-whiskering.
	\end{proof}
	
	The definition of the interchangor is now a straightforward generalization of the interchangor~\eqref{eq:interchangor_2_xm} in a 2-crossed module regarded as a one-object 3-groupoid. Consider two 2-morphisms
	\begin{equation}
		(\ldots,g_1,\Lambda_1,\Sigma_1,h_1,\Xi_1)
		\eand
		(\ldots,g_2,\Lambda_2,\Sigma_2,h_2,\Xi_2)~.
	\end{equation}
	Then there is a 3-morphism
	\begin{equation}
		\begin{aligned}
			I(\ldots,h_1,\Xi_1;\ldots,h_2,\Xi_2)&:
			R_{\sft(\ldots,h_1,\Xi_1)}(\ldots,h_2,\Xi_2)\circ L_{(\ldots,g_2,\Lambda_2,\Sigma_2)}(\ldots,h_1,\Xi_1)
			\\
			&\hspace{0.5cm}\rightarrow L_{\sfs(\ldots,h_2,\Xi_2)}(\ldots,h_1,\Xi_1)\circ R_{(\ldots,h_1,\Xi_1)}(\ldots,h_2,\Xi_2) ~,
		\end{aligned}
	\end{equation}
	where $\sfs$ and $\sft$ denote source and target, respectively, with 
	\begin{equation}
		I(\ldots,h_1,\Xi_1;\ldots,h_2,\Xi_2)=(\ldots,\{h_1,g_1\acton h_2\})~,
	\end{equation}
	which can be verified by lengthy but straightforward computation.
	
	Because of the above definition, because the $\sfL$-valued component of this interchangor is the same as that for the interchangor~\eqref{eq:interchangor_2_xm} in a 2-crossed module, and because the latter satisfies the relevant coherence axioms, it directly follows that also the former interchangor satisfies the relevant coherence axioms.
    
    \subsection{Adjusted 2-crossed modules and integration theorem}
    
    The conditions on the adjustment maps obtained above, including the conditions $\caC_{1,2,3,4}$ listed in \ref{app:raw_adjustment_condition}, can be distilled into the following definition.
    \begin{definition}[Adjusted 2-crossed module of Lie groups]\label{def:adjusted-2XM}
        An adjusted 2-crossed module of Lie groups is a 2-crossed module of Lie groups $\caG=(\sfL\xrightarrow{\sft}\sfH\xrightarrow{\sft}\sfG)$ together with maps
        \begin{equation}
            \kappa_1:\sfG\times \frg\rightarrow \frh~,~~~\kappa_2:\sfG\times \frh\rightarrow \frl~,~~~\kappa_3:\sfH\times \frg\rightarrow \frl~,
        \end{equation}
        which are unital-linear\footnote{i.e., they vanish if the group-valued argument is $\unit$, and they are linear in their Lie-algebra-valued arguments} and        
        satisfy the following relations:
        \begin{subequations}\label{eq:final_adjustment_relations}
            \begin{equation}\label{eq:cond_k1_1}
                \begin{aligned}
                    \kappa_1(\sft(h),X)&=\rmAd^{\acton}_{h}(X)-\sft\left(\kappa_3(h^{-1},X)\right)~,
                \end{aligned}
            \end{equation}
            \begin{equation}\label{eq:cond_k1_2}
                \begin{aligned}
                    g_1\acton \kappa_1(g_2,X)&=-\kappa_1\left(g_1,\nu_2(g_2,X)\right)+\kappa_1(g_1 g_2,X)~,
                \end{aligned}
            \end{equation}
            \begin{equation}
                \begin{aligned}\label{eq:cond_k2_1}
                    \kappa_2(\sft(h),Y)&=\left\{Y,h\right\}+\left\{h,Y\right\}-\kappa_3\left(h^{-1},\sft(Y)\right)~,
                \end{aligned}
            \end{equation}
            \begin{equation}\label{eq:cond_k2_2}
                \begin{aligned}
                    g_1\acton \kappa_2(g_2,Y)&=\kappa_2(g_1 g_2,Y)-\kappa_2(g_1,g_2\acton Y)-\kappa_2(g_1,\sft(\kappa_2(g_2,Y))-\kappa_1(g_2,\sft(Y)))~,
                \end{aligned}
            \end{equation}
            \begin{equation}\label{eq:cond_k3_1}
                \begin{aligned}
                    \kappa_3(\sft(\ell),X)&=\rmAd^{\acton}_{\ell^{-1}}(X)~,
                \end{aligned}
            \end{equation}
            \begin{equation}\label{eq:cond_k3_2}
                \begin{aligned}
                    \kappa_3(h_1h_2,X)&=\{h_2^{-1},\rmAd^\acton_{h_1^{-1}}X\}+\sft(h_2^{-1})\acton \kappa_3(h_1,X)+\kappa_3(h_2,X)~,
                \end{aligned}
            \end{equation}
            \begin{equation}\label{eq:cond_k3_3}
                g\acton \kappa_3(h,X)=\kappa_3(g\acton h,\nu_2(g,X))+\{g\acton h^{-1},\kappa_1(g,X)\}+\{\kappa_1(g,X),g\acton h^{-1}\}~,
            \end{equation}
        \end{subequations}
        where we introduced the useful shorthand
        \begin{equation}
            \begin{aligned}
                \nu_2(g,X)&\coloneqq \rmAd_g(X)-\sft(\kappa_1(g,X))~,
            \end{aligned}
        \end{equation}
        and denoted linearizations of maps by the same symbols as the original maps.
    \end{definition}
    Correspondingly, we have the following proposition.
    \begin{proposition}
        The consistency conditions for a BRST Lie 3-groupoid obtained in the previous sections, i.e., the conditions on the adjustment maps $R_{1,2},\kappa_{1,2,3,4}$ are equivalent to the data of an adjusted 2-crossed module.
    \end{proposition}
    \begin{proof}
        We first show that $\kappa_4$ vanishes. A direct approach\footnote{One can also work without this assumption, but the required computation is then longer.} is to assume that $\kappa_4$ is analytic. Applying the derivation rule following from \eqref{eq:der_k4} in $\caC_1$ and considering the terms proportional to $g_2^{-1}\rmd g_2$ yields
        \begin{equation}
            \kappa_4(g,X_1,X_2)=0
        \end{equation} 
        for all $g\in \sfG$ and $X_{1,2}\in \frg$. Plugging this back into $\caC_1$ and considering the terms proportional to $g_1^{-1}\rmd g_1$ yields
        \begin{equation}
            \kappa_4(X_1,g,X_2)=0
        \end{equation} 
        for all $g\in \sfG$ and $X_{1,2}\in \frg$. Since all linearizations of $\kappa_4$ are trivial, it follows that $\kappa_4$ itself is trivial.
        
        We have already established that the maps $\kappa_{1,2,3}$ are linear-unital in~\eqref{eq:der_k1_k2_unitality} and \eqref{eq:der_k3_unitality}. Furthermore, we have already derived the conditions~\eqref{eq:cond_k1_1}, \eqref{eq:cond_k1_2}, and \eqref{eq:cond_k3_3} in equations~\eqref{eq:der_k1_1}, \eqref{eq:adj_cond_1_1}, and~\eqref{eq:kappa_3_condition}, respectively.
        
        We continue by considering condition $\caC_2$ for $\Lambda=0$ and $g=\unit$. Applying the obvious derivation rule $\rmd \kappa_3(h,X)=\kappa_3(h,\rmd X)+\ldots$, where $\ldots$ represents terms involving $h^{-1}\rmd h$ which are irrelevant here, together with the Bianchi identity for $F$, cf.~\eqref{eq:Bianchi_identities}, then produces~\eqref{eq:cond_k2_1}.
        
        Similarly, we can consider condition $\caC_1$ for $\kappa_4=0$, and applying the derivation rule for $\kappa_1$, as well as the Bianchi identity for $F$. The part proportional to $H$ then yields the adjustment condition~\eqref{eq:cond_k2_2}.
        
        Next, condition $\caC_4$ for $g=\unit$ and $h=\unit$ yields~\eqref{eq:cond_k3_1}, and $\caC_3$ for $g=\unit$ yields~\eqref{eq:cond_k3_2}. Conversely, the conditions obtained so far already imply both conditions $\caC_3$ and $\caC_4$.
        
        Considering then the terms in $\caC_1$ linear in $\Lambda_2$ for $g_2=\unit$ then results in the identity
        \begin{equation}
            R_1(g,Y,X)=R_1(\unit,Y,\nu_2(g^{-1},X))~,
        \end{equation} 
        and we identify
        \begin{equation}
            R_1(\unit,Y,X)=\kappa_3'(Y,X)~.
        \end{equation}
        for some function $\kappa_3'$. 
        
        We then consider the component proportional to $\Lambda_1$ in $\caC_1$, with all possible simplifications, obtaining that $\kappa_3'$ is indeed the linearization of $\kappa_3$,
        \begin{equation}
            \kappa_3'(Y,X)=\kappa_3(Y,X)~.
        \end{equation} 
        Using this to simplify $\caC_1$ further yields
        \begin{equation}
            g\acton \kappa_3(Y,X)=\kappa_3(g\acton Y,\nu_2(g,X))-\{g\acton Y,\kappa_1(g,X)\}-\{\kappa_1(g,X),g\acton Y\}~,
        \end{equation}
        which is a linearization of condition~\eqref{eq:kappa_3_condition}. The conditions above then imply both $\caC_1$ and, using another linearization of~\eqref{eq:kappa_3_condition}, also $\caC_2$.

        This completes the equivalence between adjustment conditions and conditions for a consistent BRST Lie 3-groupoid.
    \end{proof}

	We can now conclude the first part of this paper with the following theorem:
	\begin{theorem}\label{thm:integration}
		Let $(\caG,\kappa)$ denote an adjusted 2-crossed module $\caG$. Then $\scB(U,\caG)$ forms indeed a 3-groupoid, the BRST Lie 3-groupoid for $U$ and $\caG$. Furthermore, $\scB(U,\caG)$ differentiates to the BRST Lie 3-algebroid $\frA(U,\caG)$.
	\end{theorem}
	Let us clarify this statement. For a general rigorous definition of Lie differentiation of higher groupoids, we can turn to Ševera~\cite{Severa:2006aa}, see also e.g.~\cite{Jurco:2016qwv} and~\cite{Li:2014}. Here, higher Lie algebroids are seen as 1-jets of higher Lie groupoids. However, since we are dealing with strict action Lie 3-groupoid and corresponding strict action Lie 3-algebroids, we can linearize the individual actions, i.e., the gauge and higher gauge transformations in $\scB(U,\caG)$ and compare the result to the corresponding infinitesimal actions in $\frA(U,\caG)$. Direct inspection shows that this is indeed the case.
    
    \section{Differential cohomology for principal 3-bundles}\label{sec:diff_cohomology}
    
    \subsection{Preliminaries}
    
    The next step is the stackification of the BRST Lie 3-groupoid into the differential cocycles defining a principal 3-bundle with connection. That is, local gauge potential data is glued together on overlaps by gauge transformations, which, in turn, are glued together on triple overlaps by composition and higher gauge transformations, etc. 
    
    The differential cocycles of a principal $\caG$-bundle (without adjustment) have been developed to a large extent in~\cite{Saemann:2013pca}.\footnote{See also e.g.~\cite{Jurco:2009px} for the contained Čech cocycles and~\cite{Martins:2009aa} for the definition of holonomies based on such principal 3-bundles, as well as~\cite{Breen:1994aa} for related, earlier work.} However, some cocycle relations and most coboundary relations as well as all higher coboundary relations were missing, and our construction below also fills this gap.
    
    In the following, we always consider a principal $\caG$-bundle $\caP$ over a manifold $M$, where $\caG$ is a 2-crossed module of Lie groups $(\sfL \rightarrow \sfH \rightarrow \sfG)$ with the corresponding 2-crossed module of Lie algebras $(\frl \rightarrow \frh \rightarrow \frg)$. 
    
    The principal $\caG$-bundle $\caP$ will be described subordinate to a surjective submersion $\sigma:Y \rightarrow M$. That is, all components of cocycles will be differential forms on one of the fibered products
    \begin{equation}
        Y^{[p]}\coloneqq\{ (y_2,\ldots,y_p) \in Y^p~|~\sigma(y_1)=\ldots=\sigma(y_p)\}~.
    \end{equation} 
    The 3-groupoid of general principal 3-bundles is obtained as the colimit over the directed set of covers.
    
    For a sheaf $\caS$, we denote $\caS$-valued $q$-forms on $Y^{[p]}$ by 
    \begin{equation}
        C^{p,q}(\sigma,\caS)\coloneqq \Omega^q(Y^{[p]},\caS)~,
    \end{equation} 
    and we usually denote Lie groups and Lie algebras and the corresponding sheaf of Lie-group- or Lie-algebra-valued functions by the same letter.
    
    We note that the group-valued components of gauge and higher gauge transformations will automatically form the Čech cocycles of a principal 3-bundle. The latter can be directly derived e.g.~by using weak 3-functors between the Čech groupoid and the delooping trigroupoid of $\caG$. Much simpler would be to consider the simplicial Lie group corresponding to $\caG$ and to describe a principal 3-bundle as a twisted Cartesian product, cf.~\cite{May:book:1967} for the general discussion. This description, however, contains a number of redundancies that would make the computation rather lengthy. 
    
    In the following, we merely list the results. Gluing constraints are still computed by combining strict composition with higher transformations. Again, these computations are in principle straightforward, but quite lengthy and therefore best done with a computer algebra program.
    
    \subsection{Differential cocycles}
    
    A differential cocycle of a principal $\caG$-bundle subordinate to a surjective submersion $\sigma:Y\rightarrow M$ consists of components    
    \begin{subequations}\label{eq:conventional_cocycles}
        \begin{equation}
            \begin{aligned}
                &\{\ell_{abcd}\}\in C^{3,0}(\sigma,\sfL)~,&&\{\Xi_{abc}\}\in C^{2,1}(\sigma,\frl)~,~~\{\Sigma_{ab}\}\in C^{1,2}(\sigma,\frl)~,~~\{C_a\}\in C^{0,3}(\sigma,\frl)~,\\
                &\{h_{abc}\}\in C^{2,0}(\sigma,\sfH)~,&&\{\Lambda_{ab}\}\in C^{1,1}(\sigma,\frh)~,~~~\{B_a\}\in C^{0,2}(\sigma,\frh)~,\\
                &\{g_{ab}\}\in C^{1,0}(\sigma,\sfG)~,&&\{A_a\}\in C^{0,1}(\sigma,\frg)~,
            \end{aligned}
        \end{equation}
        which satisfy a number of conditions. For the Čech part $(\ell_{abcd},h_{abc},g_{ab})$ of the cocycle, the relations read
        \begin{equation}
            \begin{gathered}
                \sft(h_{abc})g_{ab}g_{bc}=g_{ac}~,
                \\
                h_{acd}h_{abc}\sft(\ell_{abcd})=h_{abd}(g_{ab}\acton h_{bcd})~,
                \\
                \ell_{abcd}((g_{ab}\acton h^{-1}_{bcd})\acton\ell_{abde})(g_{ab}\acton\ell_{bcde})\hspace{4.5cm}
                \\
                \hspace{4.5cm}=\ (h^{-1}_{abc}\acton\ell_{acde})\{h^{-1}_{abc},g_{ac}\acton h^{-1}_{cde}\}((g_{ab}g_{bc}\acton h^{-1}_{cde})\acton \ell_{abce})~,
            \end{gathered}
        \end{equation}
        and these are already found in the literature, see e.g.~\cite{Jurco:2009px,Saemann:2013pca}. The local gauge potentials are now glued together by gauge transformations on $Y^{[2]}$:
        \begin{equation}
            \begin{aligned}
                A_b&=g^{-1}_{ab}A_{a} g_{ab}+g^{-1}_{ab}\rmd g_{ab}-\sft(\Lambda_{ab})~,
                \\
                B_b&=g^{-1}_{ab}\acton B_a+\nabla_b\Lambda_{ab}+\tfrac12[\Lambda_{ab}, \Lambda_{ab}]-\sft(\Sigma_{ab})-\kappa_1(g_{ab}^{-1},F_a)~,
                \\
                C_b&=g^{-1}_{ab}\acton C_a+\nabla_{b}\Sigma_{ab}-\{\sft(\Sigma_{ab}),\Lambda_{ab}\}-\{\Lambda_{ab},\sft(\Sigma_{ab})\}-\{B_{b},\Lambda_{ab}\}
                \\
                &\hspace{1cm} -\{\Lambda_{ab},B_{b}\}+\{\Lambda_{ab},\nabla_b\Lambda_{ab}\}+\tfrac12\{\Lambda_{ab},[\Lambda_{ab},\Lambda_{ab}]\}
                \\
                &\hspace{1cm} +\kappa_2(g_{ab}^{-1},H_a)-\kappa_3(\Lambda_{ab},\nu_2(g_{ab},F_a))~.
            \end{aligned}
        \end{equation}
        On $Y^{[3]}$ we combine composition of 1-morphisms in the BRST groupoid with 2-morphisms, and the result is:
        \begin{equation}
            \begin{aligned}
                \Lambda_{ac}&=\Lambda_{bc}+g_{bc}^{-1}\acton\Lambda_{ab}-g_{ac}^{-1}\acton(h_{abc}\nabla_ah_{abc}^{-1})+\sft(g_{ac}^{-1}\acton \Xi_{abc})~,
                \\
                \Sigma_{ac} &= \Sigma_{bc}+g_{bc}^{-1}\acton\Sigma_{ab}+g_{ac}^{-1}\acton \nabla_a\Xi_{abc}-\tfrac12g_{ac}^{-1}\acton[\Xi_{abc},\Xi_{abc}]+\{\Lambda_{bc},g_{bc}^{-1}\acton \Lambda_{ab}\}
                \\
                &\hspace{0.5cm}+\{\sft(g_{ac}^{-1}\acton \Xi_{abc}),g_{ac}^{-1}\acton (h_{abc}\nabla_a h_{abc}^{-1})\}+g_{ac}^{-1}\acton(\{B_a,h_{abc}\}+\{h_{abc},B_a\})
                \\
                &\hspace{0.5cm}
                -\{\Lambda_{ac},\sft(g_{ac}^{-1}\acton \Xi_{abc})+g_{ac}^{-1}\acton (h_{abc}\nabla_a h_{abc}^{-1})\}+g_{ac}^{-1}\acton \kappa_3(h_{abc}^{-1},F_a)~,
            \end{aligned}
        \end{equation}
        Finally, on $Y^{[4]}$, we combine composition on the faces of a tetrahedron with a 3-morphism, obtaining
        \begin{equation}
            \begin{aligned}
                \Xi_{abd}&=\sft(h_{acd})\acton \Xi_{abc}+\Xi_{acd}-(g_{ad}g_{bd}^{-1})\acton\big(\Xi_{bcd}-\{h_{bcd},\Lambda_{ab}\}+\{\Lambda_{ab},h_{bcd}\}\big)
                \\
                &\hspace{1cm}-\sft(h_{abd})\acton (h_{abd}^{-1}h_{acd}h_{abc})\acton \left(\ell_{abcd}\nabla_a \ell_{abcd}^{-1}\right)
                \\
                &\hspace{1cm}+\left\{h_{abd},(h_{abd}^{-1}h_{acd}h_{abc})\nabla_a(h^{-1}_{abc}h_{acd}^{-1}h_{abd})\right\}+\left\{h_{acd},h_{abc} \nabla_ah_{abc}^{-1}\right\}~.
            \end{aligned}
        \end{equation}
    \end{subequations}
    Here, we used again the shorthand notation
    \begin{equation}
        \nabla_a \alpha\coloneqq\rmd \alpha+A_a\acton \alpha\eand f\nabla_a f^{-1}\coloneqq f\rmd f^{-1}+\rmAd^\acton(f,A_a)
    \end{equation} 
    for $\alpha$ taking values in $\frh$ or $\frl$ and $f$ taking values in $\sfH$ or $\sfL$ with $f\rmd f^{-1}$ the evident Maurer--Cartan form and 
    \begin{equation}
        \rmAd^\acton(f,A_a)\coloneqq f (A_a\acton f^{-1})~.
    \end{equation} 
    
    The curvatures are defined as 
    \begin{equation}
        \begin{aligned}
            F_a&\coloneqq \rmd A_a+\tfrac12[A_a,A_a]+\sft(B_a)~,
            \\
            H_a&\coloneqq \rmd B_a+A_a\acton B_a+\sft(C_a)-\kappa_1(A_a,F_a)~,
            \\
            G_a&\coloneqq \nabla_a C_a+\{B_a,B_a\}+\kappa_2(A_a,H_a)+\kappa_3(B_a,F_a)~,
        \end{aligned}
    \end{equation}
    and they are related over $Y^{[2]}$ as follows:
    \begin{equation}
        \begin{aligned}
            F_b&=\nu_2(g_{ab}^{-1},F_a)~,
            \\
            H_b&=g_{ab}^{-1}\acton H_a-\kappa_1(g_{ab}^{-1},\sft(H_a))+\sft(\kappa_2(g_{ab}^{-1},H_a))~.
            \\
            G_b&=g_{ab}^{-1}\acton G_a+\kappa_2(g_{ab}^{-1},\sft(G_a))-\kappa_3(\kappa_1(g_{ab}^{-1},F_a), \nu_2(g_{ab}^{-1},F_a))
            \\
            &\hspace{1cm}-\kappa_2(g_{ab}^{-1},\kappa_1(F_a,F_a))+\{\kappa_1(g_{ab}^{-1},F_a),\kappa_1(g_{ab}^{-1},F_a)\}~.
        \end{aligned}
    \end{equation}
    
    The Bianchi identities are variants of those found in~\eqref{eq:Bianchi_identities},
    \begin{equation}
        \begin{aligned}
            \nabla_a F_a-\sft(\kappa_1(A_a,F_a))-\sft(H_a)&=0~,
            \\
            \nabla_a H_a-\kappa_1(A_a,\sft(H_a))+\sft(\kappa_2(A_a,H_a))+\kappa_1(F_a,F_a)-\sft(G_a)&=0~,
            \\
            \nabla_aG_a+\kappa_2\big(A_a,\sft(G_a)-\kappa_1(F_a,F_a)\big)-\kappa_2(F,H)-\kappa_3(H_a+\kappa_1(A_a,F_a),F_a)&=0~.
        \end{aligned}
    \end{equation}
    
    \subsection{Differential 1-coboundaries}
    
    Consider two differential cocycles $(\ell_{abcd},h_{abc},g_{ab},\ldots,A_a,B_a,C_a)$ and $(\tilde\ell_{abcd},\tilde h_{abc},\tilde g_{ab}, \ldots, $ $\tilde A_a,\tilde B_a,\tilde C_a)$. We call these equivalent if they are related by a coboundary consisting of components
    \begin{subequations}\label{eq:conventional_coboundaries}
        \begin{equation}
            \begin{aligned}
                &\{\ell_{abc}\}\in C^{2,0}(\sigma,\sfL)~,&&\{\Xi_{ab}\}\in C^{1,1}(\sigma,\frl)~,~~\{\Sigma_{a}\}\in C^{0,2}(\sigma,\frl)~,\\
                &\{h_{ab}\}\in C^{1,0}(\sigma,\sfH)~,&&\{\Lambda_{a}\}\in C^{0,1}(\sigma,\frh)~,\\
                &\{g_{a}\}\in C^{0,0}(\sigma,\sfG)
            \end{aligned}
        \end{equation}
        and they satisfy
        \begin{equation}
            \begin{aligned}
                \tilde g_{ab}&=g_a^{-1}\sft(h_{ab})g_{ab}g_b~,
                \\
                \tilde h_{abc}&=g_a^{-1}\acton(h_{ac} h_{abc}\sft(\ell_{abc})(g_{ab}\acton h^{-1}_{bc})h^{-1}_{ab})~,
                \\
                \tilde \ell_{abcd}&=g_{a}^{-1}\acton \Big((h_{ab}\,g_{ab}\acton h_{bc})\acton \left(\ell^{-1}_{abc}\times h_{abc}^{-1}\acton \left\{h_{abc},(g_{ab}g_{bc})\acton h_{cd}\right\}^{-1}\right.
                \\
                &\hspace{2cm}\times((g_{ab}g_{bc})\acton h_{cd})\acton \left((h_{abc}^{-1}\acton \ell^{-1}_{acd})\ell_{abcd}\right.
                \\
                &\hspace{2cm}\times \left(g_{ab}\acton h_{bcd}^{-1}\right)\acton \left(\ell_{abd}\, ((g_{ab}\acton h_{bd}^{-1})h_{ab}^{-1})\acton \left\{h_{ab},g_{ab}\acton h_{bd}\right\}^{-1}\right.
                \\
                &\hspace{5cm}\left.\times h_{ab}^{-1}\acton \left\{h_{ab},g_{ab}\acton h_{bcd}\right\}^{-1}\right)
                \\
                &\hspace{2cm}\left.\times \left(h_{ab}^{-1}\acton \left\{h_{ab},\sft(g_{ab}\acton \ell_{bcd})\right\}^{-1}\right)(g_{ab}\acton \ell_{bcd})(h_{ab}^{-1}\acton\left\{h_{ab},(g_{ab}g_{bc})\acton h_{cd}^{-1}\right\}^{-1}\right)
                \\
                &\hspace{2cm}\left.\times h_{ab}^{-1}\acton \left\{h_{ab},g_{ab}\acton h_{bc}^{-1}\right\}^{-1}\right)\Big)~,
            \end{aligned}
        \end{equation}
        for the Čech cocycles, which is again known from the literature. Gauge transformations of local potential forms are then identified with the 1-morphisms of the BRST Lie groupoid:        
        \begin{equation}
            \begin{aligned}
                \tilde A_a&=g^{-1}_{a}A_{a} g_{a}+g^{-1}_{a}\rmd g_{a}-\sft(\Lambda_{a})~,
                \\
                \tilde B_a&=g^{-1}_{a}\acton B_a+\tilde\nabla_a\Lambda_{a}+\tfrac12[\Lambda_{a}, \Lambda_{a}]-\sft(\Sigma_{a})-\kappa_1(g_{a}^{-1},F_a)~,
                \\
                \tilde C_a&=g^{-1}_{a}\acton C_a-\tilde \nabla_{a}\Sigma_{a}-\{\sft(\Sigma_{a}),\Lambda_{a}\}-\{\Lambda_{a},\sft(\Sigma_{a})\}-\{\tilde B_{a},\Lambda_{a}\}
                \\
                &\hspace{1cm} -\{\Lambda_{a},\tilde B_{a}\}+\{\Lambda_{a},\tilde \nabla_a\Lambda_{a}\}+\tfrac12\{\Lambda_{a},[\Lambda_{a},\Lambda_{a}]\}~,
                \\
                &\hspace{1cm} +\kappa_2(g_{a}^{-1},H_a)-\kappa_3(\Lambda_{a},\nu_2(g_{a},F_a))~.
            \end{aligned}
        \end{equation}
        To derive the gauge transformations of $\Lambda_{ab}$ and $\Sigma_{ab}$, one doubles the patches $a$ and $b$ and considers the composition of four such 1-morphisms, with the 1-morphisms on doubled patches given by $\Lambda_{a}$ and $\Lambda_b$:
        \begin{equation}
            \begin{aligned}
                \tilde \Lambda_{ab}&=\Lambda_b+g_b^{-1}\acton \Lambda_{ab}-\tilde g_{ab}^{-1}\acton \Lambda_a+(g_b^{-1}g_{ab}^{-1})\acton (h_{ab}^{-1} \nabla_a h_{ab})+\sft(g_a^{-1}\acton \Xi_{ab})~,
                \\
                \tilde \Sigma_{ab}&=g_{b}^{-1}\acton \Sigma_{ab}+g^{-1}_a\acton \nabla_a\Xi_{ab}+\{\sft(g^{-1}_a\acton \Xi_{ab}),\tilde \Lambda_{ab}\}-\tfrac{1}{2} g^{-1}_a\acton [\Xi_{ab},\Xi_{ab}]
                \\
                &\hspace{1cm}
                -\tilde g_{ab}^{-1}\acton \left(\Sigma_a-\left\{\Lambda_a,\Lambda_a\right\}\right)+\Sigma_b+\left\{\Lambda_b,g_{b}^{-1}\acton \Lambda_{ab}\right\}
                \\
                &\hspace{1cm}-\left\{g_{b}^{-1}\acton \Lambda_{ab}+\Lambda_b,\left(\tilde g_{ab}^{-1}g_{a}^{-1}\right)\acton (h_{ab}\nabla_a h_{ab}^{-1})+\tilde g_{ab}^{-1}\acton \Lambda_a\right\}
                \\
                &\hspace{1cm}+(\tilde g_{ab}^{-1}g_{a}^{-1})\acton \left(\left\{B_{a},h_{ab}\right\}+\left\{h_{ab},B_{a}\right\}-\left\{(h_{ab}\nabla_a h_{ab}^{-1}),(h_{ab}\nabla_a h_{ab}^{-1})+g_a\acton \Lambda_a\right\}\right)~,
            \end{aligned}
        \end{equation}
        For the triple overlaps, one proceeds analogously, leading to 
        \begin{equation}
            \begin{aligned}
                &\tilde \Xi_{abc}=\tilde g_{ac}\acton((g_c^{-1}g_{ac}^{-1})\acton \Xi_{abc}+g_a^{-1}\acton \Xi_{ac}-g_b^{-1}\acton \Xi_{bc}-\tilde g_{bc}^{-1}\acton g_a^{-1}\acton \Xi_{ab})
                \\
                &\hspace{0.5cm}
                +(\tilde g_{ac}g_c^{-1}g_{bc}^{-1})\acton(\{\Lambda_{ab},h_{bc}^{-1}\}+\{h_{bc}^{-1},\Lambda_{ab}\})
                \\
                &\hspace{0.5cm}-\{\tilde h_{abc},\Lambda_a-(g_a^{-1}\sft(h_{ab}))\acton h_{ab}^{-1}\nabla_a h_{ab}\}+\{\Lambda_a,\tilde h_{abc}\}+\tilde h_{abc}\acton g_a^{-1}\acton \{h_{ab},h_{ab}^{-1}\nabla_a h_{ab}\}
                \\
                &\hspace{0.5cm}+(g_a^{-1})\acton\Big(-\{h_{ac},h_{abc}\nabla_a h_{abc}^{-1}+h_{ac}^{-1}\nabla_a h_{ac}^{-1}\}
                \\
                &\hspace{3.5cm}+\{h_{ac}h_{abc},g_{ab}\acton (h_{bc}^{-1}\nabla_a h_{bc}^{-1})-g_{ab}\acton (h_{bc}^{-1} (g_{ab}^{-1}\nabla_a g_{ab})\acton h_{bc})\}\Big)
                \\
                &\hspace{0.5cm}+(g_a^{-1})\acton(h_{ac}h_{abc})\acton\Big(\ell_{abc}(g_{ab}\acton (h_{bc}^{-1}\nabla_a h_{bc}-h_{bc}^{-1}(g_{ab}^{-1}\nabla_a g_{ab})\acton h_{bc}))\ell_{abc}^{-1}-\ell_{abc}\nabla_a\ell_{abc}^{-1}\Big)
            \end{aligned}
        \end{equation}
        for the transition forms.
    \end{subequations}
    
    \subsection{Differential 2-coboundaries}
    
    Two coboundaries $(\ell_{abc},h_{ab},g_{a},\Xi_{ab},\Lambda_{a},\Sigma_{a})$ and $(\tilde\ell_{abc},\tilde h_{ab},\tilde g_{a},\tilde\Xi_{ab},\tilde\Lambda_{a},\tilde\Sigma_{a})$ are equivalent if they are related by a 2-coboundary consisting of components 
    \begin{subequations}
        \begin{equation}
            \begin{aligned}
                &\{\ell_{ab}\}\in C^{1,0}(\sigma,\sfL)~,&&\{\Xi_{a}\}\in C^{0,1}(\sigma,\frl)~,\\
                &\{h_{a}\}\in C^{0,0}(\sigma,\sfH)~,
            \end{aligned}
        \end{equation}
        and they satisfy
        \begin{equation}
            \begin{aligned}
                \tilde g_a&=\sft(h_a)g_a~,
                \\
                \tilde h_{ab}&=h_a h_{ab}\sft(\ell_{ab})(g_{ab}\acton h_b^{-1})~,
                \\
                \tilde \ell_{abc}&=(h_{abc}^{-1}\acton\{h_{abc},(g_{ab}g_{bc})\acton h_c\}^{-1})
                \\
                &\hspace{0.5cm}\times((g_{ab}g_{bc})\acton h_c)\acton \Big( h_{abc}^{-1}\acton\Big(
                \ell_{ac}^{-1}((h_{ac}^{-1}h_a^{-1})\acton\{h_a,h_{ac}\}^{-1})
                (h_a^{-1}\acton\{h_a,h_{abc}\}^{-1})\Big)
                \\
                &\hspace{3.8cm}\times (h_a^{-1}\acton \{h_a,\sft(\ell_{abc})\}^{-1})\ell_{abc}
                \\
                &\hspace{3.8cm}\times (g_{ab}\acton h_{bc}^{-1})\acton\Big(
                (h_a^{-1}\acton\{h_a,g_{ab}h_{bc}\})
                ((h_{ab}^{-1}h_a^{-1})\acton\{h_a,h_{ab}\})
                \ell_{ab}
                \Big)
                \\
                &\hspace{3.8cm}\times g_{ab}\acton \ell_{bc}\,\Big)~,
            \end{aligned}
        \end{equation}
        and the obvious higher gauge transformations
        \begin{equation}
            \begin{aligned}
                \tilde \Lambda_a&=\Lambda_a+g_a^{-1}\acton (h_a^{-1}\nabla_a h_a)+\sft(g_a^{-1}\acton\Xi_a)~,
                \\
                \tilde \Sigma_a&=\Sigma_a+g^{-1}_a\acton (\nabla_a \Xi_a+\tfrac12[\Xi_a,\Xi_a]+\{\sft(\Xi_a),\Lambda_a\}-\{h^{-1}_a,B_a\}-\{B_a,h^{-1}_a\})
                \\
                &\hspace{1cm}+\{\Lambda_a,g^{-1}_a\acton (h^{-1}_a\nabla_a h_a+\sft(\Xi_a))\}+g_a^{-1}\acton \kappa_3(h_a,F_a)~.
            \end{aligned}
        \end{equation}
        On $Y^{[2]}$, we double again the local patches and consider composition of four 2-morphisms, leading to 
        \begin{equation}
            \begin{aligned}
                \tilde \Xi_{ab}&=\Xi_{ab}+g_a^{-1}\acton\Xi_a-g_b^{-1}\acton\Xi_b+g_a^{-1}\acton(\{\Lambda_{ab},h_b^{-1}\}+\{h_b^{-1},\Lambda_{ab}\})
                \\
                &\hspace{0.5cm} +(g_a)^{-1}\acton\Big(-\ell_{ab}^{-1}\nabla_a \ell_{ab}-\ell_{ab}^{-1}((h_ah_{ab})^{-1}\nabla (h_ah_{ab}))\acton \ell_{ab}
                \\
                &\hspace{3cm}-\{h_{ab}^{-1},h_a^{-1}\nabla_a h_a\}+\{h_b^{-1},g_{ab}^{-1}\acton (\tilde h_{ab}^{-1}\nabla_a \tilde h_{ab})\}\Big)~.
            \end{aligned}
        \end{equation}
    \end{subequations}
    
    \subsection{Differential 3-coboundaries}
    
    Finally, two 2-coboundaries $(\ell_{ab},h_{a},\Xi_{a})$ and $(\tilde\ell_{ab},\tilde h_{a},\tilde\Xi_{a})$ are equivalent if they are related by a 3-coboundary consisting of components 
    \begin{subequations}
        \begin{equation}
            \{\ell_{a}\}\in C^{0,0}(\sigma,\sfL)~,
        \end{equation}
        and satisfying
        \begin{equation}
            \begin{aligned}
                \tilde h_a&=\sft(\ell_a) h_a~,
                \\
                \tilde \ell_{ab}&=((h_ah_{ab})^{-1}\acton\ell_a^{-1})\ell_{ab}((g_{ab}\acton h_{b}^{-1})\acton g_{ab}\acton \ell_b)~,
            \end{aligned}
        \end{equation}
        and
        \begin{equation}
            \begin{aligned}
                \tilde \Xi_a&=\Xi_a-h_a^{-1}\acton(\ell_a^{-1}\nabla_a \ell_a)~.
            \end{aligned}
        \end{equation}
    \end{subequations}

    \subsection{Comments on the above constructions}
    
    We start with an obvious but important remark:
    \begin{remark}
        For crossed modules, it is known that there are examples which do not admit an adjustment; see for example the case of $\sfT\sfB_n^{\rm F2}$ discussed in~\cite{Kim:2022opr}. It is therefore clear that there are 2-crossed modules for which there is no adjustment. As a consequence, not every principal 3-bundle can be endowed with a connection.
    \end{remark}
    
    Next, consider the differential cohomology, but with all maps $\kappa_{1,2,3}$ set to zero. In this case, the Chevalley--Eilenberg differential of the BRST Lie 3-algebroid does not close anymore; physicists speak of an ``open'' BRST complex, in which the differential squares to terms proportional to the {\em fake curvatures}\footnote{More generally, the fake curvature forms are all curvature forms except for the top form component.} $F_a$ and $H_a$. This is the well-known fake-flatness constraint imposed in many papers, see~\cite{Saemann:2019dsl} and~\cite{Borsten:2024gox} for a discussion. 
    
    In the BRST Lie 3-groupoid, we do not have the expected globularity of 2-morphisms: two 1-morphisms linked by a 2-morphism may have different targets, unless one imposes fake flatness. This constraint, however, is unsuitable for many applications. For example, many of the string structures such as the explicit example constructed in~\cite{Rist:2022hci}, are not fake flat. Also, the differential refinement of the description of topological T-duality in terms of principal 2-bundles only works if this condition is lifted~\cite{Kim:2022opr}. More generally, we expect a 3-bundle version of the result of~\cite{Gastel:2018joi,Saemann:2019dsl}, which shows that fake-flat connections on principal 2-bundles can locally be gauged to connections on abelian gerbes. This is too restrictive, and we stress that the analog of this statement is not true for principal bundles.
    
    \section{Examples}\label{sec:examples}
    
    In this section, we list three classes of examples of adjusted higher connections: one for local adjusted connections and two for global adjusted connections. All of them arise naturally in supergravity or string/M-theory.
    
    \subsection{Infinitesimal case: firm adjustments and \texorpdfstring{$d=4$}{d=4} gauged supergravity}
    
    A construction of adjustments for a particular class of $L_\infty$-algebras was given in~\cite{Borsten:2021ljb}. Here, one starts from a differential graded Lie algebra (dgLa) $\frh=(\oplus_{n\in \IZ}\frh_n,\sfd,[-,-]_\frh)$, which induces an $L_\infty$-algebra structure on the truncated and shifted complex 
    \begin{equation}
        \frL=\bigoplus_{k\in \IN} \frL_{-k}~,~~~\frL_{-k}\coloneqq\frg_{-k-1}
    \end{equation} 
    by a derived bracket construction~\cite{Fiorenza:0601312,Getzler:1010.5859}. As shown in~\cite{Borsten:2021ljb}, this construction can be refined to an $E_2L_\infty$-algebra, in which anti-symmetry is homotopy-lifted up to alternators. The alternators now provide an adjustment datum for the $L_\infty$-algebra $\frL$, and such adjustments were called firm in~\cite{Borsten:2021ljb}. 
    
    Concretely, start with a dgLA $\frh$ concentrated in degrees $d\geq -3$ with underlying differential complex 
    \begin{equation}
        \frh=\Big(~\frh_{-3}\xrightarrow{~\sfd~}\frh_{-2}\xrightarrow{~\sfd}\hdots~\Big)~,
    \end{equation} 
    then the above construction yields an $E_2L_\infty$-algebra with underlying complex
    \begin{equation}
        \frE=\Big(~\frL_{-2}=\frh_{-3}\xrightarrow{~\sfd~}\frL_{-1}=\frh_{-2}\xrightarrow{~\sfd~}\frL_0=\frh_{-1}~\Big)~.
    \end{equation} 
    The binary bracket in this $E_2L_\infty$-algebra reads as
    \begin{equation}
        \eps_2^0(\ell_1,\ell_2)\coloneqq [\sfd \ell_1,\ell_2]_\frh~,
    \end{equation} 
    and $(\frE,\eps_2^0)$ is a graded Leibniz algebra. The alternator providing the homotopy that describes the failure of anti-symmetry of $\eps_2^0$, is given by
    \begin{equation}
        \eps_2^1(\ell_1,\ell_2)\coloneqq [\ell_1,\ell_2]~.
    \end{equation} 
    This $E_2L_\infty$-algebra can then be anti-symmetrized to an $L_\infty$-algebra $\frL$ with the same underlying differential complex. As shown in~\cite{Borsten:2021ljb}, the components of $\eps_2^1(-,-)$ can then be identified with the adjustment data $\kappa_{1,2,3,4}$ for $\frL$.
    
    Even more concretely, this construction is implicit in the tensor hierarchies of gauged supergravity theories, and the Lie 3-algebra case corresponds to four-dimensional such theories, see the general discussion in~\cite{Borsten:2021ljb} and the original physics literature~\cite{deWit:2007mt,Bergshoeff:2009ph} and~\cite{deWit:2005ub,deWit:2008ta}.
    
    Supergravity theories come with a global symmetry group $\sfG$, and all fields form representations under this group. In particular, the field content contains local $\fru(1)$-valued connection 1-forms $A$ which are arranged in a representation $\frh_{-1}$ of the Lie algebra $\frh_0=\sfLie(\sfG)$ of $\sfG$. We now want to ``gauge'' a subgroup $\sfK$ of $\sfG$, i.e., to promote the $\fru(1)$-valued connection forms to $\sfLie(\sfK)$-valued connection forms. This is done by introducing a linear map $\sfd:\frh_{-1}\rightarrow \frh_0$, which is used to introduce an action of the 1-forms $A$ on all other fields $\phi$ of the theory by
    \begin{equation}
        A\acton \phi\coloneqq [\sfd A,\phi]~.
    \end{equation} 
    It is then found that the corresponding natural curvature expression does not transform as expected, which can be corrected by introducing a 2-form potential taking values in another representation $\frh_{-2}$ of $\frh_0$, and this process can be iterated $d-1$ times in $d$ dimensions. The end result is a differential graded Lie algebra $\frh$, and the curvatures are precisely the firmly adjusted curvatures arising from the procedure outlined above.
    
    To be even more concrete, let us sketch the case of maximal $d=4$ supergravity. In this case,  we have $\frh_0=\fre_{7(7)}$ (a particular real form of $\fre_7$), and the relevant representations are\footnote{We use physics notation, where the bold numbers label irreducible representations by their dimension; $\mathbf{56}$ denotes the fundamental representation and $\mathbf{133}$ denotes the adjoint.}
    \begin{equation}
        \frh_{-1}=\mathbf{56}~,~~~\frh_{-2}=\mathbf{133}^*~,~~~\frh_{-3}=\mathbf{912}~.
    \end{equation} 
    The embedding $\fre_{7(7)}\subset \frsp(56,\IR)$ allows one to pull back the invariant symplectic form to produce a (non-invariant) pairing $\Omega$ on $\frh_{-1}$, or a map $\frh_{-1}\rightarrow \frh_{-1}^*$. Together with the differential $\sfd:\frh_{-1}\rightarrow \frh_0$ and the representation $\rho$, we then have the following maps:
    \begin{equation}
        \begin{aligned}
            \sfd&:\mathbf{56}\rightarrow \mathbf{133}
            ~,~~~&
            \Omega&:\mathbf{56}\rightarrow \mathbf{56}^*
            ~,~~~&
            Z&=\Omega^* \circ \sfd^*: \mathbf{133}^*\rightarrow \mathbf{56}~,
            \\
            [-,-]&:\mathbf{133}\times\mathbf{133}\rightarrow \mathbf{133}
            ~,~~~&
            \rho_1&:\mathbf{133}\times\mathbf{56}\rightarrow \mathbf{56}
            ~,~~~&
            \rho_2&:\mathbf{133}\times\mathbf{912}\rightarrow \mathbf{912}~.
        \end{aligned}
    \end{equation}
    We note that $\sfd$ transforms in the $\mathbf{56}\times \mathbf{133}^*=\mathbf{56}\oplus \mathbf{912}\oplus\mathbf{6480}$, but consistency conditions restrict it to have non-trivial components only in the irreducible subrepresentation $\mathbf{912}$, which we now regard as embedded into the representation $\mathbf{56}\times \mathbf{133}^*$. This gives us another map
    \begin{equation}
        Y=[\sfd-,-]-\rho_1(-,Z(-)):\mathbf{912}\subset \mathbf{56}\times \mathbf{133}^*\rightarrow \mathbf{133}^*~.
    \end{equation} 
    With the appropriate constraints, this yields an $L_\infty$-algebra
    \begin{equation}
        \begin{aligned}
            \frL&=\Big(~\frL_{-2}\xrightarrow{~\mu_1~}\frL_{-1}\xrightarrow{~\mu_1~}\frL_{0}~\Big)
            \\
            &=\Big(~\frh_{-3}\xrightarrow{~Y~}\frh_{-2}\xrightarrow{~Z~}\frh_{-1}~\Big)
        \end{aligned}
    \end{equation}
    with higher products obtained as outlined above, e.g.
    \begin{equation}
        \mu_2(x_1,x_2)=\tfrac12([\sfd x_1,x_2]_{\frh_0}-[\sfd x_2,x_1]_{\frh_0})~.
    \end{equation} 
    
    Up to prefactors, the adjustment data is then identified as follows:
    \begin{equation}
        \begin{aligned}
            \kappa_1&=\rho^*_1(\Omega^*(-),-):\mathbf{56}\times \mathbf{56}\rightarrow \mathbf{133}^*~,
            \\
            \kappa_2&=\sfp:\mathbf{56}\times \mathbf{133}^*\rightarrow \mathbf{912}^*~,
            \\
            \kappa_3&=\sfp:\mathbf{133}^*\times \mathbf{56}\rightarrow \mathbf{912}^*~,
        \end{aligned}
    \end{equation}    
    where $\sfp$ is simply the projection from $\mathbf{56}\times \mathbf{133}^*$ to the irreducible representation $\mathbf{912}$.
    
    We conclude:
    \begin{proposition}
        The above adjustment data forms an adjustment of a 3-term $L_\infty$-algebra. In particular, the maps $\kappa_{1,2,3}$ satisfy the relations~\eqref{eq:kappa-axioms-infinitesimal}.
    \end{proposition}
    
    \subsection{Twisted string structures}
    
    A second example is the cocycle description of twisted string structures as introduced in~\cite{Sati:2009ic,Sati:2010dc}, see also~\cite{Fiorenza:2012mr}, which can be seen as a trivialization of the first fractional Pontryagin class in a 4-form background.
    
    Given a Lie group $\sfG$, we can construct a central extension of its based loop group\footnote{The group of smoothly parameterized paths $\ell:[0,1]\rightarrow \sfG$ that start and end at $\unit$, i.e., $\ell(0)=\ell(1)=\unit$.},
    \begin{equation}
        \begin{tikzcd}
            \unit \arrow[r] &\sfU(1) \arrow[r, "\iota"]& \widehat{L_0 \sfG}\arrow[r, "\pi"]& L_0 \sfG \arrow[r]& \unit
        \end{tikzcd}~,
    \end{equation}
    where the image of $\iota$ is contained in the center of $\widehat{L_0 \sfG}$. 
    
    We can combine $\widehat{L_0 \sfG}$ with the group $P_0\sfG$ of based paths in $\sfG$ to obtain a crossed module of Lie groups~\cite{Baez:2005sn}
    \begin{equation}
        \sfString(\sfG)=\left(\widehat{L_0 \sfG}\xrightarrow{~\sft~}P_0\sfG,\acton\right)~,
    \end{equation} 
    which forms a 2-crossed module of the string group, see also~\cite{Rist:2022hci} for further details.
    
    It has been shown in~\cite{Rist:2022hci} that an adjustment $\kappa^{\sfString}$ on $\sfString(G)$ is given by
    \begin{equation}
        \begin{aligned}
            \kappa^{\sfString}: P_0 G \times P_0 \mathfrak{g}&\rightarrow L_0 \mathfrak{g}\times \mathfrak{u}(1), \\
            (g ,V ) &\mapsto \left((\sfid-\wp\circ\flat)(gVg^{-1}-V), \frac{\rmi}
            {2\pi} \int_0^1 \rmd r \left \langle g^{-1} \frac{\partial g}{\partial r}, V \right \rangle\right)~,
        \end{aligned}
    \end{equation}
    where $\frg=\sfLie(\sfG)$, $\flat$ denotes the endpoint evaluation $\flat:P_0\sfG\rightarrow \sfG$, and $\wp$ is any smooth function $\wp:[0,1]\rightarrow \IR$ with $\wp(0)=0$ and $\wp(1)=1$.
    
    In order to describe twisted string structures, we need to lift the crossed module $\sfString(\sfG)$ to the 2-crossed module
    \begin{equation}
        \widehat{\sfString(\sfG)}\coloneqq \left(
        \begin{tikzcd}
            \sfU(1) \arrow[r, "\sft=\iota"] &\widehat{L_0\sfG}\arrow[r,"\sft=\flat\circ\pi"]&P_0 \sfG 
        \end{tikzcd}
        \right)~.
    \end{equation}
    Because the action of $P_0\sfG$ on $\widehat{L_0\sfG}$ satisfies the Peiffer identity, the Peiffer lifting in this crossed module can be chosen to be trivial. Together with axiom~\eqref{eq:2xm-last_axiom}, this forces us to choose the action of $P_0\sfG$ on $\sfU(1)$ to be trivial. This is compatible with the remaining axioms, as $\sfU(1)$ is abelian and the image of $\iota$ is located in the fiber of $\widehat{L_0\sfG}$ over $\unit\in L_0\sfG$.
    
    \begin{proposition}
        The 2-crossed module $\widehat{\sfString(\sfG)}$ becomes an adjusted 2-crossed module by setting 
        \begin{equation}
            \kappa_1=\kappa^{\sfString}~,~~~\kappa_2=0~,~~~\kappa_3=0~.
        \end{equation} 
    \end{proposition}
    \begin{proof}
        Using that $\{-,-\}$ is trivial, it is straightforward to verify that the conditions~\eqref{eq:final_adjustment_relations} are satisfied.
    \end{proof}
    
    While we do not have a proof of the uniqueness of adjustments for 2-crossed modules, it is clear that in this case, there is little freedom in choosing the adjustment. In particular, the conditions 
    \begin{equation}
        \begin{aligned}
            \kappa_3(t(\ell), X) &= \rmAd^\acton_\ell(X)=0~,
            \\
            \kappa_2(\sft(h),Y)&=-\sft(h)\acton\left(\{Y,h^{-1}\}+\{h^{-1},Y\}\right)+\kappa_3(h^{-1},\sft(Y))=0
        \end{aligned}
    \end{equation}
    force both maps to be trivial if the left argument is a loop.
    
    Consider now a differential cocycle describing a principal 3-bundle with structure 3-group $\widehat{\sfString(\sfG)}$. The connection forms in such a cocycle satisfy the Bianchi identity
    \begin{equation}
        \nabla_a H_a-\kappa_1(A_a,\sft(H_a))+\kappa_1(F_a,F_a)=G_a~,
    \end{equation}
    and this describes the trivialization of the 4-form $G_a-\kappa_1(F_a,F_a)$ by a (non-abelian) gerbe. This becomes clearer if we switch to the equivalent skeletal 3-term $L_\infty$-algebra\footnote{Recall that while global differential cocycles are very hard to describe explicitly in this case, the local data is unproblematic.}
    \begin{equation}
        \begin{tikzcd}
            \mathbb{R} \arrow[r, "\sfid"]& \mathbb{R}  \arrow[r,"0"]& \mathfrak{g}
        \end{tikzcd}~.
    \end{equation}
    Here, the above Bianchi identity reads as
    \begin{equation}
        \rmd H_a=G_a-(F_a,F_a)~,
    \end{equation} 
    and the gerbe described by $H_a$ trivializes the 2-gerbe with 4-form curvature $G_a-(F_a,F_a)$.
    
    A motivation for our example is the smooth moduli 3-stack of Spin connections and C-field configurations defined in ~\cite{Fiorenza:2012mr}, see also~\cite{Sati:2009ic}. On a space-time $X$ this is described by the data of a spin principal bundle $P_{\sfSpin}$ on $X$ with a $\frso$-connection $\nabla_{\frso}$, a principal $\sfE_8$ bundle $P_{\sfE_8}$ on $X$ and an abelian 2-gerbe with connection $(P_{\sfB^2\sfU(1)}, \nabla_{\sfB^2\sfU(1)})$. Furthermore, one requires an equivalence of abelian 2-gerbes between $P_{\sfB^2\sfU(1)}$ and $P_{\sfSpin}\times_XP_{\sfE_8}$
    via $\frac{1}{2}p_1 + 2a$, where $p_1$ is the first Pontryagin class of the spin bundle and $a$ is the second Chern class of $P_{\sfE_8}$.
    
    This reproduces the Witten flux quantization condition
    \begin{align}
        2[G_4] = \tfrac{1}{2}p_1 + 2a \in H^4(X, \mathbb{Z})
    \end{align}
    as in~\cite{Witten:1996md,Horava:1996ma} at the level of cohomology, and this condition is differentially refined to
    \begin{align}
        2G_4 = 2\rmd H + \tfrac{1}{2}p_1 + 2a~,
    \end{align}
    where $H$ is the characteristic class of a 1-gerbe realizing the equivalence of $U(1)$-2-gerbes between $P_{\sfB^2\sfU(1)}$ and $P_{\sfSpin}\times_XP_{\sfE_8}$.

    \subsection{Categorified torus-bundles}
    
    Finally, let us come to a finite example that we hope to be relevant to the description of $U$-duality. In this example, all adjustment functions are turned on.
    
    In the description of T-duality in terms of principal 2-bundles~\cite{Nikolaus:2018qop,Kim:2022opr,Kim:2023hqx}, gerbes on top of principal bundles were described by the categorified tori of~\cite{Ganter:2014zoa}. These tori are crossed modules of Lie groups that naturally come with an adjustment datum. We generalize this construction to 2-crossed modules.
    
    \begin{proposition}
        Consider three real vector spaces $V_i$ with prescribed lattices $\Lambda_i \subset V_i$, $i=0,\,1,\,2$, together with two bilinear maps
        \begin{equation}
            \langle -,-\rangle_0: \Lambda_0 \otimes \Lambda_0 \rightarrow \Lambda_1
            \eand
            \langle -,-\rangle_1: \Lambda_0 \otimes \Lambda_1 \rightarrow \Lambda_2~.
        \end{equation}
        Linearly extend both maps to all of the $V_i$. Then there is a 2-crossed module of Lie groups
        \begin{equation}
            \caT=(\Lambda_1\times V_2/\Lambda_2\xrightarrow{~\sft~}\Lambda_0\times V_1\xrightarrow{~\sft~}V_0)
        \end{equation} 
        with structure maps 
        \begin{equation}
            \begin{aligned}
                \sft(\lambda_0,u_1) &= \lambda_0~,\\
                \sft(\lambda_1,[w]) &= (0,\lambda_1)~,\\
                u_0 \triangleright (\lambda_0,u_1) &= (\lambda_0, u_1 + \langle u_0, \lambda_0 \rangle_0)~, \\
                u_0 \triangleright (\lambda_1,[w]) &= (\lambda_1, [w + \langle u_0, \lambda_1 \rangle_1])~,\\
                \{(\lambda_0,u_1),(\mu_0,v_1)\} &= (-\langle \lambda_0,\mu_0 \rangle_0,[ \langle \lambda_0,v_1 \rangle_1])
            \end{aligned}
        \end{equation}
        for all $u_{0,1}\in V_0$, $\lambda_0, \mu_0\in \Lambda_0$, $\lambda_1\in \Lambda_1$, and $w\in V_2$ if and only if
        \begin{equation}\label{eq:jacobi}
            \langle u_0, \langle v_0, w_0 \rangle_0 \rangle_1 + \langle v_0, \langle u_0, w_0 \rangle_0 \rangle_1 = 0~.
        \end{equation}
        We call this 2-crossed module the categorified torus for the data $(V_{0,1,2},\Lambda_{0,1,2},\langle-,-\rangle_{0,1})$.
    \end{proposition}
    \begin{proof}
        This is a straightforward verification of the axioms of a 2-crossed module of Lie groups, as listed in \ref{def:2xm-groups}.
    \end{proof}
    
    As a Lie $3$-group, $\mathcal{T}$ sits in a central extension of the form
    \begin{equation}
        1 \rightarrow B^2(V_2/\Lambda_2) \rightarrow \mathcal{T} \rightarrow \mathcal{T}_{\leq 1} \rightarrow 1~,
    \end{equation}
    where $\mathcal{T}_{\leq 1}$ is a Lie $2$-group sitting in a central extension of the form
    \begin{equation}
        1 \rightarrow B(V_1/\Lambda_1) \rightarrow \mathcal{T} \rightarrow V_0/\Lambda_0 \rightarrow 1
    \end{equation}
    and constructed from the data $(\Lambda_0,\Lambda_1,\langle -,-\rangle_0)$ as in~\cite{Ganter:2014zoa}. As we discuss in~\cite{Gagliardo:2025b}, $\mathcal{T}$ has a square root as a central extension of $\mathcal{T}_{\leq 1}$ but, in order to present this square root, a slightly weaker model for Lie $3$-groups is needed, which is beyond the scope of this paper. 
    
    \begin{proposition}
        If the two maps $\langle-,-\rangle_{0,1}$ satisfy the condition
        \begin{equation}\label{eq:torus_condition}
            [\langle u, \langle v, \lambda \rangle_0 \rangle_1] + [\langle \lambda, \langle u, v \rangle_0 \rangle_1] = 0
        \end{equation}
        for all $\lambda\in \Lambda_0$ and $u,v\in V_0$
        then an adjustment on $\caT$ is given by the maps 
        \begin{equation}
            \begin{aligned}
                \kappa_1&:& V_0 \times V_0 &\rightarrow V_1
                ~,~~~&
                (u_0,v_0) &\mapsto -\langle v_0,u_0 \rangle_0
                ~,
                \\
                \kappa_2&:& V_0 \times V_1 &\rightarrow V_2/\Lambda_2
                ~,~~~&
                (u_0,u_1) &\mapsto [\langle u_0,u_1 \rangle_1]
                ~,
                \\
                \kappa_3&:& (\Lambda_0\times V_1) \times V_0 &\rightarrow V_2/\Lambda_2
                ~,~~~&
                ((\lambda_0,u_1),u_0) &\mapsto [\langle u_0,u_1 \rangle_1]~.         
            \end{aligned}
        \end{equation}
    \end{proposition}
    \begin{proof}
        This is a straightforward verification of the adjustment conditions~\eqref{eq:final_adjustment_relations}.
    \end{proof}
    
    We note that condition~\eqref{eq:torus_condition} effectively amounts to $\langle u,\langle v,\lambda\rangle_0\rangle_1=0$. This is another hint that for an interesting categorified torus, one should consider a slightly weaker model for Lie 3-groups.
    
    Moreover, it is natural to expect that lifts of T-duality to M-theory can be described analogously to~\cite{Nikolaus:2018qop,Kim:2022opr} in terms of the higher categorified tori introduced above. This will be the topic of future work~\cite{Gagliardo:2025b}.

    \section*{Acknowledgments}
    
    We would like to thank Leron Borsten, Hyungrok Kim, Dominik Rist, and Martin Wolf for interesting discussions.
    
    %
    %
    %
    
    \appendices
    
    In this appendix, we summarize the required background on Lie 3-groups and Lie 3-algebras, starting with a short overview.
    
    \subsection{Tricategories, Lie 3-groups and Lie 3-algebras}\label{app:Lie-3-algebras}
    
    A tricategory is a collection of objects, 1-morphisms between objects, 2-morphisms between 1-morphisms and 3-morphisms between 2-morphisms. These can be composed in different ways, satisfying a number of coherence relations which can be weaker or stronger depending on the model. Contrary to the case of bicategories, the coherence theorem for (set-theoretical) tricategories does not state that a tricategory, in the weakest possible model, is always equivalent to a fully strict tricategory. Instead, the equivalence is with something slightly weaker, called a Gray tricategory \cite{Gordon:1995tr}. This will be defined in \ref{app:Gray}.
    
    A 3-groupoid is a tricategory in which the compositions come with weak inverses for 1- and 2-morphisms, as well as strict inverses for 3-morphisms. A 3-group is a 3-groupoid with a single object. Working with general tricategories is quite cumbersome, and so it is common to consider only Gray 3-groups; i.e., 3-groups which are also Gray tricategories. A convenient description of these is provided by 2-truncated simplicial groups: simplicial sets in which the spaces of simplices have additionally the structure of a group, in which all face and degeneracy maps are group homomorphisms, and in which all the spaces of simplices are determined by the spaces of 2-cells, 1-cells and 0-cells. There is much degeneracy in these, and extracting the essential information by considering the corresponding Moore complex, one arrives at 2-truncated hypercrossed complexes or, equivalently, 2-crossed modules of groups ~\cite{Conduche:1984:155,Carrasco:1991:195-235,Conduche:2003}. 
    
    Analogously, one can consider 2-truncated simplicial Lie groups, determining 2-crossed modules of Lie groups, and 2-truncated simplicial Lie algebras, determining 2-crossed modules of Lie algebras. These are models for a rather general family of Lie 3-groups and their Lie 3-algebras, although one must be warned that not all Lie 3-groups are covered by this construction due to the failure of the coherence theorem in the smooth setting. We will review the definition of 2-crossed modules in \ref{app:2crossedmodules}.
    
    For the description of Lie 3-algebras, we also use 3-term $L_\infty$-algebras, generalizations of dg-Lie algebras with ternary and quaternary brackets. We provide the necessary background on this form of Lie 3-algebras in \ref{app:L-infty}.
    
    \subsection{Gray tricategories}\label{app:Gray}

    Standard references for tricategories and their relation with 2-crossed modules include \cite{Gordon:1995tr,Kamps:2002aa,Gurski:2006}.

    \begin{adefinition}[Gray tricategory]\label{def:gray}
        A Gray tricategory $\scB$ consists of the following data. 
        \begin{enumerate}
            \item A set of objects $\scB_0$.
            \item For objects $A_1, \, A_2 \in \scB_0$, a strict bicategory $\scB_1(A_1,A_2)$. That is: 
            \begin{enumerate}
                \item $\scB_1(A_1,A_2)$ has 2-cells of the form 
                \begin{equation}
                    \begin{tikzcd}[ampersand replacement = \&, column sep=12ex ]
                        g' \& g\phantom{'} \ar[l,bend right = 40, "{h}"{name=F},,swap,pos=0.5] \ar[l,bend left = 40, "{h'}"{name=G},pos=0.5] \ar[Rightarrow,from=F,to=G,"{l}",swap,pos=0.4,, shorten=1ex] 
                    \end{tikzcd},
                \end{equation}
                where $g$ and $g'$ are called 1-morphisms between $A_1$ and $A_2$ in $\scB$, $h$ and $h'$ are called 2-morphisms between $g$ and $g'$ in $\scB$ and $l$ is called a 3-morphism between $h$ and $h'$ in $\scB$.
                \item There is a rule for vertical composition of 2-cells in $\scB_1(A_1,A_2)$,
                \begin{equation}
                    \begin{tikzcd}[ampersand replacement = \& ,column sep = 12ex]
                        g'
                        \& g\phantom{'} \ar[l,bend right = 70, "{h}"{name=F},swap] \ar[l,bend right = 10, "{h'}"{name=G},pos=0.5]
                        \ar[l,bend left = 70,"{h''}"{name=H}]
                        \ar[Rightarrow,from=F,to=G,"{l}",pos=0.4,shorten = 1ex]
                        \ar[Rightarrow,from=G,to=H,"{l'}",pos=0.4,shorten=0.5ex]
                    \end{tikzcd} = 
                    \begin{tikzcd}[ampersand replacement = \& ,column sep = 12ex]
                        g'
                        \& g\phantom{'}  \ar[l,bend right = 70, "{h}"{name=F},swap]
                        \ar[l,bend left = 70,"{h''}"{name=H}]
                        \ar[Rightarrow,from=F,to=H,"{l' \circ l}",swap,pos=0.4, shorten=1ex]
                    \end{tikzcd}~,
                \end{equation}
                which defines the transversal composition of 3-morphisms in $\scB$.
                \item There is a rule for horizontal composition of 2-cells in $\scB_1(A_1,A_2)$,
                \begin{align}
                    &
                    \begin{tikzcd}[ampersand replacement = \& , column sep=8ex]
                        g'' \& g' \ar[l,bend right = 40, "{h_1}"{name=F},pos=0.55,swap] \ar[l,bend left = 40, "{h'_1}"{name=G},pos=0.55] \ar[Rightarrow,from=F,to=G,"{l_1}",swap, shorten=1ex]  \& g\phantom{'} \ar[l,bend right = 40, "{h_2}"{name=F2},swap] \ar[l,bend left = 40, "{h'_2}"{name=G2}] \ar[Rightarrow,from=F2,to=G2,"{l_2}",swap, shorten=1ex] 
                    \end{tikzcd} =
                    \begin{tikzcd}[ampersand replacement = \&,column sep=16ex ]
                        g'' \& g\phantom{''} \ar[l,bend right = 40, "{h_1 \circ h_2}"{name=F},swap] \ar[l,bend left = 40, "{h_1' \circ h_2'}"{name=G}] \ar[Rightarrow,from=F,to=G,"{l_1 \bullet l_2}",swap,pos=0.45, shorten=1ex]
                    \end{tikzcd}~,
                \end{align}
                which defines vertical composition of 3-morphisms in $\scB$.
                \item Horizontal and vertical compositions of 2-cells in $\scB_1(A_1,A_2)$ are strictly associative and satisfy the following interchange rule,
                \begin{equation}\label{eq:interchange_rule}
                    \begin{tikzcd}[ampersand replacement = \& ,column sep = 18ex]
                        g''
                        \& g\phantom{''} \ar[l,bend right = 70, "{h_1 \circ h_2}"{name=F},swap] \ar[l,bend right = 10, "{h'_1 \circ h'_2}"{name=G},pos=0.5]
                        \ar[l,bend left = 70,"{h''_1 \circ h_2''}"{name=H}]
                        \ar[Rightarrow,from=F,to=G,"{l_1 \bullet l_2}",swap,pos=0.5,shorten=0.5ex]
                        \ar[Rightarrow,from=G,to=H,"{l'_1 \bullet l'_2}",swap,pos=0.4,shorten=0.5ex]
                    \end{tikzcd} = 
                    \begin{tikzcd}[ampersand replacement = \& ,column sep = 12ex]
                        g''
                        \& g'\phantom{'} \ar[l,bend right = 70, "{h_1}"{name=F},pos=0.5,swap]
                        \ar[l,bend left = 70,"{h''_1}"{name=H}]
                        \ar[Rightarrow,from=F,to=H,"{l'_1 \circ l_1}",swap,pos=0.4,shorten=0.5ex]
                        \& g\phantom{''} \ar[l,bend right = 70, "{h_2}"{name=F},swap]
                        \ar[l,bend left = 70,"{h''_2}"{name=H}]
                        \ar[Rightarrow,from=F,to=H,"{l'_2 \circ l_2}",swap,pos=0.4,shorten=0.5ex]
                    \end{tikzcd}.
                \end{equation}
            \end{enumerate}
            \item For objects $A_1, \,A_2,\,A_3 \in \scB_0$ and 1-morphisms $g_1 \in \scB_1(A_1,A_2)$, $g_2 \in \scB_1(A_2,A_3)$, strict functors $R_{g_1}: \scB_1(A_2,A_3) \rightarrow \scB_1(A_1,A_3)$ and $L_{g_2}: \scB_1(A_1,A_2) \rightarrow \scB_1(A_1,A_3)$ (i.e.~maps sending 2-cells to 2-cells, preserving all source, target and composition maps). These are called left and right whiskering functors, and they must satisfy the following.
            \begin{enumerate}
                \item $R_{g_1}(g_2) = L_{g_2}(g_1) =: g_1g_2 \in \scB_1(A_1,A_3)$.
                \item Given an additional $A_4 \in \scB_0$ and $g_3 \in \scB_1(A_3,A_4)$,
                \begin{equation}\label{eq:whiskering_compatibility}
                    \begin{aligned}
                        R_{g_1} \circ R_{g_2} &= R_{g_1g_2}:& \scB_1(A_3,A_4) &\rightarrow \scB_1(A_1,A_4)~,
                        \\
                        R_{g_1} \circ L_{g_3} &= L_{g_3} \circ R_{g_1}: &\scB_1(A_2,A_3) &\rightarrow \scB_1(A_1,A_4)~,
                        \\
                        L_{g_3} \circ L_{g_2} &= L_{g_2g_3}: &\scB_1(A_1,A_2) &\rightarrow \scB_1(A_1,A_4)~.
                    \end{aligned}
                \end{equation}
            \end{enumerate}
            \item For objects $A_1, \,A_2,\,A_3 \in \scB_0$, 1-morphisms $g_1,\,g_1' \in \scB_1(A_1,A_2)$, $g_2,\,g_2' \in \scB_1(A_2,A_3)$ and 2-morphisms $h_1:g_1 \rightarrow g_1'$, $h_2:g_2 \rightarrow g_2'$, a 3-morphism of the form
            \begin{equation}
                \begin{tikzcd}[column sep=18ex]
                    g_1g_2 \ar[r,"{L_{g_2}(h_1)}",{name=U}] \ar[d,"{R_{g_1}(h_2)}",swap] 
                    & g_1'g_2 \ar[d,"{R_{g_1'}(h_2)}"] \ar[Rightarrow, dl,"{I(h_1,h_2)}",swap]\\
                    g_1g_2' \ar[r,"{L_{g_2'}(h_1)}",{name=D},swap] 
                    & g_1'g_2'
                \end{tikzcd}.
            \end{equation}
            This 3-morphism is called interchangor and it must satisfy the following: 
            \begin{enumerate}
                \item Naturality on $3$-morphisms of $\scB$. That is, given $l_i: h_i \Rightarrow h_i': g_i \rightarrow g_i'$, $i = 1, \,2$, we have
                \begin{equation}
                    \scalebox{0.85}{
                    \begin{tikzcd}[ampersand replacement = \& ,column sep=12ex]
                        g_1g_2 \ar[r,"{L_{g_2}(h_1)}",{name=U}] \ar[d,"{}"{name=E}] \ar[d,bend right=90,looseness = 2.5, "{R_{g_1}(h_2')}"{name=F},swap]  \ar[Rightarrow,from=E,to=F,"{R_{g_1}(l_2)}",swap,shorten=0.5ex]
                        \& g_1'g_2 \ar[d,"{R_{g_1'}(h_2)}"] \ar[Rightarrow, dl,"{I(h_1,h_2)}",swap]
                        \\
                        g_1g_2' \ar[r,"{}"{name=D}] \ar[r,bend right=50, "{L_{g_2'}(h_1')}"{name=G},swap, looseness=1.5] \ar[Rightarrow,from=D,to=G,"{L_{g_2'}(l_1)}",pos=0.4,shorten=0.5ex]
                        \& g_1'g_2' 
                    \end{tikzcd}}
                    =\!
                    \scalebox{0.85}{
                    \begin{tikzcd}[ampersand replacement = \& ,column sep=12ex]
                        g_1g_2\phantom{'} \ar[r,bend left=50, "{L_{g_2}(h_1)}"{name=F}] \ar[r,"{}"{name=U}] \ar[d,"{R_{g_1}(h_2')}"{name=E},swap]   \ar[Rightarrow,from=F,to=U,"{L_{g_2}(l_1)}",shorten <=0.2ex,shorten >=-0.3ex,pos=0.7]
                        \& g_1'g_2\phantom{'} \ar[d,"{}"{name=R}] \ar[Rightarrow, dl,"{I(h_1',h_2')}",swap]   \ar[d,bend left=100, "{R_{g_1'}(h_2)}"{name=G},looseness=2.5] \ar[Rightarrow,from=G,to=R,"{R_{g_1'}(l_2)}",swap,shorten=0.5ex]\\
                        g_1g_2' \ar[r,"{L_{g_2'}(h_1')}"{name=D},swap] 
                        \& g_1'g_2' 
                    \end{tikzcd}
                    }~.
                \end{equation}
                \item They respect composition of $2$-morphisms in $\scB$. That is, given $g_i \stackrel{h_i}{\rightarrow} g_i' \stackrel{h_i'}{\rightarrow} g_i''$, $i=1, \,2$, we have
                \begin{equation}
                    \scalebox{0.85}{
                    \begin{tikzcd}[ampersand replacement = \& ,column sep=12ex]
                        g_1g_2 \ar[r,"{L_{g_2}(h_1)}",{name=U}] \ar[d,"{R_{g_1}(h_2)}",swap] 
                        \& g_1'g_2  \ar[d,"{}",swap] \ar[Rightarrow, dl,"{I(h_1,h_2)}",swap] \ar[r,"{L_{g_2}(h_1')}",{name=U2}] \& g_1''g_2  \ar[d,"{R_{g_1''}(h_2)}"] \ar[Rightarrow, dl,"{I(h_1',h_2)}",swap] \\
                        g_1g_2' \ar[r,"{L_{g_2'}(h_1)}",{name=D},swap] 
                        \& g_1'g_2' \ar[r,"{L_{g_2'}(h_1')}",{name=D2},swap]  \& g_1''g_2'
                    \end{tikzcd}}
                    = 
                    \scalebox{0.85}{
                    \begin{tikzcd}[ampersand replacement = \& ,column sep=15ex]
                        g_1g_2 \ar[r,"{L_{g_2}(h_1' \circ h_1)}"{name=A}] \ar[d,"{R_{g_1}(h_2)}",swap] 
                        \& g_1''g_2  \ar[d,"{R_{g_1''}(h_2)}"] \ar[Rightarrow, dl,"{I(h_1' \circ h_1,h_2)}",swap] \\
                        g_1g_2' \ar[r,"{L_{g_2'}(h_1' \circ h_1)}",{name=D},swap,pos=0.3]
                        \& g_1''g_2'
                    \end{tikzcd}}~,     
                \end{equation}
                \begin{equation}
                    \begin{tikzcd}[column sep=14ex]
                        g_1g_2 \ar[r,"{L_{g_2}(h_1)}",{name=U}] \ar[d,"{R_{g_1}(h_2)}",swap] 
                        & g_1'g_2  \ar[d,"{R_{g_1'}(h_2)}"] \ar[Rightarrow, dl,"{I(h_1,h_2)}",swap] \\
                        g_1g_2' \ar[r,"{}",{name=D}] \ar[d,"{R_{g_1}(h_2')}",swap]
                        & g_1'g_2' \ar[d,"{R_{g_1'}(h_2')}"] \ar[Rightarrow, dl,"{I(h_1,h_2')}",swap] \\
                        g_1g_2'' \ar[r,"{L_{g_2''}(h_1)}",{name=2},swap] 
                        & g_1'g_2''
                    \end{tikzcd} = 
                    \begin{tikzcd}[column sep=14ex,row sep = 10ex]
                        g_1g_2 \ar[r,"{L_{g_2}(h_1)}"{name=A}] \ar[d,"{R_{g_1}(h_2' \circ h_2)}"{name=X},swap] 
                        & g_1'g_2  \ar[d,"{R_{g_1'}(h_2' \circ h_2)}"] \ar[Rightarrow, dl,"{I(h_1,h_2' \circ h_2)}",swap] \\
                        g_1g_2'' \ar[r,"{L_{g_2''}(h_1)}",{name=D},swap]
                        & g_1'g_2''
                    \end{tikzcd} .  
                \end{equation}
                \item They respect whiskering of $2$-morphisms in $\scB$. That is, given $h_i:g_i \rightarrow g_i'$, $i=1, \,2,\, 3$, we have
                \begin{equation}
                    \!\!\!\scalebox{0.85}{
                    \begin{tikzcd}[ampersand replacement = \& ,column sep=13ex, row sep = 8ex]
                        g_1g_2g_3 \ar[r,"{L_{g_3}(L_{g_2}(h_1))}",{name=U}] \ar[d,"{L_{g_3}(R_{g_1}(h_2))}"{name=E},swap] 
                        \& g_1'g_2g_3 \ar[d,"{L_{g_3}(R_{g_1'}(h_2))}"] \ar[Rightarrow, dl,"{L_{g_3}(I(h_1,h_2))}",swap,pos=0.35]\\
                        g_1g_2'g_3 \ar[r,"{L_{g_3}(L_{g_2'}(h_1))}"{name=D},swap]
                        \& g_1'g_2'g_3 
                    \end{tikzcd}}
                    = 
                    \scalebox{0.85}{
                    \begin{tikzcd}[ampersand replacement = \& ,column sep=13ex, row sep = 8ex]
                        g_1g_2g_3 \ar[r,"{L_{g_2g_3}(h_1)}",{name=U}] \ar[d,"{R_{g_1}(L_{g_3}(h_2))}"{name=E},swap] 
                        \& g_1'g_2g_3 \ar[d,"{R_{g_1'}(L_{g_3}(h_2))}"] \ar[Rightarrow, dl,"{I(h_1,L_{g_3}(h_2))}",swap,pos=0.35]\\
                        g_1g_2'g_3 \ar[r,"{L_{g_2'g_3}(h_1)}"{name=D},swap]
                        \& g_1'g_2'g_3 
                    \end{tikzcd}}~,
                \end{equation}
                \begin{equation}
                    \scalebox{0.85}{
                    \begin{tikzcd}[ampersand replacement = \& ,column sep=14ex, row sep = 8ex]
                        g_1g_2g_3 \ar[r,"{L_{g_3}(L_{g_2}(h_1))}",{name=U}] \ar[d,"{R_{g_1g_2}(h_3)}"{name=E},swap] 
                        \& g_1'g_2g_3 \ar[d,"{R_{g_1'g_2}(h_3)}"] \ar[Rightarrow, dl,"{I(L_{g_2}(h_1),h_3)}",swap,pos=0.35]\\
                        g_1g_2g_3' \ar[r,"{L_{g_3'}(L_{g_2}(h_1))}"{name=D},swap]
                        \& g_1'g_2g_3' 
                    \end{tikzcd}}
                    = 
                    \scalebox{0.85}{
                    \begin{tikzcd}[ampersand replacement = \& ,column sep=14ex, row sep = 8ex]
                        g_1g_2g_3 \ar[r,"{L_{g_2g_3}(h_1)}",{name=U}] \ar[d,"{R_{g_1}(R_{g_2}(h_3))}"{name=E},swap] 
                        \& g_1'g_2g_3 \ar[d,"{R_{g_1'}(R_{g_2}(h_3))}"] \ar[Rightarrow, dl,"{I(h_1,R_{g_2}(h_3))}",swap,pos=0.35]\\
                        g_1g_2g_3' \ar[r,"{L_{g_2g_3'}(h_1)}"{name=D},swap]
                        \& g_1'g_2g_3' 
                    \end{tikzcd}}~,
                \end{equation}
                \begin{equation}
                    \scalebox{0.85}{
                    \begin{tikzcd}[ampersand replacement = \& ,column sep=14ex, row sep = 8ex]
                        g_1g_2g_3 \ar[r,"{L_{g_3}(R_{g_1}(h_2))}",{name=U}] \ar[d,"{R_{g_1g_2}(h_3)}"{name=E},swap] 
                        \& g_1g_2'g_3 \ar[d,"{R_{g_1g_2'}(h_3)}"] \ar[Rightarrow, dl,"{I(R_{g_1}(h_2),h_3)}",swap,pos=0.35]\\
                        g_1g_2g_3' \ar[r,"{L_{g_3'}(R_{g_1}(h_2))}"{name=D},swap]
                        \& g_1g_2'g_3' 
                    \end{tikzcd}}
                    = 
                    \scalebox{0.85}{
                    \begin{tikzcd}[ampersand replacement = \& ,column sep=14ex, row sep = 8ex]
                        g_1g_2g_3 \ar[r,"{R_{g_1}(L_{g_3}(h_2))}",{name=U}] \ar[d,"{R_{g_1}(R_{g_2}(h_3))}"{name=E},swap] 
                        \& g_1g_2'g_3 \ar[d,"{R_{g_1}(R_{g_2'}(h_3))}"] \ar[Rightarrow, dl,"{R_{g_1}(I(h_2,h_3))}",swap,pos=0.35]\\
                        g_1g_2g_3' \ar[r,"{R_{g_1}(L_{g_3'}(h_2))}"{name=D},swap]
                        \& g_1g_2'g_3' 
                    \end{tikzcd}}~.
                \end{equation}
            \end{enumerate}
        \end{enumerate}
    \end{adefinition}

    \subsection{2-crossed modules}\label{app:2crossedmodules}
    
    \begin{adefinition}[2-crossed module of groups]\label{def:2xm-groups}
        A 2-crossed module of groups $\caG$ consists of a complex of three groups
        \begin{equation}
            \begin{tikzcd}
                \sfL \arrow[r, "\sft"]& \sfH \arrow[r,"\sft"]& \sfG
            \end{tikzcd}
        \end{equation}
        together with actions by automorphisms
        \begin{equation}
            \acton:\sfG\times \sfH\rightarrow \sfH
            \eand 
            \acton:\sfG\times \sfL\rightarrow \sfL
        \end{equation} 
        with respect to which $\sft$ is equivariant, i.e.,
        \begin{equation}
            \sft(g\acton h)=g\sft(h)g^{-1}
            \eand 
            \sft(g\acton \ell)=g\acton \sft(\ell)
        \end{equation} 
        for all $g\in \sfG$, $h\in \sfH$, and $\ell\in \sfL$, and a $\sfG$-equivariant map
        \begin{equation}
            \{-,-\}: \sfH \times \sfH \longrightarrow \sfL~,
        \end{equation} 
        called the Peiffer lifting, that encodes the violation of the Peiffer identity, i.e.,
        \begin{subequations}
            \begin{equation}
                \langle h_1,h_2\rangle\coloneqq h_1 h_2 h_1^{-1} (\sft(h_1) \rhd h_2^{-1})= \sft(\{h_1, h_2\})
            \end{equation}
            for all $h_{1,2}\in \sfH$. The Peiffer lifting is required to satisfy the following axioms:
            \begin{align}
                \{\sft(\ell_1), \sft(\ell_2) \} &= \ell_1 \ell_2 \ell_1^{-1}\ell_2^{-1}~,
                \\
                \{h_1 h_2,h_3\}&=\{h_1,h_2h_3h_2^{-1}\}(\sft(h_1)\acton\{h_2,h_3\})~,
                \\
                \{h_1,h_2h_3\}&=\{h_1,h_2\}\{h_1,h_3\}\{\langle h_1,h_3\rangle^{-1},\sft(h_1)\acton h_2\}~,
                \\
                \{\sft(\ell),h\}\{h,\sft(\ell)\}&=\ell (\sft(h)\acton \ell^{-1})\label{eq:2xm-last_axiom}
            \end{align}
        \end{subequations}
        for all $h,h_{1,2,3}\in \sfH$ and $\ell\in \sfL$. We usually write $\caG=(\sfL\xrightarrow{\sft}\sfH\xrightarrow{\sft}\sfG)$ for short for such a 2-crossed module.
    \end{adefinition}
    
    Note that a crossed module of groups is simply a 2-crossed module $\caG=(\sfL\xrightarrow{\sft}\sfH\xrightarrow{\sft}\sfG)$ with $\sfL$ trivial. In particular, the data $\sfH/\sft(\sfL)\xrightarrow{\sft}\sfG$ forms a crossed module of groups. Moreover, the complex underlying the 2-crossed module is automatically normal, i.e., the images of $\sft$ form normal subgroups.
    
    In a 2-crossed module $\caG=(\sfL\xrightarrow{\sft}\sfH\xrightarrow{\sft}\sfG)$, we have a further natural action 
    \begin{equation}
        \begin{aligned}
            \acton:\sfH\times \sfL&\rightarrow \sfL~,
            \\
            h\acton \ell&\coloneqq \ell\{\sft(\ell)^{-1},h\}
        \end{aligned}
    \end{equation}
    by automorphisms. This action satisfies
    \begin{equation}
        \begin{aligned}
            \sft(h\acton \ell)&=\sft(\ell)\sft(\{\sft(\ell^{-1}),h\}) = \sft(\ell)\sft(\ell^{-1})h\sft(\ell)h^{-1} = h\sft(\ell)h^{-1}~,
            \\
            \sft(\ell_1)\acton\ell_2&=\ell_2\{\sft(\ell_2^{-1}),\sft(\ell_1)\} = \ell_1\ell_2\ell_1^{-1}~,
        \end{aligned}
    \end{equation}
    and therefore $(\sfH\xrightarrow{\sft}\sfG)$ forms a crossed module of groups.
    
    A 2-crossed module $\mathcal{G}$ determines a Gray tricategory denoted $B\mathcal{G}$. Its set of objects is $B\mathcal{G}_0 = \{\ast\}$. The 2-category $B\mathcal{G}(\ast,\ast)$ has 2-cells of the form
    \begin{equation}
        \begin{tikzcd}[ampersand replacement = \&, column sep=12ex ]
            \sft(h)g
            \& 
            \phantom{\sft(}g\phantom{h)}
            \ar[l,bend right = 40, "{h}"{name=F},swap] 
            \ar[l,bend left = 40, "{\sft(l)h}"{name=G}] \ar[Rightarrow,from=F,to=G,"{l}",swap,shorten=1ex] 
        \end{tikzcd}\!\!.
    \end{equation}
    Vertical composition of 2-cells in $B\mathcal{G}_1(\ast,\ast)$ is defined by
    \begin{equation}
        \begin{tikzcd}[ampersand replacement = \& ,column sep = 12ex]
            \sft(h)g
            \& 
            \phantom{\sft(}g\phantom{h)}
            \ar[l,bend right = 70, "{h}"{name=F},swap] 
            \ar[l,bend right = 10, "{\sft(l)h}"{name=G}]
            \ar[l,bend left = 70,"{\sft(l')\sft(l)h}"{name=H}]
            \ar[Rightarrow,from=F,to=G,"{l}",swap,pos=0.4,shorten = 1ex]
            \ar[Rightarrow,from=G,to=H,"{l'}",swap,pos=0.4,shorten=0.5ex]
        \end{tikzcd} = 
        \begin{tikzcd}[ampersand replacement = \& ,column sep = 12ex]
            \sft(h)g
            \& 
            \phantom{\sft()}g\phantom{h}
            \ar[l,bend right = 70, "{h}"{name=F},swap]
            \ar[l,bend left = 70,"{\sft(l'l)h}"{name=H}]
            \ar[Rightarrow,from=F,to=H,"{l'l}",swap,pos=0.3,shorten=1ex]
        \end{tikzcd}.
    \end{equation}
    Horizontal composition of 2-cells in $B\mathcal{G}_1(\ast,\ast)$ is defined by
    \begin{equation}
        \begin{tikzcd}[ampersand replacement = \& ]
            \sft(h')\sft(h)g
            \& 
            \sft(h)g 
            \ar[l,bend right = 40, "{h'}"{name=F2},swap] 
            \ar[l,bend left = 40, "{\sft(l')h'}"{name=G2}] \ar[Rightarrow,from=F2,to=G2,"{l'}",swap,shorten=1ex] 
            \& 
            \phantom{\sft()}g\phantom{h} 
            \ar[l,bend right = 40, "{h}"{name=F},swap] 
            \ar[l,bend left = 40, "{\sft(l)h}"{name=G}] \ar[Rightarrow,from=F,to=G,"{l}",swap,shorten=1ex] 
        \end{tikzcd} 
        =
        \begin{tikzcd}[ampersand replacement = \&,column sep=16ex ]
            \sft(h'h)g
            \& 
            g 
			\ar[l,bend right = 40, "{h'h}"{name=F},swap] 
			\ar[l,bend left = 40, "{\sft(l')h'\sft(l)h}"{name=G}] \ar[Rightarrow,from=F,to=G,"{l' (h' \triangleright l)}",swap,shorten=1ex] 
        \end{tikzcd} .
    \end{equation}
    Given $g_1,\,g_2 \in G$, the whiskering functors are defined as follows. 
    \begin{align}
        \begin{split}
            L_{g_2}: \mathcal{G}(\ast,\ast) &\rightarrow \mathcal{G}(\ast,\ast) \\
            \begin{tikzcd}[ampersand replacement = \& ]
            	\sft(h_1)g_1
                \& 
                \phantom{\sft}g_1\phantom{h()} 
                \ar[l,bend right = 40, "{h_1}"{name=F},swap] 
                \ar[l,bend left = 40, "{\sft(l_1)h_1}"{name=G}] \ar[Rightarrow,from=F,to=G,"{l_1}",swap,shorten=1ex] 
            \end{tikzcd} &\mapsto
            \begin{tikzcd}[ampersand replacement = \&, column sep = 10ex ]
                \sft(h_1)g_1g_2
                \& 
                \phantom{\sft()}g_1g_2\phantom{h_1} 
                \ar[l,bend right = 40, "{h_1}"{name=F},swap] 
                \ar[l,bend left = 40, "{\sft(l_1)h_1}"{name=G}] \ar[Rightarrow,from=F,to=G,"{ l_1}",swap,shorten=1ex] 
            \end{tikzcd}
        \end{split}.
    \end{align}
    \begin{align}
        \begin{split}
            R_{g_1}: B\mathcal{G}_1(\ast,\ast) &\rightarrow B\mathcal{G}_1(\ast,\ast) \\
            \begin{tikzcd}[ampersand replacement = \& ]
                \sft(h_2)g_2
                \& 
                \phantom{\sft()}g_2\phantom{h} 
                \ar[l,bend right = 50, "{h_2}"{name=F},swap] 
                \ar[l,bend left = 50, "{\sft(l_2)h_2}"{name=G}] \ar[Rightarrow,from=F,to=G,"{l_2}",swap,shorten=1ex] 
            \end{tikzcd} 
            &\mapsto
            \begin{tikzcd}[ampersand replacement = \& , column sep = 10ex]
                g_1\sft(h_2)g_2
                \& 
                \phantom{\sft}g_1g_2\phantom{h_2()} 
                \ar[l,bend right = 40, "{g_1 \triangleright h_2}"{name=F},swap] 
                \ar[l,bend left = 40, "{g_1 \triangleright (\sft(l_2)h_2)}"{name=G}] \ar[Rightarrow,from=F,to=G,"{g_1 \triangleright l_2}",swap,shorten=1ex] 
            \end{tikzcd}
        \end{split},
    \end{align}
    The interchangor reads as 
    \begin{equation}\label{eq:interchangor_2_xm}
        \begin{tikzcd}[column sep=20ex, row sep = 8ex]
            g_1g_2 \ar[r,"{h_1}",{name=U}] \ar[d,"{g_1 \triangleright h_2}",swap] 
            & \sft(h_1)g_1g_2 \ar[d,"{(\sft(h_1)g_1) \triangleright h_2}"] \ar[Rightarrow, dl,"{\{h_1,g_1 \triangleright h_2\}}",swap]\\
            g_1\sft(h_2)g_2 \ar[r,"{ h_1}",{name=D},swap] 
            & \sft(h_1)g_1\sft(h_2)g_2
        \end{tikzcd}~~.
    \end{equation}
    The specialization to 2-crossed modules of Lie groups is evident: we require that all maps are smooth. Applying the tangent functor, we then arrive at the definition of 2-crossed modules of Lie algebras.
    \begin{adefinition}[2-crossed module of Lie algebras]
        A 2-crossed module of Lie algebras is a complex of three Lie algebras
        \begin{equation}
            \begin{tikzcd}
                \mathfrak{l}\arrow[r, "\sft"]& \mathfrak{h}\arrow[r, "\sft"]& \mathfrak{g}~,
            \end{tikzcd}
        \end{equation}
        together with actions by derivations
        \begin{equation}
            \acton:\frg\times \frh\rightarrow \frh
            \eand 
            \acton:\frg\times \frl\rightarrow \frl~,
        \end{equation}  
        such that $\sft$ is a $\frg$-homomorphism, i.e.,
        \begin{equation}
            \sft(X\acton Y)=[X,\sft(Y)]
            \eand 
            \sft(X\acton Z)=X\acton \sft(Z)
        \end{equation} 
        for all $X\in \frg$, $Y\in \frh$, and $Z\in \frl$, and a $\frg$-equivariant bilinear map
        \begin{equation}
            \{-,-\}: \frh \times \frh \longrightarrow \frl~,
        \end{equation} 
        called the Peiffer lifting, that encodes the violation of the Peiffer identity, i.e.,
        \begin{subequations}
            \begin{equation}\label{eq:def_Peiffer_lifting}
                \langle Y_1,Y_2\rangle\coloneqq 
                [Y_1, Y_2] - \sft(Y_1) \acton Y_2= \sft(\{Y_1, Y_2\})
            \end{equation}
            for all $Y_{1,2}\in \frh$. The Peiffer lifting is required to satisfy the following axioms:
            \begin{align}
                \{\sft(Z_1),\sft(Z_2)\}&=[Z_1,Z_2]~,
                \\ 
                \{[Y_1,Y_2],Y_3\}&=\sft(Y_1)\acton\{Y_2,Y_3\}+\{Y_1,[Y_2,Y_3]\}\nonumber
                \\
                &\hspace{0.6cm}-\sft(Y_2)\acton\{Y_1,Y_3\}-\{Y_2,[Y_1,Y_3]\}~,
                \\
                \{Y_1,[Y_2,Y_3]\}&=\{\sft(\{Y_1,Y_2\}),Y_3\}-\{\sft(\{Y_1,Y_3\}),Y_2\}~,
                \\
                \{\sft(Z),Y\}+\{Y,\sft(Z)\}&=-\sft(Y)\acton Z
            \end{align}
            for all $Y,Y_{1,2}\in \frh$ and $Z,Z_{1,2}\in \frl$.
        \end{subequations}
    \end{adefinition}
    
    Again a crossed module of Lie algebras $(\frh\xrightarrow{\sft}\frg)$ is a 2-crossed module of Lie algebras $(\frl\xrightarrow{\sft}\frh\xrightarrow{\sft}\frg)$ with $\frl$ trivial.
    
    The analog of the $\sfH$-action on $\sfL$ in a 2-crossed module of Lie groups $\caG=(\sfL\xrightarrow{\sft}\sfH\xrightarrow{\sft}\sfG)$ is the action 
    \begin{equation}
        \begin{aligned}
            \acton:\frh\times \frl&\rightarrow \frl~,
            \\
            Y\acton Z&\coloneqq -\{\sft(Z),Y\}~,
        \end{aligned}
    \end{equation}
    and with this action $(\frl\xrightarrow{~\sft}\frh)$ is a crossed module of Lie algebras.
    
    It is clear that in a crossed module of Lie algebras $(\frl\xrightarrow{\sft}\frh\xrightarrow{\sft}\frg)$ obtained by differentiating a crossed module of Lie groups $(\sfL\xrightarrow{\sft}\sfH\xrightarrow{\sft}\sfG)$, the group $\sfG$ acts on the Lie algebras $\frl$, $\frh$, and $\frg$. We will denote this action by the same symbol $\acton$. Similarly, we have further ``half-linearized'' Peiffer liftings
    \begin{equation}\label{eq:half-linearizedPL}
        \{-,-\}:\sfH\times \frh\rightarrow \frl
        \eand 
        \{-,-\}:\frh\times \sfH\rightarrow \frl~,
    \end{equation} 
    and we list the half-linearized forms of the 2-crossed module axioms in the next section.
    
    \subsection{Helpful 2-crossed module relations}\label{app:Formulas}
    
    In a 2-crossed module of groups $\caG=(\sfL\xrightarrow{\sft}\sfH\xrightarrow{\sft}\sfG)$, we also have the following useful identities:
    \begin{align}
        \{\unit,h\}&=\{h,\unit\}=\unit~,
        \\
        \sft(\{h_1,h_2\}^{-1})&=\sft((h_1h_2h_1^{-1})\acton\{h_1,h_2^{-1}\})~,
        \\
        \ell_1\ell_2 &= \{\sft(\ell_1),\sft(\ell_2)\}\ell_2\ell_1~,
        \\
        h\acton\ell &= \sft(h)\acton(\ell\{h^{-1},h\sft(\ell^{-1})h^{-1}\})~,
        \\
        \{h_1,h_2\}\ell &=(\sft(\{h_1,h_2\})\acton \ell)\{h_1,h_2\}~,
        \\
        \sft(h_1)\acton h_2&=h_1\sft(h_1^{-1}\acton\{h_1,h_2\}^{-1})h_2h_1^{-1}
    \end{align}
    for all $h,h_{1,2}\in \sfH$ and $\ell,\ell_{1,2}\in \sfL$. For proofs, see e.g.~\cite{Martins:2009aa,Saemann:2013pca} and references therein.
    
    The half-linearized forms of the Peiffer lifting~\eqref{eq:half-linearizedPL} satisfy the following relations:
    \begin{equation}
        \begin{aligned}
            X\acton\{Y,h\}&=\{X\acton Y,h\}+\{Y, h^{-1} X\acton h\}-\{[Y,h^{-1} X\acton h]-\sft(Y)\acton (h^{-1} X\acton h),h\}~,
            \\
            X\acton\{h,Y\}&=\{h,[h^{-1}X\acton h,Y]\}+\sft(h)\acton\{h^{-1}X\acton h,Y\}+\{h,X\acton Y\}~,
            \\
            \sft(\{h,Y\})&=hYh^{-1}-\sft(h)\acton Y~,
            \\
            Y-\sft(h)\acton Y&=\sft(\{h,Y\}+\{Y,h\})-h(\sft(Y)\acton h^{-1})~,
            \\
            Z-\sft(h)\acton Z&=\{\sft(Z),h\}+\{h,\sft(Z)\}~,
        \end{aligned}
    \end{equation}
    \begin{equation}
        \begin{aligned}
            \{[Y_1,Y_2],h\}&=\{\sft(\{Y_1,Y_2\}),h\}-\{\sft(\{Y_1,h\}),Y_2\}-\{\sft(\{Y_2,Y_1\}),h\}+\{\sft(\{Y_2,h\}),Y_1\}
            \\
            &\hspace{1cm}+\sft(Y_1)\acton\{Y_2,h\}-\sft(Y_2)\acton\{Y_1,h\}~,
            \\
            \{h,[Y_1,Y_2]\}&=\{\sft(\{h,Y_1\}),\sft(h)\acton Y_2\}-\{\sft(\{h,Y_2\}),\sft(h)\acton Y_1\}
            +[\{h,Y_1\},\{h,Y_2\}]~,
            \\
            \{hY_1h^{-1},Y_2\}&=\{h,[Y_1,h^{-1}Y_2h]\}+\sft(h)\acton\{Y_1,h^{-1}Y_2 h\}+\sft(hY_1h^{-1})\acton\sft(h)\acton\{h^{-1},Y_2\}
            \\
            \{Y_1,hY_2h^{-1}\}&=\{Y_1,Y_2\}-\{\sft(\{Y_1,Y_2\},h\}-\{\sft(\{Y_1,h^{-1}\}),Y_2\}
            \\
            &\hspace{1cm}+\{\sft(\{\sft(\{Y_1,h^{-1}\}),Y_2\}),h\}~,
            \\
            \sft(h(X\acton h^{-1}))&=\sft(h)X\sft(h^{-1})-X=\sft(h)\acton X-X~,
            \\
            h(\sft(Y)\acton h^{-1})&=h Y h^{-1}+\sft(\{Y,h\})-Y~,
        \end{aligned}
    \end{equation}
    and
    \begin{equation}
        \begin{aligned}
            \{h_1h_2,Y\}&=\{h_1,h_2Yh_2^{-1}\}+\sft(h_1)\acton\{h_2,Y\}~,
            \\
            \{Y,h_1h_2\}&=\{Y,h_2\}+\{h_2Yh_2^{-1}-h_2\sft(Y)\acton h_2^{-1},h_1\}
        \end{aligned}
    \end{equation}
    for all $h\in \sfH$, $X\in \frg$, $Y,Y_{1,2}\in \frh$, and $Z\in \frl$. These relations are readily derived from linearizing the axioms in a 2-crossed module of Lie groups.
    
    Inverses inside the Peiffer lifting can be resolved as follows:
    \begin{equation}
        \begin{aligned}
            \{h_1^{-1},h_2\}&=(\sft(h_1^{-1})\acton\{h_1,h_1^{-1}h_2h_1\})^{-1}~,
            \\
            \{h_1,h_2^{-1}\}&=\{h_1,h_2\}^{-1}\{(\sft(h_1)\acton h_2^{-1})h_1h_2h_1^{-1},\sft(h_1)\acton h_2\}^{-1}~,
        \end{aligned}
    \end{equation}
    which linearize as
    \begin{equation}
        \begin{aligned}
            \{h^{-1},Y\}&=-\sft(h^{-1})\acton\{h,h^{-1}Yh\}~,
            \\
            \{Y,h^{-1}\}&=\{h^{-1}(\sft(Y)\acton h),h\}-\{h^{-1}Yh,h\}
        \end{aligned}
    \end{equation}
    for all $h,h_{1,2}\in \sfH$ and $Y\in \frh$.
    
    Finally, we can list a number of relations for 2-crossed-module-valued forms. Consider $U$ an open patch of some manifold. Then 
    \begin{equation}
        \begin{aligned}
            (g\acton h)^{-1}\rmd(g\acton h) &= g\acton[ h^{-1}((g^{-1}\rmd g)\acton h)]+g\acton( h^{-1}\rmd h)~,
            \\
            \rmd(\alpha\acton h) &= \rmd\alpha\acton h+(-1)^p\alpha\acton\rmd h~,
            \\
            \rmd(g\acton\beta) &= g\acton((g^{-1}\rmd g)\acton\beta)+g\acton\rmd\beta~,
            \\
            \rmd(h^{-1}(\alpha\acton h)) &= h^{-1}(\rmd\alpha\acton h)+(-1)^p[h^{-1}(\alpha\acton h),h^{-1}\rmd h]+(-1)^p\alpha\acton h^{-1}\rmd h 
        \end{aligned}
    \end{equation}
    for all $g\in\scC^\infty(U,\sfG)$, $h\in\scC^\infty(U,\sfH)$, $\alpha\in\Omega^p(U,\frg)$, and $\beta\in\Omega^q(U,\frh)$.
    
    The differential acts on the adjustment function $\kappa_1$ as follows: 
    \begin{equation}
        \begin{aligned}
            \rmd\kappa_1(g,X)&=g\acton \kappa_1(g^{-1}\rmd g,X)+\kappa(g,[g^{-1}\rmd g,X]-\sft(\kappa_1(g^{-1}\rmd g,X)))+\kappa_1(g,\rmd X)~.
        \end{aligned}
    \end{equation}    
    
    \subsection{3-term \texorpdfstring{$L_\infty$}{L infinity}-algebras}\label{app:L-infty}
    
    Another useful model for Lie 3-algebras exists in the form of 3-term $L_\infty$-algebras. Recall that $L_\infty$-algebras are the homotopy algebras of Lie algebras. They generalize differential graded Lie algebras in that the Jacobi identity only holds up to higher homotopy. These homotopies, in turn, satisfy further, higher Jacobi identities. For a review of $L_\infty$-algebras in our conventions, see~\cite{Jurco:2018sby}.
    
    Concretely, an $L_\infty$-algebra $\frL=(\frL,\mu_k)$ is a graded vector space $\frL=\oplus_{k\in \IZ}\frL_k$ endowed with multilinear maps $\mu_k:\frL^{\wedge k}\rightarrow \frL$ of degree~$|\mu_k|=2-k$. These brackets satisfy the homotopy Jacobi identities 
    \begin{equation}
        \sum_{j+k=i}\sum_{\sigma\in {\rm Sh}(j;i) }\chi(\sigma;\ell_1,\ldots,\ell_{i})(-1)^{k}\mu_{k+1}(\mu_j(\ell_{\sigma(1)},\ldots,\ell_{\sigma(j)}),\ell_{\sigma(j+1)},\ldots,\ell_{\sigma(i)}) = 0
    \end{equation}
    for all $i\in\IN$ and elements $\ell_1,\ldots,\ell_{i}$ of $\frL$. Here, ${\rm Sh}(j;i)$ denotes the set of $(j;i)$-unshuffles, i.e., permutations such that $\sigma(1)<\cdots<\sigma(j)$ and $\sigma(j+1)<\cdots<\sigma(i)$, and $\chi(\sigma;\ell_1,\ldots,\ell_{i})$ denotes the Koszul sign that arises in the graded permutations required to transform $\ell_1\wedge \ldots \wedge \ell_i$ to $\ell_{\sigma(1)}\wedge \ldots \wedge \ell_{\sigma(i)}$.
    
    Higher or categorified Lie algebras can be modeled by $L_\infty$-algebras concentrated in non-positive degrees. In particular, a general class of Lie $n$-algebras can be described by $n$-term $L_\infty$-algebras, which are $L_\infty$-algebras concentrated (i.e., non-trivial) in degrees $-n+1$ to $0$. 
    
    For Lie 3-algebras, we consider 3-term $L_\infty$-algebras of the form
    \begin{equation}
        \frL=\frL_{-2}\oplus \frL_{-1}\oplus \frL_0~.
    \end{equation} 
    For degree reasons, the only non-trivial higher products in a 3-term $L_\infty$-algebra are $\mu_k$ with $1\leq k\leq 4$. As usual in an $L_\infty$-algebra $(\frL,\mu_1)$ forms a differential graded vector space. 
    
    We call an $L_\infty$-algebra strict if $\mu_k$ is trivial for $k\geq 3$. While strict 2-term $L_\infty$-algebras agree with crossed modules of Lie algebras (i.e., 2-crossed modules of Lie algebras $(\frl\xrightarrow{~\sft~}\frh\xrightarrow{~\sft~}\frg)$ with $\frl$ trivial), strict 3-term $L_\infty$-algebras have an overlap with 2-crossed modules of Lie algebras, but neither contains the other. For details, see the discussion in~\cite[Appendix D]{Kim:2019owc}. 
    
    The key result here is that any  2-crossed module of Lie algebras maps to a strict 3-term $L_\infty$-algebra with the following identification:
    \begin{equation}\label{eq:ident_L-infty_2CM}
        \begin{aligned}
            \mu_1&: \frL_{-1}\rightarrow \frL_0
            ~~~&&\leftrightarrow~~~
            &\sft&: \frh\rightarrow \frg~,
            \\
            \mu_1&: \frL_{-2}\rightarrow \frL_{-1}
            ~~~&&\leftrightarrow~~~
            &\sft&: \frl\rightarrow \frh~,
            \\
            \mu_2&: \frL_0\times \frL_0\rightarrow \frL_0
            ~~~&&\leftrightarrow~~~
            &[-,-]&: \frg\times \frg\rightarrow \frg~,
            \\
            \mu_2&: \frL_0\times \frL_{-1}\rightarrow \frL_{-1}
            ~~~&&\leftrightarrow~~~
            &\acton&: \frg\times \frh\rightarrow \frh~,
            \\
            \mu_2&: \frL_0\times \frL_{-2}\rightarrow \frL_{-2}
            ~~~&&\leftrightarrow~~~
            &\acton&: \frg\times \frl\rightarrow \frl~,
            \\
            \mu_2&: \frL_{-1}\times \frL_{-1}\rightarrow \frL_{-2}
            ~~~&&\leftrightarrow~~~
            &-\{-,-\}-\{-,-\}\circ\sigma_2&: \frh\times \frh\rightarrow \frl~,
        \end{aligned}
    \end{equation}
    where $\sigma_2$ denotes the permutation $\sigma_2(Y_1,Y_2)=(Y_2,Y_1)$. Note that in this map, information about the anti-symmetric part of the Peiffer lifting is lost.
    
    \subsection{Raw adjustment conditions}\label{app:raw_adjustment_condition}
    
    Below, we list the explicit form of the raw adjustment conditions $\caC_{1,2,3,4}$.
    \begin{equation}
        \begin{aligned}
            \caC_1&=
            -\left\{\kappa_1\left(g_2^{-1},\nu_2\left(g_1^{-1},F\right)\right),g_2^{-1}\acton \Lambda_1\right\}
            -R_1\left(g_2,\Lambda_2,\nu_2\left(g_1^{-1},F\right)\right)
            \\&\hspace{1cm}
            +R_1\left(g_1g_2,g_2^{-1}\acton \Lambda_1,F\right)
            +R_1(g_1g_2,\Lambda_2,F)
            -\left\{g_2^{-1}\acton \Lambda_1,\kappa_1\left(g_2^{-1},\nu_2\left(g_1^{-1},F\right)\right)\right\}
            \\&\hspace{1cm}
            -g_2^{-1}\acton R_1(g_1,\Lambda_1,F)
            +g_2^{-1}\acton \kappa_2\left(g_1^{-1},H\right)
            -\kappa_2\left(g_2^{-1},[\kappa_1\left(g_1^{-1},F\right),\Lambda_1]\right)
            \\&\hspace{1cm}
            -\kappa_2\left(g_2^{-1},g^{-1}_1 \nabla g_1\acton \kappa_1\left(g_1^{-1},F\right)\right)
            +\kappa_2\left(g_2^{-1},g_1^{-1}\acton H\right)
            \\&\hspace{1cm}
            +\kappa_2\left(g_2^{-1},(g_1^{-1}Fg_1)\acton \Lambda_1\right)
            +\kappa_2\left(g_2^{-1},g_1^{-1}\acton \kappa_1(\rmAd_{g_1}(g^{-1}_1 \nabla g_1),F)\right)
            \\&\hspace{1cm}
            +\kappa_2\left(g_2^{-1},\sft\left(\left\{\kappa_1\left(g_1^{-1},F\right),\Lambda_1\right\}\right)\right)
            -\kappa_2\left(g_2^{-1},\sft(R_1(g_1,\Lambda_1,F))\right)
            \\&\hspace{1cm}
            +\kappa_2\left(g_2^{-1},\sft\left(\kappa_2\left(g_1^{-1},H\right)\right)\right)
            -\kappa_2\left(g_2^{-1},\kappa_1\left(g_1^{-1},\rmd F\right)\right)-\kappa_2\left(g_2^{-1},\kappa_1\left(g_1^{-1},\nu_2(A,F)\right)\right)
            \\&\hspace{1cm}
            +\kappa_2\left(g_2^{-1},\kappa_1\left(g_1^{-1},\nu_2(\rmAd_{g_1}(g^{-1}_1 \nabla g_1),F)\right)\right)
            -\kappa_2\left(g_2^{-1},\kappa_1\left(g^{-1}_1 \nabla g_1,\nu_2\left(g_1^{-1},F\right)\right)\right)
            \\&\hspace{1cm}
            +\kappa_2\left(g_2^{-1},\kappa_1\left(\sft(\Lambda_1),\nu_2\left(g_1^{-1},F\right)\right)\right)
            -\kappa_2\left((g_1g_2)^{-1},H\right)
        \end{aligned}
    \end{equation}
    
    \begin{equation}
        \begin{aligned}
            \caC_2&=-\left\{\kappa_1\left(g^{-1},F\right),g^{-1}\acton h^{-1}\nabla h\right\}-g^{-1}\acton \rmd(\kappa_3(h,F))+\kappa_2\left(g^{-1},H\right)-\kappa_2\left((\sft(h)g)^{-1},H\right)
            \\&\hspace{1cm}
            -R_1(g,\Lambda ,F)+R_1(\sft(h)g,\Lambda ,F)+R_1\left(\sft(h)g,g^{-1}\acton h^{-1}\nabla h,F\right)
            \\&\hspace{1cm}
            +R_1\left(\sft(h)g,\sft\left(g^{-1}\acton \Xi\right),F\right)+\left(\sft\left(\kappa_1\left(g^{-1},F\right)\right)\acton \left(g^{-1}\acton \Xi\right)\right)
            \\&\hspace{1cm}
            +g^{-1}\acton \Big(\left\{H,h^{-1}\right\}+\left\{h^{-1},H\right\}+\left\{h^{-1},\kappa_1(A,F)\right\}
            +\left\{\kappa_1(A,F),h^{-1}\right\}
            \\&\hspace{1cm}
            -\left\{h^{-1}\nabla h,\rmAd^\acton_{h^{-1}}(F)+g\acton \kappa_1\left(g^{-1},F\right)\right\}
            -A\acton \kappa_3(h,F)-\sft\left(h^{-1}\nabla h\right)\acton \kappa_3(h,F)\Big)
        \end{aligned}
    \end{equation}
    
    \begin{equation}
        \begin{aligned}
            \caC_3&=\left\{g^{-1}\acton h_1{}^{-1},g^{-1}\acton \rmAd^\acton_{h_2^{-1}}(F)\right\}+(\sft\left(h_1\right)g)^{-1}\acton \kappa_3\left(h_2,F\right)
            \\
            &\hspace{1cm}-g^{-1}\acton \kappa_3\left(h_2h_1,F\right)+g^{-1}\acton \kappa_3\left(h_1,F\right)
        \end{aligned}
    \end{equation}
    
    \begin{equation}
        \begin{aligned}
            \caC_4&=R_2(g,\sft(\ell)h,X)-R_2(g,h,X)-g^{-1}\acton h^{-1}\acton(\ell^{-1} X\acton \ell)
        \end{aligned}
    \end{equation}
    
    \bibliographystyle{latexeu2}
    
    \bibliography{bigone}

\newcommand{\etalchar}[1]{$^{#1}$}
\begin{thebibliography}{BJFJ{\etalchar{+}}25}

\bibitem[ACJ05]{Aschieri:2003mw}
P.~Aschieri, L.~Cantini, and B.~Jur\v{c}o,
{\em Nonabelian bundle gerbes, their differential geometry and gauge theory,}
\href{https://dx.doi.org/10.1007/s00220-004-1220-6}{Commun. Math. Phys. {\bf
  254}  (2005) 367} [{\tt
  \href{https://www.arxiv.org/abs/hep-th/0312154}{hep-th/0312154}}].

\bibitem[BEM04]{Bouwknegt:2003vb}
P.~Bouwknegt, J.~Evslin, and V.~Mathai,
{\em T-Duality: Topology change from H-flux,}
\href{https://dx.doi.org/10.1007/s00220-004-1115-6}{Commun. Math. Phys. {\bf
  249}  (2004) 383} [{\tt
  \href{https://www.arxiv.org/abs/hep-th/0306062}{hep-th/0306062}}].

\bibitem[BH11]{Baez:2010ya}
J.~C.~Baez and J.~Huerta,
{\em {An invitation to higher gauge theory},}
\href{https://dx.doi.org/10.1007/s10714-010-1070-9}{Gen. Relativ. Gravit. {\bf
  43}  (2011) 2335} [{\tt \href{https://www.arxiv.org/abs/1003.4485}{1003.4485
  [hep-th]}}].

\bibitem[BHH{\etalchar{+}}09]{Bergshoeff:2009ph}
E.~A.~Bergshoeff, J.~Hartong, O.~Hohm, M.~H{\" u}bscher, and T.~Ort{\'\i}n,
{\em {Gauge theories, duality relations and the tensor hierarchy},}
\href{https://dx.doi.org/10.1088/1126-6708/2009/04/123}{JHEP {\bf 0904}  (2009)
  123} [{\tt \href{https://www.arxiv.org/abs/0901.2054}{0901.2054 [hep-th]}}].

\bibitem[BJFJ{\etalchar{+}}25]{Borsten:2024gox}
L.~Borsten, M.~Jalali~Farahani, B.~Jur{\v{c}}o, H.~Kim,
  J.~N{\'a}ro{\v{z}}n{\'y}, D.~Rist, C.~Saemann, and M.~Wolf,
{\em {Higher Gauge Theory},}
\href{https://dx.doi.org/10.1016/B978-0-323-95703-8.00217-2}{Encyclopedia Math.
  Phys. (2025) 159} [{\tt
  \href{https://www.arxiv.org/abs/2401.05275}{2401.05275 [hep-th]}}].

\bibitem[BKS05]{Bojowald:0406445}
M.~Bojowald, A.~Kotov, and T.~Strobl,
{\em Lie algebroid morphisms, Poisson sigma models, and off-shell closed gauge
  symmetries,}
\href{https://dx.doi.org/10.1016/j.geomphys.2004.11.002}{J. Geom. Phys. {\bf
  54}  (2005) 400} [{\tt
  \href{https://www.arxiv.org/abs/math.DG/0406445}{math.DG/0406445}}].

\bibitem[BKS21]{Borsten:2021ljb}
L.~Borsten, H.~Kim, and C.~Saemann,
{\em {$E_2L_\infty$-algebras, Generalized Geometry, and Tensor Hierarchies},}
{\tt \href{https://www.arxiv.org/abs/2106.00108}{2106.00108 [hep-th]}}.

\bibitem[BM05]{Breen:math0106083}
L.~Breen and W.~Messing,
{\em Differential geometry of gerbes,}
\href{https://dx.doi.org/10.1016/j.aim.2005.06.014}{Adv. Math. {\bf 198}
  (2005) 732} [{\tt
  \href{https://www.arxiv.org/abs/math.AG/0106083}{math.AG/0106083}}].

\bibitem[Bre94]{Breen:1994aa}
L.~Breen,
{\em On the classification of 2-gerbes and 2-stacks,}
\href{http://www.numdam.org/issues/AST_1994__225__1_0/}{Ast{\'e}risque {\bf
  225}  (1994)}.

\bibitem[Bry93]{0817647309}
J.-L.~Brylinski,
{\em Loop spaces, characteristic classes and geometric quantization,}
  Birkh{\"a}user, Boston, 1993.

\bibitem[BSCS07]{Baez:2005sn}
J.~C.~Baez, D.~Stevenson, A.~S.~Crans, and U.~Schreiber,
{\em {From loop groups to 2-groups},}
\href{http://projecteuclid.org/euclid.hha/1201127333}{Homol. Homot. Appl. {\bf
  9}  (2007) 101} [{\tt
  \href{https://www.arxiv.org/abs/math.QA/0504123}{math.QA/0504123}}].

\bibitem[Car50a]{Cartan:1949aaa}
H.~Cartan,
{\em {Cohomologie r{\'e}elle d'un espace fibr{\'e} principal diff{\'e}rentiable
  I},}
\href{http://www.numdam.org/item/SHC_1949-1950__2__A18_0/}{Talk no.~19,
  {S{\'e}minaire Henri Cartan}. {\bf 2}  (1950)}.

\bibitem[Car50b]{Cartan:1949aab}
H.~Cartan,
{\em {Cohomologie r{\'e}elle d'un espace fibr{\'e} principal diff{\'e}rentiable
  II},}
\href{http://www.numdam.org/item/SHC_1949-1950__2__A19_0/}{Talk no.~20,
  {S{\'e}minaire Henri Cartan}. {\bf 2}  (1950)}.

\bibitem[CC91]{Carrasco:1991:195-235}
P.~Carrasco and A.~Cegarra,
{\em Group-theoretic algebraic models for homotopy types,}
\href{https://dx.doi.org/10.1016/0022-4049(91)90133-m}{J. Pure Appl. Alg. {\bf
  75}  (1991) 195}.

\bibitem[Con84]{Conduche:1984:155}
D.~Conduch{\'e},
{\em Modules crois{\'e}s g{\'e}n{\'e}ralis{\'e}s de longueur 2,}
\href{https://dx.doi.org/10.1016/0022-4049(84)90034-3}{J. Pure Appl. Algebra.
  {\bf 34}  (1984) 155}.

\bibitem[Con03]{Conduche:2003}
D.~Conduch{\'e},
{\em Simplicial crossed modules and mapping cones,}
\href{ftp://ftp.gwdg.de/ftp/pub/EMIS/journals/GMJ/vol10/v10n4-3.pdf.gz}{Georgian
  Math. J. {\bf 10}  (2003) 623}.

\bibitem[dWNS08]{deWit:2008ta}
B.~de~Wit, H.~Nicolai, and H.~Samtleben,
{\em {Gauged supergravities, tensor hierarchies, and M-theory},}
\href{https://dx.doi.org/10.1088/1126-6708/2008/02/044}{JHEP {\bf 0802}  (2008)
  044} [{\tt \href{https://www.arxiv.org/abs/0801.1294}{0801.1294 [hep-th]}}].

\bibitem[dWST05]{deWit:2005ub}
B.~de~Wit, H.~Samtleben, and M.~Trigiante,
{\em Magnetic charges in local field theory,}
\href{https://dx.doi.org/10.1088/1126-6708/2005/09/016}{JHEP {\bf 0509}  (2005)
  016} [{\tt \href{https://www.arxiv.org/abs/hep-th/0507289}{hep-th/0507289}}].

\bibitem[dWST07]{deWit:2007mt}
B.~de~Wit, H.~Samtleben, and M.~Trigiante,
{\em The maximal D=4 supergravities,}
JHEP {\bf 0706}  (2007) 049 [{\tt
  \href{https://www.arxiv.org/abs/0705.2101}{0705.2101 [hep-th]}}].

\bibitem[FJFKS24]{Fischer:2024vak}
S.-R.~Fischer, M.~Jalali~Farahani, H.~Kim, and C.~Saemann,
{\em {Adjusted connections I: Differential cocycles for principal groupoid
  bundles with connection},}
{\tt \href{https://www.arxiv.org/abs/2406.16755}{2406.16755 [math.DG]}}.

\bibitem[FM07]{Fiorenza:0601312}
D.~Fiorenza and M.~Manetti,
{\em $L_\infty$ structures on mapping cones,}
\href{https://dx.doi.org/10.2140/ant.2007.1.301}{Alg. Numb. Th. {\bf 1}  (2007)
  301} [{\tt
  \href{https://www.arxiv.org/abs/math.QA/0601312}{math.QA/0601312}}].

\bibitem[FSS15]{Fiorenza:2012mr}
D.~Fiorenza, H.~Sati, and U.~Schreiber,
{\em {The $E_8$ moduli 3-stack of the C-field in M-theory},}
\href{https://dx.doi.org/10.1007/s00220-014-2228-1}{Commun. Math. Phys. {\bf
  333}  (2015) 117} [{\tt \href{https://www.arxiv.org/abs/1202.2455}{1202.2455
  [hep-th]}}].

\bibitem[Gan18]{Ganter:2014zoa}
N.~Ganter,
{\em Categorical tori,}
\href{https://dx.doi.org/10.3842/SIGMA.2018.014}{SIGMA {\bf 14}  (2018) 014}
  [{\tt \href{https://www.arxiv.org/abs/1406.7046}{1406.7046 [math.RT]}}].

\bibitem[Gas19]{Gastel:2018joi}
A.~Gastel,
{\em {Canonical gauges in higher gauge theory},}
\href{https://dx.doi.org/10.1007/s00220-019-03530-4}{Commun. Math. Phys. {\bf
  376}  (2019) 1053} [{\tt
  \href{https://www.arxiv.org/abs/1810.06278}{1810.06278 [math-ph]}}].

\bibitem[Get10]{Getzler:1010.5859}
E.~Getzler,
{\em Higher derived brackets,}
{\tt \href{https://www.arxiv.org/abs/1010.5859}{1010.5859 [math-ph]}}.

\bibitem[Gir71]{Giraud:1971}
J.~Giraud,
{\em Cohomologie non ab\'elienne,} Springer, Berlin, 1971
[\href{https://dx.doi.org/10.1007/978-3-662-62103-5}{doi}].

\bibitem[GPS95]{Gordon:1995tr}
R.~Gordon, A.~J.~Power, and R.~Street,
{\em Coherence for tricategories,}
\href{https://dx.doi.org/10.1090/memo/0558}{Memoirs AMS {\bf 558}  (1995)}.

\bibitem[GSTD]{Gagliardo:2025b}
G.~Gagliardo, C.~Saemann, and R.~Tellez-Dominguez,
{\em Geometric $T_2$-duality,}
in preparation.

\bibitem[Gur06]{Gurski:2006}
N.~Gurski,
{\em An algebraic theory of tricategories,} PhD thesis, Chicago (2006).

\bibitem[HW96]{Horava:1996ma}
P.~Horava and E.~Witten,
{\em Eleven-dimensional supergravity on a manifold with boundary,}
\href{https://dx.doi.org/10.1016/0550-3213(96)00308-2}{Nucl. Phys. B {\bf 475}
  (1996)~94} [{\tt
  \href{https://www.arxiv.org/abs/hep-th/9603142}{hep-th/9603142}}].

\bibitem[JRSW19]{Jurco:2018sby}
B.~Jur\v{c}o, L.~Raspollini, C.~Saemann, and M.~Wolf,
{\em {$L_\infty$-algebras of classical field theories and the
  {B}atalin--{V}ilkovisky formalism},}
\href{https://dx.doi.org/10.1002/prop.201900025}{Fortsch. Phys. {\bf 67}
  (2019) 1900025} [{\tt \href{https://www.arxiv.org/abs/1809.09899}{1809.09899
  [hep-th]}}].

\bibitem[JSW16]{Jurco:2016qwv}
B.~Jur\v{c}o, C.~Saemann, and M.~Wolf,
{\em {Higher groupoid bundles, higher spaces, and self-dual tensor field
  equations},}
\href{https://dx.doi.org/10.1002/prop.201600031}{Fortschr. Phys. {\bf 64}
  (2016) 674} [{\tt \href{https://www.arxiv.org/abs/1604.01639}{1604.01639
  [hep-th]}}].

\bibitem[Jur11]{Jurco:2009px}
B.~Jur{\v c}o,
{\em {Nonabelian bundle 2-gerbes},}
\href{https://dx.doi.org/10.1142/S0219887811004963}{Int. J. Geom. Meth. Mod.
  Phys. {\bf 08}  (2011)~49} [{\tt
  \href{https://www.arxiv.org/abs/0911.1552}{0911.1552 [math.DG]}}].

\bibitem[KP02]{Kamps:2002aa}
K.~H.~Kamps and T.~Porter,
{\em 2-groupoid enrichments in homotopy theory and algebra,}
\href{https://dx.doi.org/10.1023/A:1016051407785}{K-Theory {\bf 25}  (2002)
  373}.

\bibitem[KS15]{Kotov:2007nr}
A.~Kotov and T.~Strobl,
{\em {Characteristic classes associated to $Q$-bundles},}
\href{https://dx.doi.org/10.1142/S0219887815500061}{Int. J. Geom. Meth. Mod.
  Phys. {\bf 12}  (2015) 1550006} [{\tt
  \href{https://www.arxiv.org/abs/0711.4106}{0711.4106 [math.DG]}}].

\bibitem[KS20]{Kim:2019owc}
H.~Kim and C.~Saemann,
{\em Adjusted parallel transport for higher gauge theories,}
\href{https://dx.doi.org/10.1088/1751-8121/ab8ef2}{J. Phys. A {\bf 52}  (2020)
  445206} [{\tt \href{https://www.arxiv.org/abs/1911.06390}{1911.06390
  [hep-th]}}].

\bibitem[KS22]{Kim:2022opr}
H.~Kim and C.~Saemann,
{\em Non-geometric T-duality as higher groupoid bundles with connections,}
{\tt \href{https://www.arxiv.org/abs/2204.01783}{2204.01783 [hep-th]}}.

\bibitem[KS23]{Kim:2023hqx}
H.~Kim and C.~Saemann,
{\em T-duality as correspondences of categorified principal bundles with
  adjusted connections,}
\href{https://dx.doi.org/10.22323/1.436.0336}{PoS {\bf CORFU2022}  (2023) 336}
  [{\tt \href{https://www.arxiv.org/abs/2303.16162}{2303.16162 [hep-th]}}].

\bibitem[Li14]{Li:2014}
D.~Li,
{\em Higher groupoid actions, bibundles and differentiation,} PhD thesis,
  Universit{\"a}t G{\"o}ttingen (2014)
[{\tt \href{https://www.arxiv.org/abs/1512.04209}{1512.04209 [math.DG]}}].

\bibitem[May93]{May:book:1967}
J.~P.~May,
{\em Simplicial objects in algebraic topology,} University of Chicago Press,
  1993
[\href{https://dx.doi.org/10.1090/mmono/198/02}{doi}].

\bibitem[MP11]{Martins:2009aa}
J.~F.~Martins and R.~Picken,
{\em The fundamental {G}ray 3-groupoid of a smooth manifold and local
  3-dimensional holonomy based on a 2-crossed module,}
\href{https://dx.doi.org/10.1016/j.difgeo.2010.10.002}{Diff. Geom. App. {\bf
  29}  (2011) 179} [{\tt \href{https://www.arxiv.org/abs/0907.2566}{0907.2566
  [math.CT]}}].

\bibitem[NW13]{Nikolaus:2011ag}
T.~Nikolaus and K.~Waldorf,
{\em {Four equivalent versions of non-abelian gerbes},}
\href{https://dx.doi.org/10.2140/pjm.2013.264.355}{Pacific J. Math. {\bf 246}
  (2013) 355} [{\tt \href{https://www.arxiv.org/abs/1103.4815}{1103.4815
  [math.DG]}}].

\bibitem[NW19]{Nikolaus:2018qop}
T.~Nikolaus and K.~Waldorf,
{\em Higher geometry for non-geometric T-duals,}
\href{https://dx.doi.org/10.1007/s00220-019-03496-3}{Commun. Math. Phys. {\bf
  374}  (2019) 317} [{\tt
  \href{https://www.arxiv.org/abs/1804.00677}{1804.00677 [math.AT]}}].

\bibitem[RSW26]{Rist:2022hci}
D.~Rist, C.~Saemann, and M.~Wolf,
{\em Explicit non-Abelian gerbes with connections,}
\href{https://dx.doi.org/10.1088/1751-8121/ae2e60}{J. Phys. A {\bf 59}  (2026)
  035201} [{\tt \href{https://www.arxiv.org/abs/2203.00092}{2203.00092
  [hep-th]}}].

\bibitem[Sat11]{Sati:2010dc}
H.~Sati,
{\em {Geometric and topological structures related to M-branes II: Twisted
  string and string{${}^c$} structures},}
\href{https://dx.doi.org/10.1017/S1446788711001261}{J. Austral. Math. Soc. {\bf
  90}  (2011)~93} [{\tt \href{https://www.arxiv.org/abs/1007.5419}{1007.5419
  [hep-th]}}].

\bibitem[{\v S}ev06]{Severa:2006aa}
P.~{\v S}evera,
{\em $L_\infty$-algebras as 1-jets of simplicial manifolds (and a bit beyond),}
{\tt \href{https://www.arxiv.org/abs/math.DG/0612349}{math.DG/0612349}}.

\bibitem[SS20]{Saemann:2019dsl}
C.~Saemann and L.~Schmidt,
{\em {Towards an {M}5-brane model {II}: Metric string structures},}
\href{https://dx.doi.org/10.1002/prop.202000051}{Fortschr. Phys. {\bf 68}
  (2020) 2000051} [{\tt \href{https://www.arxiv.org/abs/1908.08086}{1908.08086
  [hep-th]}}].

\bibitem[SSS09]{Sati:2008eg}
H.~Sati, U.~Schreiber, and J.~Stasheff,
{\em $L_\infty$-algebra connections and applications to String- and
  Chern--Simons $n$-transport,}
in: ``Quantum Field Theory,'' eds. B. Fauser, J. Tolksdorf and E. Zeidler, p.
  303, Birkh{\"a}user 2009
[\href{https://dx.doi.org/10.1007/978-3-7643-8736-5_17}{doi}]
[{\tt \href{https://www.arxiv.org/abs/0801.3480}{0801.3480 [math.DG]}}].

\bibitem[SSS12]{Sati:2009ic}
H.~Sati, U.~Schreiber, and J.~Stasheff,
{\em {Twisted differential string and fivebrane structures},}
\href{https://dx.doi.org/10.1007/s00220-012-1510-3}{Commun. Math. Phys. {\bf
  315}  (2012) 169} [{\tt \href{https://www.arxiv.org/abs/0910.4001}{0910.4001
  [math.AT]}}].

\bibitem[SW14]{Saemann:2013pca}
C.~Saemann and M.~Wolf,
{\em {Six-dimensional superconformal field theories from principal 3-bundles
  over twistor space},}
\href{https://dx.doi.org/10.1007/s11005-014-0704-3}{Lett. Math. Phys. {\bf 104}
   (2014) 1147} [{\tt \href{https://www.arxiv.org/abs/1305.4870}{1305.4870
  [hep-th]}}].

\bibitem[SWY22]{Song:2021jnw}
D.~Song, K.~Wu, and J.~Yang,
{\em 3-form {Y}ang--{M}ills based on 2-crossed modules,}
\href{https://dx.doi.org/10.1016/j.geomphys.2022.104537}{J. Geometry and Phys.
  {\bf 178}  (2022) 104537} [{\tt
  \href{https://www.arxiv.org/abs/2108.12852}{2108.12852 [math-ph]}}].

\bibitem[TD23]{Tellez-Dominguez:2023wwr}
R.~Tellez-Dominguez,
{\em {C}hern correspondence for higher principal bundles,}
{\tt \href{https://www.arxiv.org/abs/2310.12738}{2310.12738 [math.DG]}}.

\bibitem[Wan14]{Wang:2013dwa}
W.~Wang,
{\em {On 3-gauge transformations, 3-curvatures, and Gray-categories},}
\href{https://dx.doi.org/10.1063/1.4870640}{J. Math. Phys. {\bf 55}  (2014)
  043506} [{\tt \href{https://www.arxiv.org/abs/1311.3796}{1311.3796
  [math-ph]}}].

\bibitem[Wit97]{Witten:1996md}
E.~Witten,
{\em {On flux quantization in M-theory and the effective action},}
\href{https://dx.doi.org/10.1016/S0393-0440(96)00042-3}{J. Geom. Phys. {\bf 22}
   (1997)~1} [{\tt
  \href{https://www.arxiv.org/abs/hep-th/9609122}{hep-th/9609122}}].

\end{thebibliography}
    
\end{document}